\documentclass{lmcs}

\usepackage[T1]{fontenc}
\usepackage{amssymb}
\usepackage{mathrsfs}
\usepackage{stmaryrd}
\usepackage{tikz-cd}
\usetikzlibrary{positioning,arrows.meta,decorations.pathreplacing,calc,patterns}
\usepackage{booktabs}

\definecolor{obsColor}{HTML}{2D5FAA}
\definecolor{conjColor}{HTML}{3D348B}
\definecolor{mainPosColor}{HTML}{6A4C93}
\definecolor{otherPosColor}{HTML}{9163CB}
\definecolor{negObsColor}{HTML}{C77DBA}
\definecolor{negsColor}{HTML}{E8638B}
\definecolor{namedFill}{HTML}{ffd0ed}
\definecolor{unnFill}{HTML}{e6e2f9}

\newcommand{\eqName}[1]{\mathrm{#1}}
\newcommand{\Msys}{\mathscr{M}}
\newcommand{\Nsys}{\mathscr{N}}
\newcommand{\HS}{\mathscr{H}}
\newcommand{\TC}{\mathscr{C}}
\newcommand{\DG}{\mathscr{G}}
\newcommand{\CT}{\mathscr{T}}
\newcommand{\SL}{\mathscr{L}}

\newcommand{\BG}{\mathscr{R}}
\newcommand{\SE}{\mathscr{E}}
\newcommand{\FK}{\mathscr{F}}
\newcommand{\PT}{\mathscr{P}}
\newcommand{\reach}{\to^{\!*}}
\newcommand{\step}{\Delta}
\newcommand{\reachpred}{\Gamma}
\newcommand{\init}{\mathsf{init}}

\newcommand{\patheq}{\sim}
\newcommand{\bismark}{\approx}
\newcommand{\bimark}{\rightleftarrows}
\newcommand{\Th}[1]{\mathbb{T}_{#1}}
\newcommand{\ET}[1]{\mathcal{E}[\Th{#1}]}

\newcommand{\pe}{\Theta}
\newcommand{\Hom}{\mathrm{Hom}}
\newcommand{\id}{\mathrm{id}}
\newcommand{\Sh}{\mathbf{Sh}}
\newcommand{\HML}{\mathit{HML}}
\newcommand{\ST}{\mathrm{ST}}

\makeatletter
\def\endfront@text{}
\def\enddoc@text{}
\def\@setamsclass{}
\def\@setkeywords{}
\def\@settitlecomment{}

\makeatother

\begin{document}

\title[A Classifying Topos for the Spectrum of Equivalences]{A Classifying Topos for the Spectrum of Equivalences}
\author{Kenan Oggad}
\address{Universit\'{e} Paris-Saclay}
\email{kenan.oggad@universite-paris-saclay.fr}

\begin{abstract}
What makes two computational systems equivalent?  Topos theory answers
with classifying toposes: a system's semantic content is encoded in the
geometric theory it classifies, and two presentations are equivalent
when their classifying toposes coincide.  Process algebra answers with
the linear time--branching time spectrum of van~Glabbeek: a hierarchy of
behavioral equivalences from trace equivalence to bisimilarity, each
determined by which observations can distinguish processes.  We show that these answers are aspects of a
single structure, and develop the foundations of a topos-geometric
theory of semantic equivalence in which behavioral abstraction is
localization.

Each labeled transition system~$\Msys$ receives a geometric
theory~$\Th{\Msys}$ whose classifying
topos~$\ET{\Msys}$ determines its provable geometric
sequents.  We establish a strict hierarchy where mutual simulation is strictly
coarser than bisimulation, which is strictly coarser than topos
equivalence, and identify the boundary:
diamond-only Hennessy--Milner logic characterizes the
bisimulation-invariant fragment of geometric logic.  This geometric analogue of van~Benthem's theorem is proved
in full generality via bounded tree unraveling.  Explicit Grothendieck topologies yield a strict chain
$J_{\mathrm{bisim}} \subsetneq J_{\mathrm{sim}} \subsetneq
J_{\mathrm{trace}}$, constructive for trace and bisimulation;
a constructive counterexample proves the
naive existence-based observation-class approach provably inadequate for
simulation, motivating Caramello's quotient-theory duality as the natural
alternative; an energy--topology framework extends this embedding to all
13~named equivalences of the spectrum.

Closing the named equivalences under lattice operations yields
exactly 30~elements, including 17~unnamed hybrids absent from the
classical spectrum because the energy-game framework computes but does
not close.  Lattice~$L_{30}$ is directly indecomposable;
its Heyting implication $S \to F = \mathrm{IF}$ identifies impossible
futures as the algebraic mediator of the simulation--failures divide.
A Geometric Closure Theorem computes presheaf Heyting implications by
testing at a single free extension.  All three levels of the hierarchy, the
$L_{30}$ bi-Heyting structure, and the Geometric Closure Theorem are
proved constructively, a result internal to the topos with
no known process-algebraic proof.  Viewed through the classifying topos,
the spectrum is a finite sub-poset of an infinite coframe
whose algebraic operations (meets, implications,
subtractions) produce structure inaccessible from within process
algebra.  All constructions are formalized in Lean~4 with Mathlib.
\end{abstract}

\maketitle

\section{Introduction}
\label{sec:intro}

Topos theory and process algebra offer independent frameworks for when
two computational systems should be considered equivalent: the former
through classifying toposes and their lattices of subtoposes, the latter
through behavioral observation hierarchies culminating in van~Glabbeek's
spectrum.  We show that these perspectives fit within a single
structure---the spectrum embeds as a finite sub-poset of the coframe of
subtoposes of the classifying topos---by constructing explicit
Grothendieck topologies for trace and bisimulation equivalence on a
classifying topos of labeled transition systems and completing the strict
chain $J_{\mathrm{bisim}} \subsetneq J_{\mathrm{sim}} \subsetneq
J_{\mathrm{trace}}$ via Caramello's quotient-theory duality.
Simulation requires this indirect route: trace and bisimulation
coverages operate at the sieve level ($0$-truncated), but simulation
requires only homomorphism existence ($(-1)$-truncated), a
truncation-level mismatch that \emph{provably prevents} naive
existence-based observation-class extensions
(Theorem~\ref{thm:naive-sim-instability}).

An energy--topology framework extends this three-point embedding to all
13~named process equivalences of the van~Glabbeek spectrum.
Assignment $E \mapsto J_E$ is antitone; the induced coframe map is
injective for $|L| \geq 2$ and reveals algebraic structure invisible to
the original spectrum: the coframe meet of two named-equivalence
subtoposes can produce elements with no process-algebraic name
(Theorem~\ref{thm:energy-lt-bridge}).  It connects
Bisping's energy-game hierarchy to Lawvere--Tierney topologies through a
chain: energy budget $\to$ observation class $\to$ covering predicate
$\to$ Grothendieck topology $\to$ subtopos, an instance of Caramello's
bridge technique~\cite{caramello2017theories}:

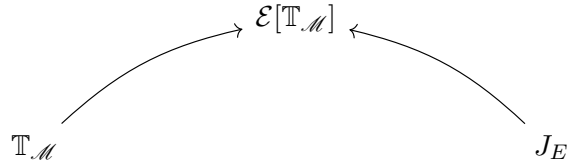
\begin{figure}[ht]
\centering
\begin{tikzcd}[column sep=6em, row sep=2.8em]
  & \ET{\Msys} & \\
  \Th{\Msys}
    \arrow[ur, bend left=15]
  & &
  J_E
    \arrow[ul, bend right=15]
\end{tikzcd}
\caption{The geometric theory~$\Th{\Msys}$ classifies the LTS structure;
each topology~$J_E$ localizes the classifying topos to the subtopos
enforcing the corresponding equivalence.  The middle
topology~$J_{\mathrm{sim}}$ is the first requiring Caramello's
quotient-theory machinery; the 10~intermediate equivalences of the
Energy--LT Bridge depend on it as well
(\S\ref{sec:explicit-topologies}).}
\label{fig:bridge}
\end{figure}

Our topos-theoretic framework does not merely rederive known
process-algebraic results; it provides the algebraic context in which
new questions become natural.  Process algebra is Turing-complete; that
these results \emph{could} in principle be computed without topos theory
is a triviality.  What topos theory provides is the conceptual lens:
the coframe distributive law, the Heyting
implication $S \to F = \mathrm{IF}$, and the 17~unnamed spectrum
elements are consequences of asking ``what is the lattice of all
behavioral abstractions?'', a question that requires the subtopos
coframe to formulate.
Each subtopos of the classifying topos corresponds to a choice of
which geometric properties to retain; the spectrum structure emerges
from the coframe of all such choices.
This contrasts with the closest sheaf-theoretic precedent:
Eberhart, Hirschowitz, and Seiller~\cite{eberhart2017sheaf} construct a
single Grothendieck topology enforcing bisimulation on interaction trees
for the $\pi$-calculus; our approach fixes the site
($\mathrm{f.p.LTS}_L$) and varies the topology across the
\emph{entire} spectrum, producing a lattice of subtoposes rather than a
single quotient.

Two features of the bridge are irreducibly topos-theoretic.  Its
subtopos lattice carries coframe structure: meets distribute over
joins, providing the spectrum with semantic lattice operations it does
not intrinsically possess; while one can close the spectrum under meets
and joins in a product frame, the coframe law is an invariant of the
classifying topos, not of the energy-game representation (indeed, the
lattice closure itself depends on the coordinatization:
Remark~\ref{rem:coord-dependence}).  Furthermore,
the coframe meet of two named equivalences can produce elements not
generated by any single energy-game: the 17~unnamed elements in the
lattice closure have concrete energy 6-tuples yet correspond to no named
process equivalence.

Our construction proceeds through three stages.  First, each system~$\Msys$
receives a geometric theory~$\Th{\Msys}$ whose classifying
topos~$\ET{\Msys}$ determines its provable geometric sequents; explicit
separating pairs---the fork~$\FK$ versus the path~$\PT$, the
hub-spokes~$\HS$ versus the two-cycle~$\TC$---establish a strict
three-level hierarchy from mutual simulation through bisimulation to
topos equivalence
(Theorems~\ref{thm:first-separation}--\ref{thm:second-separation};
\S\ref{sec:hierarchy}).  Hub-spokes' 5-element Lindenbaum
algebra and 8~Grothendieck topologies give a concrete instance of the
full spectrum stratification (\S\ref{sec:hierarchy}).  Second, a geometric van~Benthem theorem
identifies the bisimulation-invariant fragment of geometric logic with
diamond-only Hennessy--Milner logic: the forward direction is proved
constructively (Theorem~\ref{thm:hml-invariant}), and the converse is
established via bounded tree unraveling
(Theorem~\ref{thm:gvb-formalized}), with independent confirmation at
depths~$\leq 2$
(Theorems~\ref{thm:depth0}--\ref{thm:depth2};
\S\ref{sec:hml}).  Third, energy budgets determine Grothendieck
topologies: $J_{\mathrm{bisim}}$ and $J_{\mathrm{trace}}$ receive
constructive covering predicates, a constructive counterexample shows
the observation-class approach provably fails for simulation
(Theorem~\ref{thm:naive-sim-instability}), and $J_{\mathrm{sim}}$
completes the chain via Caramello's duality
(Corollary~\ref{cor:three-topology-chain}).  Energy--LT
extends this to all 13~named equivalences
(Theorem~\ref{thm:energy-lt-bridge}), and lattice closure yields the
30-element spectrum lattice~$L_{30}$, whose bi-Heyting structure is
detailed in
\S\ref{sec:spectrum-lattice}--\S\ref{sec:explicit-topologies}.  At the
presheaf level, the Geometric Closure Theorem reduces Heyting
implication to testing at a single canonical free extension
(Theorem~\ref{thm:geometric-closure};
\S\ref{sec:geometric-closure}), making the coframe operations concretely
computable.

\section{Geometric Theories and Classifying Toposes}
\label{sec:geometric}

Geometric logic is the fragment of first-order logic closed under
$\land$, $\lor$ (including infinitary), $\exists$, and~$\top$; all
theories in this paper are finitary.
This fragment is natural for toposes because geometric morphisms preserve
exactly these connectives~\cite{johnstone2002sketches}.

A \emph{geometric theory}~$\mathbb{T}$ over a language~$\mathcal{L}$
consists of axioms in the form of \emph{geometric sequents}
$\phi \vdash_{\vec{x}} \psi$, where $\phi$ and $\psi$ are geometric
formulas.  A \emph{model} of~$\mathbb{T}$ in a topos~$\mathcal{E}$
interprets the sorts, function symbols, and relation symbols
of~$\mathcal{L}$ as objects and morphisms in~$\mathcal{E}$, satisfying all
axioms of~$\mathbb{T}$.

The \emph{classifying topos} $\mathcal{E}[\mathbb{T}]$ is the universal
topos representing~$\mathbb{T}$: for any Grothendieck topos~$\mathcal{E}$,
there is an equivalence between geometric morphisms
$\mathcal{E} \to \mathcal{E}[\mathbb{T}]$ and models of~$\mathbb{T}$
in~$\mathcal{E}$, so that the study of models in any topos reduces to
the study of geometric morphisms.  It is constructed
as the sheaf topos on the \emph{syntactic
category}~$\mathbf{C}_\mathbb{T}$---whose objects are formulas-in-context
and whose morphisms are provably functional relations---with the
syntactic topology generated by the axioms
of~$\mathbb{T}$~\cite{makkai1977first}.  Two geometric theories over the same language have
equivalent classifying toposes if and only if they prove the same
geometric sequents: the generic model
in~$\mathcal{E}[\mathbb{T}]$ validates exactly the geometric consequences
of~$\mathbb{T}$~\cite[D1.4.3]{johnstone2002sketches}.  This is our
invariant: provable sequents distinguish classifying toposes.

A \emph{bi-interpretation} between geometric theories $\mathbb{T}$ and
$\mathbb{T}'$ consists of an interpretation $I \colon \mathbb{T} \to
\mathbb{T}'$ and an interpretation $J \colon \mathbb{T}' \to \mathbb{T}$
such that the round-trip composites $JI$ and $IJ$ are definably
isomorphic to the respective identities.  If
$\mathbb{T} \bimark \mathbb{T}'$, the Comparison
Lemma~\cite[Theorem~C2.2.3]{johnstone2002sketches} implies
$\mathcal{E}[\mathbb{T}] \simeq \mathcal{E}[\mathbb{T}']$.  Its converse
does not hold: bi-interpretation is strictly stronger than Morita
equivalence~\cite{caramello2017theories}.  In Section~\ref{sec:multiway}, we
will specialize this theory-level notion to rooted transition systems, where it
becomes a condition on simulations.

\begin{defi}[Bisimulation-invariant sequent]
\label{def:bisim-invariant}
A geometric sequent $\phi \vdash_{\vec{x}} \psi$ over a relational
signature is \emph{bisimulation-invariant} if for every relational
bisimulation~$R$ between systems~$M$ and~$N$ (defined precisely in
Definition~\ref{def:rel-bisim}): if
$M \models \phi \vdash_{\vec{x}} \psi$ then
$N \models \phi \vdash_{\vec{x}} \psi$.
\end{defi}

\section{Rooted Transition Systems and the Three-Level Hierarchy}
\label{sec:multiway}
\label{sec:hierarchy}

A \emph{rooted transition system}~$\Msys$ consists of a set~$S_\Msys$ of states, a
transition relation on $S_\Msys$, and an initial
state $\star_\Msys \in S_\Msys$.  A \emph{simulation} $f \colon \Msys \to \Nsys$ is a
function $f \colon S_\Msys \to S_\Nsys$ that preserves transitions ($s \to s'$ in $\Msys$
implies $f(s) \to f(s')$ in $\Nsys$) and initial states ($f(\star_\Msys) = \star_\Nsys$).
Simulations compose and include identities, so rooted transition systems form a
category.

To each rooted transition system~$\Msys$ we associate a geometric theory~$\Th{\Msys}$ over
the \emph{rooted transition language}~$\mathcal{L}_\Msys$ with sort~$S$,
constants~$c_s$ for each state $s \in S_\Msys$ (including a distinguished
constant~$\init$ for the initial state), and binary predicates
$\step$~(transition), $\reachpred$~(reachability), $\pe$~(path
equivalence).  The theory~$\Th{\Msys}$ consists of two components:
\begin{enumerate}
\item \emph{Structural axioms} (6~sequents): $\step(x,y) \vdash \reachpred(x,y)$,
  $\top \vdash \reachpred(x,x)$, transitivity, $\pe$ definition, $\pe$ symmetry,
  and $\pe \vdash \reachpred$.
\item \emph{System-specific axioms}: for each transition $s \to t$ in $\Msys$,
  the existence axiom $\top \vdash \step(c_s,c_t)$; for each state~$s$,
  the completeness axiom $\step(c_s,x) \vdash x = c_{t_1} \lor \cdots \lor x = c_{t_n}$
  where $\{t_1, \ldots, t_n\}$ is the set of successors of~$s$
  (when $s$ has no successors, the disjunction is empty, yielding
  $\step(c_s,x) \vdash \bot$); negative axioms
  $\reachpred(c_s,c_t) \vdash \bot$ for each non-reachable pair and
  $\pe(c_s,c_t) \vdash \bot$ for each non-path-equivalent pair
  (pinning down all three predicates from above, not only from below);
  and the domain closure axiom
  $\top \vdash x = c_{s_1} \lor \cdots \lor x = c_{s_k}$
  where $\{s_1,\ldots,s_k\} = S_\Msys$.
\end{enumerate}
These are all geometric sequents.  The theory~$\Th{\Msys}$ records the
structural properties of~$\Msys$ expressible in geometric logic, including
branching degree, local determinism, and confluence patterns.  The
classifying topos~$\ET{\Msys}$ is the classifying topos of~$\Th{\Msys}$.
For finite-state systems, $\Th{\Msys}$ is finite; the construction is effective.

\begin{lem}[Soundness and Completeness]
\label{lem:soundness-completeness}
Geometric provability from~$\Th{\Msys}$ coincides with semantic validity.
\begin{enumerate}
\item \emph{Soundness.}  If a geometric sequent~$s$ is provable from a
  geometric theory~$\mathbb{T}$, then every model of~$\mathbb{T}$ satisfies~$s$.
\item \emph{Completeness.}  A geometric sequent~$s$ is provable
  from~$\Th{\Msys}$ if and only if the canonical model satisfies~$s$.
\end{enumerate}
\end{lem}

\begin{proof}
Soundness is by structural induction on the derivation: axioms hold by
hypothesis, cut composes, and the geometric connectives ($\land$, $\lor$,
$\exists$) preserve satisfaction.
For completeness, the generic model in~$\ET{\Msys}$
validates exactly the provable geometric
sequents~\cite[Cor.~2.1.12]{caramello2017theories}.
The system-specific axioms determine any Set-model of~$\Th{\Msys}$ up
to isomorphism: domain closure forces every element to equal
some~$c_s$; the existence and completeness axioms pin down~$\step$;
and the negative axioms pin down~$\reachpred$ and~$\pe$.
Since~$\Th{\Msys}$ is coherent (all disjunctions are finite for
finite-state systems), its classifying topos has enough
Set-points by Deligne's theorem~\cite[D3.3.13]{johnstone2002sketches}.
Thus: the canonical model satisfies~$s$ $\Rightarrow$
every Set-model satisfies~$s$ (uniqueness) $\Rightarrow$
the generic model satisfies~$s$ (enough points) $\Rightarrow$
$s$~is provable.  The converse is soundness applied to the canonical model.
\end{proof}

\noindent
The separation proofs in
\S\S\ref{sec:first-sep}--\ref{sec:second-sep} and
\S\ref{sec:mechanisms} rely on four \emph{semantic bridge lemmas}
that characterize when the separating sequents hold:
$\sigma_{\mathrm{tot}}$ holds iff every state has a successor,
$\sigma_{\mathrm{det}}$ holds iff the transition relation is
functional,
$\sigma_{\mathrm{conf}}$ holds iff every pair of successors has a
common successor, and
$\sigma_{\mathrm{loop}}$ holds iff every state has a self-loop.

\begin{rem}[Size of the theory]
\label{rem:full-theory}
For a finite system with $|S_\Msys|$~states, $\Th{\Msys}$ is finite:
6~structural axioms, at most $|S_\Msys|^2$ existence axioms,
$|S_\Msys|$~completeness axioms, $O(|S_\Msys|^2)$ negative axioms
for~$\reachpred$ and~$\pe$, and one domain closure
axiom, for $O(|S_\Msys|^2)$ sequents in total.  For infinite systems the
completeness and domain closure axioms may involve infinite
disjunctions (geometric but no longer coherent); the classifying topos
exists regardless~\cite[D1.4]{johnstone2002sketches}.  All witness
systems in this paper are finite.
\end{rem}

\begin{defi}[Path equivalence]
\label{sec:path-equiv}
Two states $s, t \in S_\Msys$ are \emph{path-equivalent}, written $s \patheq t$,
if they reach exactly the same states and are reached from exactly the same
states: $s \patheq t$ iff for all $u \in S_\Msys$,
$s \reach u \Leftrightarrow t \reach u$ and
$u \reach s \Leftrightarrow u \reach t$.
\end{defi}

\noindent
Path equivalence is an equivalence relation.  The key property is that
it is implied by mutual reachability.

\begin{lem}[Lifting Lemma]
\label{lem:lifting}
In any rooted transition system~$\Msys$:
\[
\boxed{\quad s \reach t \;\text{and}\; t \reach s
  \quad\Longrightarrow\quad s \patheq t \quad}
\]
\end{lem}

\begin{proof}
Suppose $s \reach t$ and $t \reach s$.  We must show that $s$ and $t$
reach exactly the same states and are reached from exactly the same states.
If $s \reach u$, then $t \reach s \reach u$, so $t \reach u$.  If
$t \reach u$, then $s \reach t \reach u$, so $s \reach u$.  If
$u \reach s$, then $u \reach s \reach t$, so $u \reach t$.  If
$u \reach t$, then $u \reach t \reach s$, so $u \reach s$.  Thus
$s \patheq t$.
\end{proof}

\begin{defi}[Relational bisimulation]
\label{def:rel-bisim}
A \emph{relational bisimulation} between rooted transition systems~$\Msys$ and~$\Nsys$ is a
relation $R \subseteq S_\Msys \times S_\Nsys$ satisfying the back-and-forth lifting
property: if $(s, t) \in R$ and $s \to s'$ in $\Msys$, then there exists~$t'$ with
$t \to t'$ in $\Nsys$ and $(s', t') \in R$ (\emph{forth}); and if $(s, t) \in R$ and
$t \to t'$ in $\Nsys$, then there exists~$s'$ with $s \to s'$ in $\Msys$ and
$(s', t') \in R$ (\emph{back}).  The systems $\Msys$ and $\Nsys$ are \emph{bisimilar},
written $\Msys \bismark \Nsys$, if there exists a relational bisimulation~$R$ with
$(\star_\Msys, \star_\Nsys) \in R$.
\end{defi}

We also introduce a stronger, functorial notion that connects to
bi-interpretation.

\begin{defi}[Functional bisimulation]
\label{def:func-bisim}
A \emph{functional bisimulation} between rooted transition systems~$\Msys$ and~$\Nsys$ is a
pair of simulations $f \colon \Msys \to \Nsys$, $g \colon \Nsys \to \Msys$ with
$gf(s) \patheq_\Msys s$ for all $s \in S_\Msys$ and $fg(t) \patheq_\Nsys t$ for all
$t \in S_\Nsys$.
\end{defi}

\noindent
In the coalgebraic literature, ``functional bisimulation'' sometimes
denotes a single function whose graph is a relational bisimulation,
i.e.\ a coalgebra homomorphism~\cite{rutten2000universal}.  Our usage
is distinct: we require a \emph{pair} of simulations with round-trip
path-equivalence coherence, capturing bi-interpretation rather than
one-sided functional witnessing.

Every functional bisimulation implies relational bisimilarity.  The converse
does not hold: relational bisimilarity does not in general imply the
existence of a functional bisimulation (see
Remark~\ref{rem:func-vs-rel}).  Throughout this paper, ``bisimulation''
refers to the relational notion of Park and
Milner~\cite{park1981,milner1989} (Definition~\ref{def:rel-bisim})
unless explicitly qualified as ``functional.''

\begin{rem}[Functional vs.\ relational bisimulation]
\label{rem:func-vs-rel}
Functional bisimulation is strictly stronger than relational bisimulation,
even for image-finite systems.  A relational bisimulation can relate one
state in~$\Msys$ to many states in~$\Nsys$ without requiring any functional
coherence.  For example, a system with a single self-looping state is
relationally bisimilar to any total system (every state has a successor) via
the total relation, since the back-and-forth conditions are trivially met
when one side can always respond with its self-loop.  However, no functional
bisimulation need exist between such systems if their path-equivalence
quotients differ.
\end{rem}

The \emph{self-loop system}~$\SL$ has a single state~$a$ with
transition $a \to a$ and initial state~$a$.  The \emph{two-cycle}~$\TC$
has states $\{x, y\}$ with transitions $x \to y$, $y \to x$ and initial
state~$x$ (defined in full in \S\ref{sec:second-sep}).

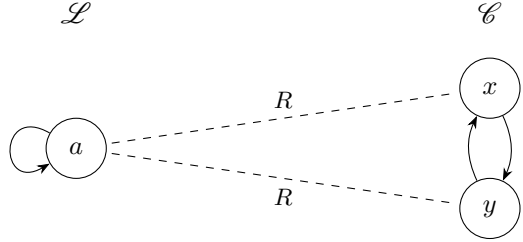
\begin{figure}[ht]
\centering
\begin{tikzpicture}[
    >=Stealth,
    state/.style={circle, draw, minimum size=8mm, inner sep=1pt, font=\small},
  ]
  \node at (-2.5, 1.8) {$\SL$};
  \node[state] (a) at (-2.5, 0) {$a$};
  \draw[->] (a) to [out=150, in=210, looseness=5] (a);

  \node at (3.0, 1.8) {$\TC$};
  \node[state] (x) at (3.0, 0.8) {$x$};
  \node[state] (y) at (3.0, -0.8) {$y$};
  \draw[->, bend left=25] (x) to (y);
  \draw[->, bend left=25] (y) to (x);

  \draw[dashed, shorten >=2pt, shorten <=2pt]
    (a) to node[above, font=\footnotesize] {$R$} (x);
  \draw[dashed, shorten >=2pt, shorten <=2pt]
    (a) to node[below, font=\footnotesize] {$R$} (y);
\end{tikzpicture}
\caption{The self-loop~$\SL$ and two-cycle~$\TC$.  The dashed lines
show the relational bisimulation $R = \{(a,x),(a,y)\}$.  No simulation
$\SL \to \TC$ exists: $\TC$~has no self-loops.}
\label{fig:sl-tc}
\end{figure}

\begin{prop}[Functional vs.\ relational separation]
\label{prop:four-level}
Functional bisimulation is strictly stronger than relational bisimulation.
The systems $\SL$ and~$\TC$ are relationally bisimilar
(via $R = \{(a,x),(a,y)\}$) but not functionally bisimilar.
\end{prop}

\begin{proof}
\emph{Relational bisimulation.}\enspace
The relation $R = \{(a,x),(a,y)\}$ satisfies back-and-forth.
At $(a,x)$: the transition $a \to a$ in~$\SL$ is matched by $x \to y$
in~$\TC$ with $(a,y) \in R$.  At $(a,y)$: the transition $a \to a$ is
matched by $y \to x$ with $(a,x) \in R$.  Back conditions are
symmetric.

\emph{No functional bisimulation.}\enspace
A functional bisimulation requires a simulation $f \colon \SL \to \TC$.
The self-loop $a \to a$ must map to $f(a) \to f(a)$ in~$\TC$.  But
$\TC$~has no self-loops: $x$ has unique successor~$y \neq x$, and $y$
has unique successor~$x \neq y$.  No simulation
$\SL \to \TC$ exists, hence no functional bisimulation.
\end{proof}

\noindent
This is a refinement of the three-level hierarchy: together with the First
and Second Separation theorems below, it yields four strict levels:
\[
\text{mutual simulation}
\;\supsetneq\; \text{relational bisimulation}
\;\supsetneq\; \text{functional bisimulation}
\;\supsetneq\; \text{topos equivalence.}
\]
The main narrative treats ``bisimulation'' as the relational notion
(per the convention above), giving three principal levels; the
functional--relational distinction is secondary.

\subsection{Bi-interpretation of Systems}
\label{sec:bi-interp}

A \emph{bi-interpretation of rooted transition systems}~$\Msys$ and~$\Nsys$, written
$\Msys \bimark \Nsys$, is a pair of simulations $f \colon \Msys \to \Nsys$,
$g \colon \Nsys \to \Msys$ satisfying the weaker coherence condition that $gf(s)$
and $s$ are \emph{mutually reachable} ($gf(s) \reach s$ and
$s \reach gf(s)$) for all $s \in S_\Msys$, and $fg(t)$ and $t$ are mutually
reachable for all $t \in S_\Nsys$.

Despite the shared name, a bi-interpretation of systems does \emph{not}
induce a bi-interpretation of the associated geometric theories.  A
simulation $f \colon \Msys \to \Nsys$ preserves all geometric sequents in the
bisimulation-invariant fragment (the tree-shaped, equality-free formulas
corresponding to~$\HML$), but need not preserve arbitrary geometric sequents
of~$\Th{\Msys}$; for instance, $f$ can map a deterministic system into a
non-deterministic one, losing the sequent~$\sigma_{\mathrm{det}}$.  The gap
between system-level bi-interpretation and theory-level bi-interpretation is
exactly the content of the Second Separation
(Theorem~\ref{thm:second-separation}): bisimilar (hence mutually
bi-interpretable, by Theorem~\ref{thm:correspondence}) systems can have
non-equivalent classifying toposes.

\begin{thm}[Functional Bisimulation--Bi-Interpretation Correspondence]
\label{thm:correspondence}
For all rooted transition systems~$\Msys$ and~$\Nsys$, there exists a functional
bisimulation between~$\Msys$ and~$\Nsys$ if and only if there exists a
bi-interpretation between them:
\[
\boxed{\quad \Msys \text{ and } \Nsys \text{ are functionally bisimilar}
  \quad\Longleftrightarrow\quad
  \Msys \bimark \Nsys \quad}
\]
\end{thm}

\begin{proof}
Forward, path equivalence implies mutual
reachability, since if $s \patheq t$ then in particular $s \reach t$ and
$t \reach s$ (take $u = t$ and $u = s$ respectively in the definition of
path equivalence; more precisely, $s \reach s$ holds trivially, and
$s \patheq t$ gives $t \reach s$, and symmetrically $s \reach t$).
Therefore any functional bisimulation is a bi-interpretation.

Conversely, the Lifting Lemma applies.  Suppose
$(f, g) \colon \Msys \bimark \Nsys$ is a bi-interpretation, so that $gf(s) \reach s$
and $s \reach gf(s)$ for all $s \in S_\Msys$, and $fg(t) \reach t$ and
$t \reach fg(t)$ for all $t \in S_\Nsys$.  By Lemma~\ref{lem:lifting}, mutual
reachability implies path equivalence: $gf(s) \patheq_\Msys s$ for all $s$ and
$fg(t) \patheq_\Nsys t$ for all $t$.  Thus $(f, g)$ is a functional bisimulation.
\end{proof}

\medskip
With the framework in place, we establish the strict three-level hierarchy
by exhibiting witnesses for each separation and proving a bridge theorem
relating the two finer levels.

\subsection{First Separation: Mutual Simulation $\supsetneq$ Bisimulation}
\label{sec:first-sep}

The \emph{fork}~$\FK$ has states $\{a, b, c\}$, transitions $a \to b$,
$a \to c$, and $b \to b$, and initial state~$a$.  The initial state~$a$
branches: one branch leads to a \emph{halting} state~$c$ (no outgoing
transitions), and the other to a self-looping state~$b$.
The \emph{path}~$\PT$ has states $\{x, y\}$, transitions $x \to y$ and
$y \to y$, and initial state~$x$.  It is a single step followed by a
self-loop.

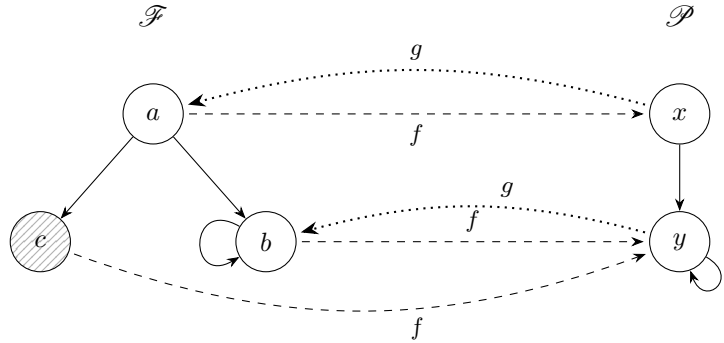
\begin{figure}[ht]
\centering
\begin{tikzpicture}[
    >=Stealth,
    state/.style={circle, draw, minimum size=8mm, inner sep=1pt, font=\small},
    dead/.style={circle, draw, minimum size=8mm, inner sep=1pt,
                 font=\small, pattern=north east lines, pattern color=black!30},
  ]
  \node at (-3.5, 1.8) {$\FK$};
  \node[state] (a) at (-3.5, 0.5) {$a$};
  \node[dead]  (c) at (-5.0, -1.2) {$c$};
  \node[state] (b) at (-2.0, -1.2) {$b$};

  \draw[->] (a) -- (c);
  \draw[->] (a) -- (b);
  \draw[->] (b) to [out=150, in=210, looseness=5] (b);

  \node at (3.5, 1.8) {$\PT$};
  \node[state] (x) at (3.5, 0.5) {$x$};
  \node[state] (y) at (3.5, -1.2) {$y$};

  \draw[->] (x) -- (y);
  \draw[->] (y) to [out=330, in=290, looseness=5] (y);

  \draw[->, dashed, shorten >=2pt, shorten <=2pt]
    (a) to node[below, font=\footnotesize] {$f$} (x);
  \draw[->, dashed, shorten >=2pt, shorten <=2pt]
    (b) to node[above, font=\footnotesize, sloped] {$f$} (y);
  \draw[->, dashed, shorten >=2pt, shorten <=2pt, bend right=20]
    (c) to node[below, font=\footnotesize, sloped, pos=0.6] {$f$} (y);

  \draw[->, dotted, thick, shorten >=2pt, shorten <=2pt, bend right=15]
    (x) to node[above, font=\footnotesize, pos=0.5] {$g$} (a);
  \draw[->, dotted, thick, shorten >=2pt, shorten <=2pt, bend right=15]
    (y) to node[above, font=\footnotesize, sloped, pos=0.4] {$g$} (b);
\end{tikzpicture}
\caption{The fork~$\FK$ and path~$\PT$ with mutual simulations.
Dashed arrows: $f(a)=x$, $f(b)=f(c)=y$.
Dotted arrows: $g(x)=a$, $g(y)=b$.
The halting state~$c$ (hatched) has no preimage under~$g$.}
\label{fig:fk-pt}
\end{figure}

\begin{thm}[First Separation]
\label{thm:first-separation}
The systems $\FK$ and $\PT$ mutually simulate.  They are not bisimilar, and
their classifying toposes are not equivalent:
\[
\boxed{\quad \Hom(\FK, \PT) \neq \varnothing,\;
  \Hom(\PT, \FK) \neq \varnothing,\quad
  \FK \not\bismark \PT,\quad
  \ET{\FK} \not\simeq \ET{\PT} \quad}
\]
\end{thm}

\begin{proof}
\emph{Mutual simulation.}\enspace
The simulation $f \colon \FK \to \PT$ is defined by $f(a) = x$, $f(b) = y$,
$f(c) = y$.  This preserves transitions: $a \to b$ maps to $x \to y$,
$a \to c$ maps to $x \to y$, and $b \to b$ maps to $y \to y$.  The halting
state~$c$ has no outgoing transitions, so there is nothing to verify.  It
preserves the initial state: $f(a) = x$.
The simulation $g \colon \PT \to \FK$ is defined by $g(x) = a$, $g(y) = b$.
This preserves transitions: $x \to y$ maps to $a \to b$, and $y \to y$ maps
to $b \to b$.  It preserves the initial state: $g(x) = a$.

\emph{Non-bisimilarity.}\enspace
Suppose for contradiction
that~$R$ is a relational bisimulation with $(a, x) \in R$.  The forth
condition applied to the transition $a \to c$ in~$\FK$ yields a state~$t'$
with $x \to t'$ in~$\PT$ and $(c, t') \in R$.  Since $x$ has unique
successor~$y$, we have $t' = y$ and thus $(c, y) \in R$.  Now apply the
back condition at $(c, y)$: the transition $y \to y$ in~$\PT$ requires a
state $s'$ with $c \to s'$ in~$\FK$ and $(s', y) \in R$.  But $c$ has no
outgoing transitions; it is a halting state.  This is a contradiction.

\emph{Topos separation.}\enspace
The
\emph{totality sequent}
$\sigma_{\mathrm{tot}} := \top \vdash_x \exists y.\,\step(x,y)$
asserts that every state has a successor.  This sequent is provable
in~$\Th{\PT}$: state~$x$ has successor~$y$, and state~$y$ has
successor~$y$.  It is not provable in~$\Th{\FK}$: state~$c$ has no
successor.  Since $\sigma_{\mathrm{tot}}$ is a geometric sequent provable
in~$\Th{\PT}$ but not in~$\Th{\FK}$, the two theories have different
deductive closures.  Provability of geometric sequents is invariant under
equivalence of classifying toposes---it is detected by the generic
model's internal logic~\cite[D1.4.3]{johnstone2002sketches}---so
$\ET{\FK} \not\simeq \ET{\PT}$.
\end{proof}

\begin{rem}[Halting intuition]
\label{rem:halting-intuition}
The First Separation has a natural computational interpretation.  The
fork~$\FK$ can halt (by reaching state~$c$), while the path~$\PT$ cannot.
Yet they mutually simulate: every simulation from~$\PT$ into~$\FK$ maps to
the non-halting branch, and the simulation from~$\FK$ to~$\PT$ collapses
the halting branch.  Bisimulation, however, requires the back condition to
hold at every related pair, and no state in~$\PT$ can match the halting
behavior of~$c$.
\end{rem}

\subsection{Second Separation: Bisimulation $\supsetneq$ Topos Equivalence}
\label{sec:second-sep}

The \emph{hub-spokes system}~$\HS$ has states $\{a, b, c\}$, transitions
$a \to b$, $a \to c$, $b \to a$, $c \to a$, and initial state~$a$.  The
hub~$a$ connects to two spokes $b$ and $c$, each of which connects back to
the hub.  The \emph{two-cycle}~$\TC$ has states $\{x, y\}$, transitions
$x \to y$, $y \to x$, and initial state~$x$.  It is a simple alternating
cycle.

\begin{figure}[ht]
\centering
\begin{tikzpicture}[
    >=Stealth,
    state/.style={circle, draw, minimum size=8mm, inner sep=1pt, font=\small},
  ]
  \node at (-3.0, 1.8) {$\HS$};
  \node[state] (a) at (-3.0, 0.5) {$a$};
  \node[state] (c) at (-4.5, -1.2) {$c$};
  \node[state] (b) at (-1.5, -1.2) {$b$};

  \draw[->, bend left=15] (a) to (c);
  \draw[->, bend left=15] (c) to (a);
  \draw[->, bend left=15] (a) to (b);
  \draw[->, bend left=15] (b) to (a);

  \node at (3.0, 1.8) {$\TC$};
  \node[state] (x) at (3.0, 0.5) {$x$};
  \node[state] (y) at (3.0, -1.2) {$y$};

  \draw[->, bend left=20] (x) to (y);
  \draw[->, bend left=20] (y) to (x);

  \draw[->, dashed, shorten >=2pt, shorten <=2pt]
    (a) to node[above, font=\footnotesize] {$I$} (x);
  \draw[->, dashed, shorten >=2pt, shorten <=2pt]
    (b) to node[above, font=\footnotesize, sloped] {$I$} (y);
  \draw[->, dashed, shorten >=2pt, shorten <=2pt, bend right=20]
    (c) to node[below, font=\footnotesize, sloped, pos=0.6] {$I$} (y);
\end{tikzpicture}
\caption{The hub-spokes~$\HS$ and two-cycle~$\TC$.
Dashed arrows show the simulation $I \colon \HS \to \TC$ with
$I(a)=x$, $I(b)=I(c)=y$; the reverse simulation is
$K(x)=a$, $K(y)=b$.  Non-injectivity $I(b)=I(c)$ collapses
the branching that $\sigma_{\mathrm{det}}$ detects.}
\label{fig:hs-tc}
\end{figure}
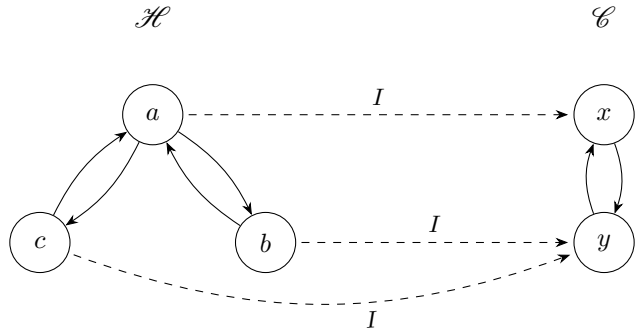

\begin{thm}[Second Separation]
\label{thm:second-separation}
The systems $\HS$ and $\TC$ are bisimilar, yet their classifying toposes are
not equivalent:
\[
\boxed{\quad \HS \bismark \TC \qquad\text{yet}\qquad
  \ET{\HS} \not\simeq \ET{\TC} \quad}
\]
\end{thm}

\begin{proof}
\emph{Functional bisimulation.}\enspace
Define $I \colon \HS \to \TC$ by $I(a) = x$, $I(b) = y$,
$I(c) = y$.  This preserves transitions: $a \to b$ maps to $x \to y$,
$a \to c$ maps to $x \to y$, $b \to a$ maps to $y \to x$, $c \to a$ maps
to $y \to x$.  Define $K \colon \TC \to \HS$ by $K(x) = a$, $K(y) = b$.
This preserves transitions: $x \to y$ maps to $a \to b$, $y \to x$ maps to
$b \to a$.

\emph{Coherence.}\enspace
The round-trips on~$\HS$ are:
$KI(a) = a$, $KI(b) = b$, and $KI(c) = b$.  The
first two are identities.  For the third,
$b \patheq_{\HS} c$ since $b \to a \to c$ and
$c \to a \to b$ witness mutual reachability (Lemma~\ref{lem:lifting}).
The round-trips on~$\TC$ are:
$IK(x) = x$ and $IK(y) = y$, both identities.  Thus
$(I, K)$ is a functional bisimulation.

\emph{Relational bisimilarity.}\enspace
The relation $R = \{(a,x),(b,y),(c,y)\}$ is a relational bisimulation
directly: at~$(a,x)$, the transitions $a \to b$ and $a \to c$ are matched by
$x \to y$ with $(b,y),(c,y) \in R$; at~$(b,y)$, $b \to a$ is matched by
$y \to x$ with $(a,x) \in R$; at~$(c,y)$, $c \to a$ is matched by $y \to x$
with $(a,x) \in R$; back conditions are symmetric.  Hence $\HS \bismark \TC$.

\emph{Topos separation.}\enspace
The determinism sequent
$\sigma_{\mathrm{det}} := \step(x,y) \land \step(x,z) \vdash y = z$ is
provable in~$\Th{\TC}$: state~$x$ has exactly one successor~$y$, and
state~$y$ has exactly one successor~$x$.  The sequent is not provable
in~$\Th{\HS}$: state~$a$ has two distinct successors $b$ and $c$ with
$b \neq c$.  Since $\sigma_{\mathrm{det}}$ is a geometric sequent provable
in~$\Th{\TC}$ but not in~$\Th{\HS}$, the two theories have different
deductive closures.  Provability of geometric sequents is invariant under
equivalence of classifying toposes~\cite[D1.4.3]{johnstone2002sketches},
so $\ET{\HS} \not\simeq \ET{\TC}$.
\end{proof}

\begin{rem}[Multiple witnesses]
\label{rem:multiple-witnesses}
The Second Separation is not limited to the determinism sequent.  In
Section~\ref{sec:mechanisms}, two additional bisimilar pairs independently
demonstrate non-equivalent classifying toposes: the diamond graph and
confluence tree (separated by a weak confluence sequent), and the self-loop
system and two-cycle (separated by a universal self-loop sequent).
\end{rem}

\begin{cor}[Bi-interpretation does not imply topos equivalence]
\label{cor:biinterp-not-topos}
System-level bi-interpretation does not imply topos equivalence:
$\HS$ and~$\TC$ satisfy $\HS \bimark \TC$ (bi-interpretation, by
Theorem~\ref{thm:correspondence}) and
$\HS \bismark \TC$ (relational bisimilarity), yet
$\ET{\HS} \not\simeq \ET{\TC}$.
\end{cor}

\subsection{The Quotient Bridge}
\label{sec:quotient}

For a rooted transition system~$\Msys$, the \emph{path-equivalence quotient}
$\Msys/{\patheq}$ has states $S_\Msys / {\patheq}$, transitions $[s_1] \to [s_2]$
whenever some representatives have $s_1' \to s_2'$ with $s_1' \patheq s_1$
and $s_2' \patheq s_2$, and initial state $[\star_\Msys]$.  The quotient
collapses path-equivalent states into single equivalence classes.

\begin{thm}[Quotient Bridge]
\label{thm:quotient-bridge}
For all rooted transition systems~$\Msys$ and~$\Nsys$, if there exists a functional
bisimulation between~$\Msys$ and~$\Nsys$, then their path-equivalence
quotients have equivalent classifying toposes:
\[
\boxed{\quad \Msys \text{ and } \Nsys \text{ are functionally bisimilar}
  \;\Longrightarrow\;
\ET{\Msys/{\patheq}} \simeq \ET{\Nsys/{\patheq}} \quad}
\]
\end{thm}

\begin{proof}
Let $(f, g)$ be a functional bisimulation between~$\Msys$ and~$\Nsys$.  The
argument has two steps: descent to quotients, then mutual inverse.

\smallskip\noindent\emph{Well-definedness.}\enspace
We show $f$ preserves path equivalence: if $s_1 \patheq_\Msys s_2$, then
$f(s_1) \patheq_\Nsys f(s_2)$.  For forward reachability, suppose
$f(s_1) \reach v$ in~$\Nsys$; we must show $f(s_2) \reach v$.  Since $g$ is a
simulation, $gf(s_1) \reach g(v)$ in~$\Msys$.  The coherence condition
$gf(s_1) \patheq_\Msys s_1$ gives $s_1 \reach g(v)$, and then
$s_1 \patheq_\Msys s_2$ gives $s_2 \reach g(v)$.  Since $f$ is a simulation,
$f(s_2) \reach fg(v)$.  The coherence condition $fg(v) \patheq_\Nsys v$ implies
$fg(v) \reach v$ (path equivalence implies mutual reachability), so
$f(s_2) \reach fg(v) \reach v$.  Backward reachability is symmetric.
Thus $f$ induces $\bar{f} \colon \Msys/{\patheq} \to \Nsys/{\patheq}$ by
$\bar{f}([s]) = [f(s)]$; similarly $g$ induces~$\bar{g}$.

\smallskip\noindent\emph{Mutual inverse.}\enspace
The composite $\bar{g}\bar{f}([s]) = [gf(s)]$.  The coherence condition
$gf(s) \patheq_\Msys s$ gives $[gf(s)] = [s]$, so $\bar{g}\bar{f} = \id$.
Symmetrically $\bar{f}\bar{g} = \id$.  Since $f$ preserves both transitions
and path equivalence, $\bar{f}$ is a simulation of quotient systems; likewise
$\bar{g}$.  Thus $\bar{f}$ and $\bar{g}$ are mutually inverse simulations,
giving $\Msys/{\patheq} \cong \Nsys/{\patheq}$ as rooted transition systems.  Isomorphic
systems have identical geometric theories, hence equivalent classifying toposes.
\end{proof}

By the Quotient Bridge, functional bisimulation preserves the
classifying topos of the quotient, even though relational bisimulation need
not preserve the classifying topos of the original system.  The quotient
collapses exactly the branching structure that the determinism sequent
detects: in $\HS/{\patheq}$, the states $b$ and $c$ merge into
$[b] = [c]$, making $\HS/{\patheq} \cong \TC/{\patheq}$.

\begin{rem}[Functorial structure]
\label{sec:functoriality}
The assignment $\Msys \mapsto \ET{\Msys}$ extends to a contravariant $2$-functor
$\mathcal{E}[-] \colon \mathbf{Mwy}^{\mathrm{op}} \to \mathbf{BTop}$,
where $\mathbf{Mwy}$ is the $2$-category of rooted transition systems, simulations,
and path-equivalence homotopies, and $\mathbf{BTop}$ is the $2$-category of
Grothendieck toposes, geometric morphisms, and geometric transformations.
\label{thm:functoriality}
A simulation $f \colon \Msys \to \Nsys$ induces a geometric morphism
$f^* \colon \ET{\Nsys} \to \ET{\Msys}$ by pulling back models along~$f$;
a $2$-cell $\alpha \colon f \Rightarrow g$ induces a geometric
transformation $g^{*} \Rightarrow f^{*}$.  In this language, the
three-level hierarchy encodes three fiber conditions: mutual simulation
means geometric morphisms exist both ways; functional bisimulation means
they compose to equivalences on path-equivalence quotients
(Theorem~\ref{thm:quotient-bridge}); topos equivalence means Morita
equivalence of the geometric theories.  Every such geometric morphism
factors as hyperconnected (adding axioms) followed by localic (expanding
the language); the Second Separation is purely hyperconnected, since the
determinism sequent is an axiom over the same language.
$2$-functoriality follows from the standard fact that pullback of
models along simulations preserves geometric morphism composition, and
that path-equivalence homotopies lift to geometric transformations via
the universal property of the syntactic category.  In summary:
\[
\begin{tikzcd}
  \Msys \arrow[r, "f"] & \Nsys
\end{tikzcd}
\quad\longmapsto\quad
\begin{tikzcd}
  \ET{\Nsys} \arrow[r, "f^*"] & \ET{\Msys}
\end{tikzcd}
\]
\end{rem}

\section{The Bisimulation-Invariant Fragment}
\label{sec:hml}

The hierarchy of Section~\ref{sec:hierarchy} shows that geometric logic sees
structure invisible to bisimulation, but does not explain \emph{why}.  This
section identifies the boundary: the bisimulation-invariant fragment of
geometric logic is exactly diamond-only Hennessy--Milner logic, and three
syntactic mechanisms account for all the ways geometric logic crosses it.

\subsection{Diamond-Only Hennessy--Milner Logic}
\label{sec:hml-def}

\begin{defi}[HML]
\label{def:hml}
The \emph{diamond-only Hennessy--Milner logic}~$\HML$ is generated by:
\[
  \varphi \;::=\; \top \;\mid\; \bot \;\mid\; \varphi_1 \land \varphi_2
  \;\mid\; \varphi_1 \lor \varphi_2 \;\mid\; \Diamond \varphi
\]
The semantics at a state~$s$ in a rooted transition system~$\Msys$ are standard:
$\Msys, s \models \top$ always; $\Msys, s \models \bot$ never;
$\Msys, s \models \varphi_1 \land \varphi_2$ iff both conjuncts hold;
$\Msys, s \models \varphi_1 \lor \varphi_2$ iff at least one disjunct holds;
$\Msys, s \models \Diamond\varphi$ iff there exists~$t$ with $s \to t$ and
$\Msys, t \models \varphi$.
\end{defi}

The \emph{standard translation} embeds~$\HML$ into the positive existential
fragment of first-order logic over the rooted transition language:
\begin{align*}
  \ST_x(\top) &= \top &
  \ST_x(\bot) &= \bot \\
  \ST_x(\varphi_1 \land \varphi_2) &= \ST_x(\varphi_1) \land \ST_x(\varphi_2) &
  \ST_x(\varphi_1 \lor \varphi_2) &= \ST_x(\varphi_1) \lor \ST_x(\varphi_2) \\
  \ST_x(\Diamond\varphi) &= \exists y.\,(\step(x,y) \land \ST_y(\varphi))
\end{align*}
The image of the standard translation has three distinctive properties.
First, it is \emph{equality-free}: no equality symbol appears.  Second, it
has \emph{linear variables}: each existentially bound variable~$y$ appears
in exactly one $\step$-atom as target and in no other atom position.  Third,
it has \emph{tree-shaped dependencies}: the variable dependency graph (where
$x \to y$ means $y$ is introduced in a $\Diamond$ under~$x$) is a forest.
In the terminology of database theory~\cite{yannakakis1981algorithms}, the
standard translation produces exactly the \emph{acyclic conjunctive queries}
over the binary relation~$\step$.

\subsection{Forward Direction: HML Is Bisimulation-Invariant}
\label{sec:hml-invariant}

\begin{thm}[HML Bisimulation Invariance]
\label{thm:hml-invariant}
Let $R$ be a relational bisimulation between rooted transition systems~$\Msys$ and~$\Nsys$,
and let $(s, t) \in R$.  For any~$\HML$ formula~$\varphi$:
\[
\boxed{\quad \Msys, s \models \varphi \;\;\Longleftrightarrow\;\; \Nsys, t \models \varphi \quad}
\]
\end{thm}

\begin{proof}
We prove both directions simultaneously by structural induction on
$\varphi$.  The case $\varphi = \top$ is immediate since $\top$ holds
everywhere.  The case $\varphi = \bot$ is immediate since $\bot$ holds
nowhere.  For $\varphi = \varphi_1 \land \varphi_2$: since $(s,t) \in R$,
the inductive hypothesis gives $\Msys, s \models \varphi_i$ iff
$\Nsys, t \models \varphi_i$ for $i = 1, 2$, hence the conjunction is preserved.
The case $\varphi = \varphi_1 \lor \varphi_2$ is analogous.

For $\varphi = \Diamond\psi$: suppose $\Msys, s \models \Diamond\psi$, so there
exists $s'$ with $s \to s'$ in $\Msys$ and $\Msys, s' \models \psi$.  Since $(s, t) \in R$
and $R$ is a bisimulation, the \emph{forth} condition gives $t'$ with
$t \to t'$ in $\Nsys$ and $(s', t') \in R$.  By the inductive hypothesis applied to
$\psi$ and the pair $(s', t')$, we have $\Nsys, t' \models \psi$, hence
$\Nsys, t \models \Diamond\psi$.  The backward direction uses the \emph{back}
condition symmetrically: given $t'$ with $t \to t'$ in $\Nsys$ and
$\Nsys, t' \models \psi$, the back condition provides $s'$ with $s \to s'$ in $\Msys$
and $(s', t') \in R$, and the inductive hypothesis gives
$\Msys, s' \models \psi$, hence $\Msys, s \models \Diamond\psi$.
\end{proof}

\begin{rem}[Hennessy--Milner converse]
\label{rem:hennessy-milner}
The converse of Theorem~\ref{thm:hml-invariant} (for image-finite systems,
$\HML$-equivalence implies bisimilarity) is the Hennessy--Milner
theorem~\cite{hennessy1985algebraic}.  Our Theorem~\ref{thm:gvb}
addresses a different question: not whether~$\HML$ characterizes
bisimulation, but whether bisimulation-invariance characterizes~$\HML$
within geometric logic.
\end{rem}

\subsection{Three Separating Mechanisms}
\label{sec:mechanisms}

Each mechanism exhibits a geometric property~$P$ expressible as a geometric
sequent, a bisimilar pair $(M, N)$ with $M \models P$ and
$N \not\models P$, and a syntactic explanation of why $P$ lies
outside~$\HML$.

\subsubsection{Mechanism 1: Consequent Equality ($\sigma_{\mathrm{det}}$)}

The \emph{determinism sequent}
$\sigma_{\mathrm{det}} := \step(x,y) \land \step(x,z) \vdash y = z$
asserts that every state has at most one successor.

The witness pair is the hub-spokes system~$\HS$ and the two-cycle~$\TC$,
whose bisimulation was established in Theorem~\ref{thm:second-separation}.
The two-cycle is deterministic: $x$ has unique successor~$y$ and $y$ has
unique successor~$x$, so $\TC \models \sigma_{\mathrm{det}}$.  The
hub-spokes system is not deterministic: $a$ has two distinct successors $b$
and $c$, so $\HS \not\models \sigma_{\mathrm{det}}$.  Therefore
$\sigma_{\mathrm{det}}$ is not bisimulation-invariant.

The syntactic explanation is that the consequent contains $y = z$, an
equality between universally bound variables.  The standard translation
of~$\HML$ never produces equality: each $\Diamond$ introduces a fresh
variable, and no equality symbol appears anywhere in the translation.

\subsubsection{Mechanism 2: Cyclic Variable Sharing ($\sigma_{\mathrm{conf}}$)}

The \emph{weak confluence sequent}
$\sigma_{\mathrm{conf}} := \step(x,y) \land \step(x,z) \vdash
\exists w.\,(\step(y,w) \land \step(z,w))$
asserts that any two successors of a state have a common successor.

The \emph{diamond graph}~$\DG$ has states $\{a, b, c, d\}$, transitions
$a \to b$, $a \to c$, $b \to d$, $c \to d$, $d \to d$, and initial
state~$a$.  Two paths from~$a$ converge at~$d$, forming a diamond shape.
The \emph{confluence tree}~$\CT$ has states $\{a, b, c, d_1, d_2\}$,
transitions $a \to b$, $a \to c$, $b \to d_1$, $c \to d_2$,
$d_1 \to d_1$, $d_2 \to d_2$, and initial state~$a$.  Two paths from~$a$
diverge into separate self-looping sinks.  (Terminal states have self-loops
so that the step relation is total at those states, making weak confluence
well-defined.)

\begin{figure}[ht]
\centering
\begin{tikzpicture}[
    >=Stealth,
    state/.style={circle, draw, minimum size=7mm, inner sep=1pt, font=\small},
    bisim/.style={dashed, shorten >=2pt, shorten <=2pt},
  ]
  \node at (-3.5, 1.5) {$\DG$};
  \node[state] (ga) at (-3.5, 0.7) {$a$};
  \node[state] (gb) at (-4.7, -0.5) {$b$};
  \node[state] (gc) at (-2.3, -0.5) {$c$};
  \node[state] (gd) at (-3.5, -1.7) {$d$};

  \draw[->] (ga) -- (gb);
  \draw[->] (ga) -- (gc);
  \draw[->] (gb) -- (gd);
  \draw[->] (gc) -- (gd);
  \draw[->] (gd) to [out=210, in=250, looseness=5] (gd);

  \node at (3.5, 1.5) {$\CT$};
  \node[state] (ta) at (3.5, 0.7) {$a$};
  \node[state] (tb) at (2.3, -0.5) {$b$};
  \node[state] (tc) at (4.7, -0.5) {$c$};
  \node[state] (td1) at (2.3, -1.7) {$d_1$};
  \node[state] (td2) at (4.7, -1.7) {$d_2$};

  \draw[->] (ta) -- (tb);
  \draw[->] (ta) -- (tc);
  \draw[->] (tb) -- (td1);
  \draw[->] (tc) -- (td2);
  \draw[->] (td1) to [out=210, in=250, looseness=5] (td1);
  \draw[->] (td2) to [out=330, in=290, looseness=5] (td2);

  \draw[bisim] (ga) -- node[above, yshift=4pt, font=\footnotesize] {$R$} (ta);
  \draw[bisim] (gb) to[out=-20, in=200] (tb);
  \draw[bisim] (gc) to[out=15, in=165] (tc);
  \draw[bisim] (gd) to[out=0, in=180] (td1);
  \draw[bisim] (gd) to[out=-20, in=200] (td2);
\end{tikzpicture}
\caption{The diamond graph~$\DG$ and confluence tree~$\CT$ with
bisimulation $R = \{(a,a),(b,b),(c,c),(d,d_1),(d,d_2)\}$.
The merge point~$d$ in~$\DG$ splits into distinct sinks $d_1$, $d_2$
in~$\CT$; both have self-loops but share no common successor.}
\label{fig:dg-ct}
\end{figure}
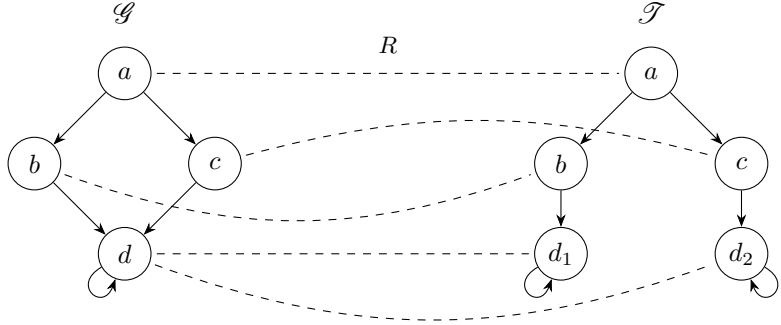

The relation $R = \{(a,a), (b,b), (c,c), (d,d_1), (d,d_2)\}$ is a
relational bisimulation between~$\DG$ and~$\CT$.  We check all five
forth/back pairs; the key pair is $(d,d_1)$, where the diamond's merge
becomes the tree's split.

\smallskip\noindent\emph{Forth.}\enspace
At $(a,a)$: $a \to b$ is matched by $a \to b$ with $(b,b) \in R$;
$a \to c$ by $a \to c$ with $(c,c) \in R$.
At $(b,b)$: $b \to d$ by $b \to d_1$ with $(d,d_1) \in R$.
At $(c,c)$: $c \to d$ by $c \to d_2$ with $(d,d_2) \in R$.
At $(d,d_1)$: $d \to d$ by $d_1 \to d_1$ with $(d,d_1) \in R$.
At $(d,d_2)$: $d \to d$ by $d_2 \to d_2$ with $(d,d_2) \in R$.

\smallskip\noindent\emph{Back.}\enspace
Verified symmetrically: each transition in~$\CT$ from a related state has a
matching transition in~$\DG$.

The diamond graph satisfies weak confluence: at state~$a$, the successors
$b$ and $c$ both reach~$d$ (via $b \to d$ and $c \to d$), so $w = d$ is
the required common successor.  The confluence tree does not satisfy weak
confluence: at state~$a$, the successors $b$ and $c$ reach $d_1$ and $d_2$
respectively, and $d_1$ and $d_2$ each have only a self-loop, so they share
no common successor.  Thus $\sigma_{\mathrm{conf}}$ is not
bisimulation-invariant.

The syntactic explanation is that the variable~$w$ appears as the target
of~$\step$ from two different sources $y$ and $z$, creating a \emph{join
constraint}, a cycle in the variable dependency graph.  In the standard
translation of~$\HML$, each existentially bound variable has exactly one
parent in the dependency graph.

\subsubsection{Mechanism 3: Repeated Variables ($\sigma_{\mathrm{loop}}$)}

The \emph{universal self-loop sequent}
$\sigma_{\mathrm{loop}} := \top \vdash_x \step(x, x)$
asserts that every state has a self-loop.

The \emph{self-loop system}~$\SL$ has a single state~$a$ with transition
$a \to a$ and initial state~$a$.  It is the simplest possible system: one
state with a self-loop.  The two-cycle~$\TC$ was defined in
Section~\ref{sec:second-sep}.

The relation $R = \{(a, x), (a, y)\}$ is a relational bisimulation
between~$\SL$ and~$\TC$.
\emph{Forth:}
at $(a, x)$, $a \to a$ is matched by $x \to y$ with $(a, y) \in R$;
at $(a, y)$, $a \to a$ by $y \to x$ with $(a, x) \in R$.
\emph{Back:}
at $(a, x)$, $x \to y$ is matched by $a \to a$ with $(a, y) \in R$;
at $(a, y)$, $y \to x$ by $a \to a$ with $(a, x) \in R$.

The self-loop system satisfies $\sigma_{\mathrm{loop}}$: the unique
state~$a$ has the self-loop $a \to a$.  The two-cycle does not: neither~$x$
nor~$y$ has a self-loop (the transitions are $x \to y$ and $y \to x$, with
$x \neq y$).  Thus $\sigma_{\mathrm{loop}}$ is not bisimulation-invariant.

The syntactic explanation is that the variable~$x$ appears in \emph{both}
argument positions of~$\step$, an implicit equality between source and
target.  In the standard translation, $\Diamond\varphi$ produces
$\exists y.\,(\step(x, y) \land \cdots)$ where $x$ and $y$ are always
distinct variables.

\subsubsection{Unifying Observation}

\begin{thm}[Three-Mechanism Separation]
\label{thm:three-mechanisms}
There exist three relationally bisimilar pairs of finite systems, each
separated by a geometric property witnessing a distinct syntactic mechanism
outside~$\HML$:
\begin{enumerate}
\item \emph{Consequent equality:} $\HS \bismark \TC$ and
  $\TC \models \sigma_{\mathrm{det}}$,
  $\HS \not\models \sigma_{\mathrm{det}}$;
\item \emph{Cyclic variable sharing:} $\DG \bismark \CT$ and
  $\DG \models \sigma_{\mathrm{conf}}$,
  $\CT \not\models \sigma_{\mathrm{conf}}$;
\item \emph{Repeated variables:} $\SL \bismark \TC$ and
  $\SL \models \sigma_{\mathrm{loop}}$,
  $\TC \not\models \sigma_{\mathrm{loop}}$.
\end{enumerate}
\end{thm}

\noindent
All three mechanisms assert that two structural positions are filled by the
same state.  Consequent equality asserts $y = z$ (two targets coincide).
Cyclic sharing asserts that two paths meet (two sources share a target).
Repeated variables assert $x = y$ (source equals target).  Bisimulation can
always duplicate a state into distinct copies, breaking any such
identification constraint.  Diamond-only~$\HML$ avoids all three: each
$\Diamond$ introduces a fresh variable connected to exactly one parent atom,
with no equality, no sharing, and no repetition.

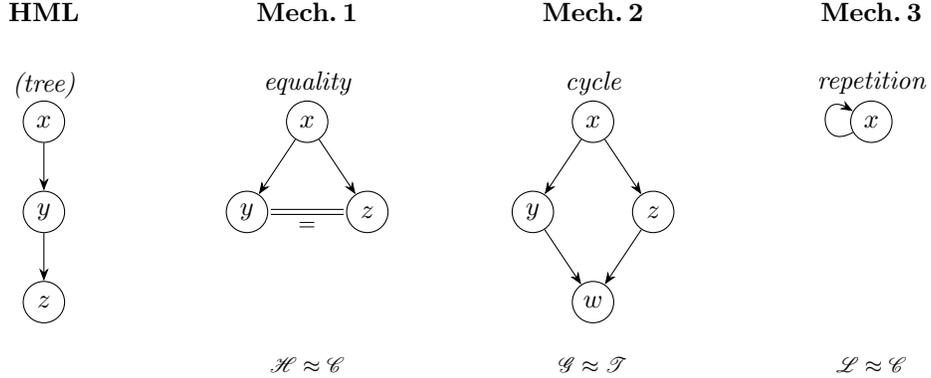
\begin{figure}[ht]
\centering
\begin{tikzpicture}[
    >=Stealth,
    var/.style={circle, draw, minimum size=5.5mm, inner sep=0pt,
                font=\small},
    eqlabel/.style={font=\scriptsize, midway},
    steplabel/.style={font=\scriptsize, midway, right},
    mechlabel/.style={font=\small\bfseries, anchor=south},
    sublabel/.style={font=\small\itshape, anchor=north},
  ]
  \node[mechlabel] at (-5.5, 2.4) {HML};
  \node[sublabel]  at (-5.5, 2.0) {(tree)};
  \node[var] (hx) at (-5.5, 1.2) {$x$};
  \node[var] (hy) at (-5.5, 0)   {$y$};
  \node[var] (hz) at (-5.5,-1.2) {$z$};
  \draw[->] (hx) -- (hy);
  \draw[->] (hy) -- (hz);

  \node[mechlabel] at (-2, 2.4) {Mech.\,1};
  \node[sublabel]  at (-2, 2.0) {equality};
  \node[var] (m1x) at (-2, 1.2)   {$x$};
  \node[var] (m1y) at (-2.8, 0)   {$y$};
  \node[var] (m1z) at (-1.2, 0)   {$z$};
  \draw[->] (m1x) -- (m1y);
  \draw[->] (m1x) -- (m1z);
  \draw[double, double distance=1.5pt, shorten >=1pt, shorten <=1pt]
    (m1y) -- node[below, font=\scriptsize] {$=$} (m1z);

  \node[mechlabel] at (1.8, 2.4) {Mech.\,2};
  \node[sublabel]  at (1.8, 2.0) {cycle};
  \node[var] (m2x) at (1.8, 1.2)  {$x$};
  \node[var] (m2y) at (1.0, 0)    {$y$};
  \node[var] (m2z) at (2.6, 0)    {$z$};
  \node[var] (m2w) at (1.8,-1.2)  {$w$};
  \draw[->] (m2x) -- (m2y);
  \draw[->] (m2x) -- (m2z);
  \draw[->] (m2y) -- (m2w);
  \draw[->] (m2z) -- (m2w);

  \node[mechlabel] at (5.5, 2.4) {Mech.\,3};
  \node[sublabel]  at (5.5, 2.0) {repetition};
  \node[var] (m3x) at (5.5, 1.2)  {$x$};
  \draw[->] (m3x) to [out=210, in=150, looseness=5] (m3x);

  \node[font=\scriptsize, anchor=north] at (-2, -1.8)
    {$\HS \bismark \TC$};
  \node[font=\scriptsize, anchor=north] at (1.8, -1.8)
    {$\DG \bismark \CT$};
  \node[font=\scriptsize, anchor=north] at (5.5, -1.8)
    {$\SL \bismark \TC$};
\end{tikzpicture}
\caption{Variable dependency graphs for the three separating mechanisms.
HML produces tree-shaped dependencies (left); each mechanism breaks the
tree in a distinct way.  Edges represent $\step$ atoms; the double line
in Mechanism~1 represents $y = z$.  Witness pairs are shown below each
mechanism (see Figures~\ref{fig:sl-tc}--\ref{fig:hs-tc}
and the DG/CT diagram above).}
\label{fig:three-mechanisms}
\end{figure}

\subsection{The Geometric van~Benthem Theorem}
\label{sec:gvb-theorem}

\begin{thm}[Geometric van~Benthem]
\label{thm:gvb}
A geometric sequent $\phi \vdash_{\vec{x}} \psi$ over a relational
signature is bisimulation-invariant if and only if $\psi$ is equivalent
(over models of~$\phi$) to a disjunction of formulas in the image of the
standard translation of diamond-only~$\HML$:
\[
\boxed{\quad \phi \vdash_{\vec{x}} \psi \;\text{is bisim-invariant}
  \quad\Longleftrightarrow\quad
  \psi \equiv \textstyle\bigvee_i \ST_x(\varphi_i) \quad}
\]
\end{thm}

The forward direction (Theorem~\ref{thm:hml-invariant}) is proved
constructively.  The backward direction---that every
bisimulation-invariant geometric formula is equivalent to one in
$\HML$'s image---is established via bounded tree unraveling
(Section~\ref{sec:tree-unraveling}).  A characteristic formula approach
decomposes the transfer into three cases and resolves all constructively:
the characteristic $\HML$ formula~$\chi_{T,v}$ for tree unravelings
reduces geometric satisfaction to $\HML$ satisfaction via geo-tree
decomposition, and simulation invariance of positive $\HML$ closes the
TRUE-on-tree cases.  The remaining case, where $\phi$ is false on both
graph and tree, follows from the contrapositive of the transfer and
requires no additional geometric content.

This theorem targets an open cell in the preservation-theorem landscape.
Van Benthem's original theorem~\cite{vanbenthem1976modal} identifies modal
logic as the bisimulation-invariant fragment of full first-order logic;
Rosen~\cite{rosen1997modal} extends this to finite structures.
Celani--Jansana~\cite{celani1999priestley} characterize positive modal logic
($\Diamond$, $\Box$, no negation) as the bisimulation-invariant fragment of
positive first-order logic, and Rossman~\cite{rossman2008homomorphism} shows
that existential positive logic is the homomorphism-preserved fragment over
finite structures.  The geometric analogue, that diamond-only~$\HML$ is
the bisimulation-invariant fragment of geometric logic, is the content of
Theorem~\ref{thm:gvb}.

The characterization connects to conjunctive query
theory~\cite{yannakakis1981algorithms}: the tree-shaped fragment corresponds
to the \emph{acyclic conjunctive queries}, where acyclicity ensures
polynomial-time evaluation.  These three mechanisms are exactly the
three ways a conjunctive query can fail to be acyclic with respect to a
binary relation.

There is a structural reason to expect exactly three mechanisms and no
others.  Every geometric formula over a binary relation~$\step$ and one free
variable~$x$ is built from atoms $\step(t_1, t_2)$ and equalities
$t_1 = t_2$, where $t_1, t_2$ range over variables.  For such a formula to
lie outside the image of the standard translation, it must violate at least
one of the three properties characterizing that image: equality-freeness,
linearity of variables, or tree-shaped dependencies.  Violating
equality-freeness means the formula contains $t_1 = t_2$; when both
variables are universally bound this is Mechanism~1 (consequent equality),
and when $t_1 = t_2$ identifies source and target of a step atom this is
Mechanism~3 (repeated variables).  Violating the tree-shaped dependency
condition means the variable dependency graph contains a cycle, which is
Mechanism~2 (cyclic variable sharing).  Violating linearity (a variable
appearing as target in two step atoms) forces either a cycle (reducing to
Mechanism~2) or an equality (reducing to Mechanism~1 or~3).  The gap between
this syntactic observation and the full proof is the
formula-level decomposition: one must show that bisimulation-invariance of a
complex formula forces its equivalence to an~$\HML$ formula, not merely that
individual non-$\HML$ atoms are separated.

\begin{cor}[Mechanism completeness]
\label{cor:mechanism-completeness}
The three mechanisms (consequent equality, cyclic variable sharing, and
repeated variables) exhaust all non-$\HML$ geometric atoms at every
quantifier depth: for each depth~$d$, every non-bisimulation-invariant
depth-$d$ geometric atom over one free variable and a binary relation is
separated by a witness pair exercising one of the three mechanisms.
\end{cor}

This follows from Theorem~\ref{thm:gvb}: if every bisimulation-invariant
geometric formula decomposes into~$\HML$, then every non-$\HML$ atom must
be non-bisimulation-invariant.  The converse does not hold; the corollary
asserts separation of individual atoms, not the decomposition of arbitrary
formulas.  Independent verification through
depth~2 (Theorems~\ref{thm:depth0}--\ref{thm:depth2}) provides
constructive witnesses.

\subsection{Confirmatory Depth-Bounded Instances}
\label{sec:bounded}

The following depth-bounded instances confirm Theorem~\ref{thm:gvb}
by exhaustive verification at low depths, providing independent
evidence for the general result.  Define the \emph{diamond
depth} of an~$\HML$ formula by
$\mathrm{depth}(\top) = \mathrm{depth}(\bot) = 0$,
$\mathrm{depth}(\varphi_1 \circ \varphi_2) =
\max(\mathrm{depth}(\varphi_1), \mathrm{depth}(\varphi_2))$ for
$\circ \in \{\land, \lor\}$, and
$\mathrm{depth}(\Diamond\varphi) = \mathrm{depth}(\varphi) + 1$.

At each depth, the geometric atoms over one free variable~$x$ and the
binary relation~$\step$ are finitely enumerable.  We show that at depths~0
and~1, every non-$\HML$ atom is \emph{not} bisimulation-invariant, while
every~$\HML$ atom \emph{is}, confirming the theorem at these depths.

\subsubsection{Depth 0}

\begin{thm}[Depth-0 van~Benthem]
\label{thm:depth0}
Every depth-$0$~$\HML$ formula is constant ($\top$ or~$\bot$ at every
state of every system), and $\step(x, x)$ is not bisimulation-invariant.
\end{thm}

\begin{proof}
We prove constancy by structural induction on the formula.  The formula
$\top$ is constantly true at every state of every system.  The formula $\bot$
is constantly false.  If $\varphi$ and $\psi$ are both constant, their
conjunction $\varphi \land \psi$ is constant: if both are constantly true,
the conjunction is constantly true; if either is constantly false, the
conjunction is constantly false.  Disjunction is analogous.  The diamond
case is vacuous: a formula of depth~0 cannot have a diamond as its outermost
connective, since $\mathrm{depth}(\Diamond\varphi) = \mathrm{depth}(\varphi)
+ 1 \geq 1$.

For the separation, the self-loop system~$\SL$ and the two-cycle~$\TC$
provide the witness.  The bisimulation $R = \{(a, x), (a, y)\}$ was
verified in Section~\ref{sec:mechanisms} (Mechanism~3).  In~$\SL$, the
state~$a$ satisfies $\step(a, a)$.  In~$\TC$, the state~$x$ does not
satisfy $\step(x, x)$ since $x$'s only transition is $x \to y$ with
$y \neq x$.  Since $(a, x) \in R$, the self-loop property $\step(x,x)$ is
not bisimulation-invariant.
\end{proof}

Since $\step(x,x)$ is the only non-trivial depth-0 geometric atom over one
free variable and the binary relation~$\step$ (the only conjunction of atoms
from $\{\step(x,x)\}$ beyond the empty conjunction~$\top$), and since every
depth-0~$\HML$ formula is constant, the bisimulation-invariant depth-0
geometric formulas are exactly the constants $\top$ and $\bot$, which are
exactly the depth-0~$\HML$ formulas.

\subsubsection{Depth 1}

At depth~1, the only~$\HML$ atom beyond constants is
$\Diamond\top = \exists y.\,\step(x, y)$, asserting that~$x$ has a
successor.  The depth-1 geometric atoms over $\{x\}$ and~$\step$ include
all properties of the form $\exists y.\,\alpha(x, y)$ where $\alpha$ is a
conjunction of atoms from
$\{\step(x, y),\allowbreak \step(y, x),\allowbreak \step(x, x),\allowbreak \step(y, y),\allowbreak x = y\}$.  Beyond
$\Diamond\top = \exists y.\,\step(x, y)$, each such atom involves one of: a
self-loop ($\step(x,x)$ or $\step(y,y)$), a back-edge ($\step(y,x)$), or
an equality ($x = y$).

To separate the atoms involving back-edges, we pair the path~$\PT$
from Section~\ref{sec:first-sep} (states $\{x, y\}$, transitions
$x \to y$ and $y \to y$) with a new system.  The
\emph{back-edge graph}~$\BG$ has states
$\{\mathsf{rt}, \mathsf{lp}\}$, transitions
$\mathsf{rt} \to \mathsf{lp}$, $\mathsf{lp} \to \mathsf{lp}$, and
$\mathsf{lp} \to \mathsf{rt}$, and initial state~$\mathsf{rt}$: the
path~$\PT$ with one added back-edge.

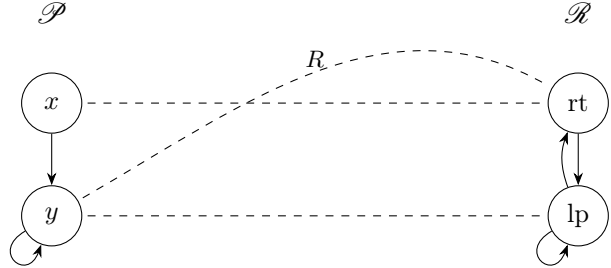
\begin{figure}[ht]
\centering
\begin{tikzpicture}[
    >=Stealth,
    state/.style={circle, draw, minimum size=8mm, inner sep=1pt, font=\small},
    bisim/.style={dashed, shorten >=2pt, shorten <=2pt},
  ]
  \node at (-3.5, 1.2) {$\PT$};
  \node[state] (x) at (-3.5, 0) {$x$};
  \node[state] (y) at (-3.5, -1.5) {$y$};

  \draw[->] (x) -- (y);
  \draw[->] (y) to [out=210, in=250, looseness=5] (y);

  \node at (3.5, 1.2) {$\BG$};
  \node[state] (root) at (3.5, 0) {rt};
  \node[state] (loop) at (3.5, -1.5) {lp};

  \draw[->] (root) -- (loop);
  \draw[->] (loop) to [out=210, in=250, looseness=5] (loop);
  \draw[->, bend left=20] (loop) to (root);

  \draw[bisim] (x) -- node[above, yshift=10pt, font=\footnotesize] {$R$} (root);
  \draw[bisim] (y) -- (loop);
  \draw[bisim] (y) to[out=30, in=150] (root);
\end{tikzpicture}
\caption{The path~$\PT$ and back-edge graph~$\BG$ with
bisimulation $R = \{(x, \mathsf{rt}), (y, \mathsf{lp}), (y, \mathsf{rt})\}$.
The back-edge $\mathsf{lp} \to \mathsf{rt}$ in~$\BG$ is absorbed by the
self-loop $y \to y$ in~$\PT$.}
\label{fig:pt-bg}
\end{figure}

The relation
$R = \{(x, \mathsf{rt}), (y, \mathsf{lp}),
(y, \mathsf{rt})\}$ is a relational bisimulation
between~$\PT$ and~$\BG$.  The key
pair is $(y, \mathsf{lp})$, where the back-edge
$\mathsf{lp} \to \mathsf{rt}$ must be absorbed by the
self-loop~$y \to y$.

\smallskip\noindent\emph{Forth.}\enspace
At $(x, \mathsf{rt})$:
$x \to y$ is matched by
$\mathsf{rt} \to \mathsf{lp}$ with
$(y, \mathsf{lp}) \in R$.
At $(y, \mathsf{lp})$:
$y \to y$ by
$\mathsf{lp} \to \mathsf{lp}$ with
$(y, \mathsf{lp}) \in R$.
At $(y, \mathsf{rt})$:
$y \to y$ by
$\mathsf{rt} \to \mathsf{lp}$ with
$(y, \mathsf{lp}) \in R$.

\smallskip\noindent\emph{Back.}\enspace
At $(x, \mathsf{rt})$:
$\mathsf{rt} \to \mathsf{lp}$ is matched by
$x \to y$ with
$(y, \mathsf{lp}) \in R$.
At $(y, \mathsf{lp})$:
$\mathsf{lp} \to \mathsf{lp}$ by
$y \to y$ with
$(y, \mathsf{lp}) \in R$;
$\mathsf{lp} \to \mathsf{rt}$ by
$y \to y$ with
$(y, \mathsf{rt}) \in R$.
At $(y, \mathsf{rt})$:
$\mathsf{rt} \to \mathsf{lp}$ by
$y \to y$ with
$(y, \mathsf{lp}) \in R$.

\begin{thm}[Depth-1 van~Benthem]
\label{thm:depth1}
Among depth-$1$ geometric atoms, the only bisimulation-invariant one is
$\Diamond\top$.  Specifically, $\exists y.\,\step(x,y)$ is
bisimulation-invariant, and the following five properties are not
bisimulation-invariant: the self-loop property $\step(x,x)$, the
has-predecessor property $\exists y.\,\step(y,x)$, the
exists-self-loop property $\exists y.\,\step(y,y)$, the mutual-neighbor
property $\exists y.\,(\step(x,y) \land \step(y,x))$, and the
successor-with-self-loop property
$\exists y.\,(\step(x,y) \land \step(y,y))$.
\end{thm}

\begin{proof}
The has-successor property $\exists y.\,\step(x,y)$ is the standard
translation of~$\Diamond\top$, so it is bisimulation-invariant by
Theorem~\ref{thm:hml-invariant}.  Concretely: if $(s,t) \in R$ for a
relational bisimulation~$R$ and $s$ has a successor~$s'$, the forth
condition gives $t'$ with $t \to t'$ and $(s',t') \in R$, so $t$ has a
successor.  The back condition gives the reverse.

The self-loop property $\step(x,x)$ is not bisimulation-invariant: the
bisimulation $R = \{(a,x), (a,y)\}$ between~$\SL$ and~$\TC$ (verified in
Section~\ref{sec:mechanisms}) relates $a$ to~$x$, and $\SL \models
\step(a,a)$ while $\TC \not\models \step(x,x)$.

The has-predecessor property $\exists y.\,\step(y,x)$ is not
bisimulation-invariant: the bisimulation~$R$ between~$\PT$ and~$\BG$
relates $x$ to~$\mathsf{rt}$.  In~$\PT$, the state~$x$ has no
predecessor (no transition has $x$ as target: the only transitions are
$x \to y$ and $y \to y$).  In~$\BG$, the state~$\mathsf{rt}$
has predecessor~$\mathsf{lp}$ (via the transition
$\mathsf{lp} \to \mathsf{rt}$).

The exists-self-loop property $\exists y.\,\step(y,y)$ is not
bisimulation-invariant: the bisimulation between~$\SL$ and~$\TC$ relates
$a$ to~$x$.  In~$\SL$, the witness $y = a$ satisfies $\step(a,a)$.
In~$\TC$, neither~$x$ nor~$y$ has a self-loop, so no witness exists.

The mutual-neighbor property
$\exists y.\,(\step(x,y) \land \step(y,x))$ is not
bisimulation-invariant: the bisimulation between~$\PT$ and~$\BG$ relates
$x$ to~$\mathsf{rt}$.  In~$\BG$, the state~$\mathsf{rt}$
has mutual neighbor~$\mathsf{lp}$: $\mathsf{rt} \to \mathsf{lp}$ and
$\mathsf{lp} \to \mathsf{rt}$.  In~$\PT$, the state~$x$
has only one successor~$y$, and $y$ does not
transition back to~$x$ (its only transition is the self-loop
$y \to y$).

The successor-with-self-loop property
$\exists y.\,(\step(x,y) \land \step(y,y))$ is not
bisimulation-invariant: the bisimulation between~$\SL$ and~$\TC$ relates
$a$ to~$x$.  In~$\SL$, the witness $y = a$ satisfies $\step(a,a) \land
\step(a,a)$.  In~$\TC$, the state~$x$ has unique successor~$y$, and $y$
does not have a self-loop (its only transition is $y \to x$).
\end{proof}

\noindent
In particular, $\Diamond\top$ is the \emph{unique} bisimulation-invariant
depth-1 atom: every depth-1 geometric atom is bisimulation-invariant if and
only if it is equivalent to~$\Diamond\top$.

\begin{rem}[Depth-1 coverage]
\label{rem:coverage}
The coverage argument is as follows.  Every depth-1 geometric atom
$\exists y.\,\alpha(x, y)$ where $\alpha$ is not simply $\step(x, y)$ must
include at least one additional conjunct from
$\{\step(x,x),\allowbreak \step(y,y),\allowbreak \step(y,x),\allowbreak x = y\}$.
The atoms involving $\step(x,x)$, $\step(y,y)$, or
$x = y$ are all separated by the self-loop/two-cycle witness pair (since
$x = y$ reduces $\exists y.\,(\step(x,y) \land x = y)$ to $\step(x,x)$,
which is already covered).  The atoms involving $\step(y,x)$ are separated
by the path/back-edge witness pair.  Thus the five atoms above exhaust all
non-$\HML$ depth-1 phenomena.
\end{rem}

\subsubsection{Depth 2}

At depth~2, the geometric atoms over one free variable~$x$ and the binary
relation~$\step$ have the form $\exists y.\,\exists z.\,\alpha(x,y,z)$
or $\step(x,x) \land \exists y.\,\beta(x,y)$, where $\alpha$ and~$\beta$
are conjunctions of step-atoms and equalities.  The only depth-2~$\HML$
formula (up to equivalence) beyond constants and~$\Diamond\top$ is
$\Diamond\Diamond\top = \exists y.\,(\step(x,y) \land \exists z.\,\step(y,z))$,
asserting that~$x$ has a successor that itself has a successor.

\begin{thm}[Depth-2 van~Benthem]
\label{thm:depth2}
Among depth-$2$ geometric atoms, the only bisimulation-invariant one is
$\Diamond\Diamond\top$.  Specifically, seven non-$\HML$ atoms are not
bisimulation-invariant:
\begin{enumerate}
\item $\exists y.\,(\step(x,y) \land \step(y,y))$
  \hfill (successor has self-loop)
\item $\exists y.\,\exists z.\,(\step(x,y) \land \step(y,z) \land \step(z,y))$
  \hfill (successor on 2-cycle)
\item $\exists y.\,\exists z.\,(\step(x,y) \land \step(y,z) \land \step(z,z))$
  \hfill (successor reaches self-loop)
\item $\exists y.\,\exists z.\,(\step(x,y) \land \step(y,z) \land \step(x,z))$
  \hfill (successor has co-successor)
\item $\step(x,x) \land \exists y.\,\step(x,y)$
  \hfill (self-loop and has successor)
\item $\step(x,x) \land \exists y.\,(\step(x,y) \land \exists z.\,\step(y,z))$
  \hfill (self-loop and successor has successor)
\item $\exists y.\,\exists z.\,(\step(x,y) \land \step(y,z) \land \step(z,x))$
  \hfill (successor reaches back)
\end{enumerate}
\end{thm}

\begin{proof}
The has-successor-has-successor property $\Diamond\Diamond\top$ is the
standard translation of $\Diamond\Diamond\top$ in~$\HML$, so it is
bisimulation-invariant by Theorem~\ref{thm:hml-invariant}.

For the separation, two witness pairs suffice.  The self-loop~$\SL$ and
two-cycle~$\TC$ pair (from Section~\ref{sec:mechanisms}) separates atoms
(1) and (3)--(7): the self-loop has $\step(a,a)$ while the two-cycle has
no self-loops, and the branching structure of the self-loop ($a$ reaches
only~$a$) differs from the two-cycle ($x$ reaches $y$ and~$x$).

Atom~(2) requires a new witness pair.  The \emph{cycle-entry
system}~$\TC_e$ has states $\{a, b, c\}$ with transitions $a \to b$,
$b \to c$, $c \to b$, and initial state~$a$: state~$a$ leads into the
2-cycle $b \leftrightarrow c$.  The \emph{stretched-entry
system}~$\SE$ has states $\{a, b, c, d\}$ with transitions $a \to b$,
$b \to c$, $c \to d$, $d \to c$, and initial state~$a$: state~$a$
leads through an intermediate path into the 2-cycle
$c \leftrightarrow d$.

The relation
$R = \{(a,a)\} \cup (\{b,c\} \times \{b,c,d\})$
is a relational bisimulation between~$\TC_e$ and~$\SE$.
\emph{Forth:} At~$(a,a)$: $a \to b$ matched by $a \to b$ with
$(b,b) \in R$.  For any $(s,t)$ with $s \in \{b,c\}$ and
$t \in \{b,c,d\}$: the unique transition from~$s$ (either $b \to c$
or $c \to b$) is matched by a transition from~$t$---namely $b \to c$,
$c \to d$, or $d \to c$---and the target pair lies
in~$\{b,c\} \times \{b,c,d\}$.
\emph{Back:} At~$(a,a)$: $a \to b$ in~$\SE$ matched by $a \to b$ with
$(b,b) \in R$.  For any $(s,t)$ with $s \in \{b,c\}$ and
$t \in \{b,c,d\}$: the unique transition from~$t$ (either $b \to c$,
$c \to d$, or $d \to c$) is matched by $b \to c$ or $c \to b$
from~$s$, and the target pair again lies in~$\{b,c\} \times \{b,c,d\}$.

In~$\TC_e$, starting from~$a$, the successor~$b$ has successor~$c$ with
$c \to b$, so atom~(2) is satisfied.  In~$\SE$, starting from~$a$, the
successor~$b$ has unique successor~$c$, but $c$'s only successor is~$d
\neq b$, so atom~(2) is not satisfied.
\end{proof}

\begin{figure}[ht]
\centering
\begin{tikzpicture}[
    >=Stealth,
    state/.style={circle, draw, minimum size=8mm, inner sep=1pt, font=\small},
    bisim/.style={dashed, shorten >=2pt, shorten <=2pt},
  ]
  \node at (-3.5, 1.2) {$\TC_e$};
  \node[state] (a1) at (-3.5, 0) {$a$};
  \node[state] (b1) at (-3.5, -1.5) {$b$};
  \node[state] (c1) at (-3.5, -3.0) {$c$};

  \draw[->] (a1) -- (b1);
  \draw[->, bend left=20] (b1) to (c1);
  \draw[->, bend left=20] (c1) to (b1);

  \node at (3.5, 1.2) {$\SE$};
  \node[state] (a2) at (3.5, 0) {$a$};
  \node[state] (b2) at (3.5, -1.5) {$b$};
  \node[state] (c2) at (3.5, -3.0) {$c$};
  \node[state] (d2) at (3.5, -4.5) {$d$};

  \draw[->] (a2) -- (b2);
  \draw[->] (b2) -- (c2);
  \draw[->, bend left=20] (c2) to (d2);
  \draw[->, bend left=20] (d2) to (c2);

  \draw[bisim] (a1) -- node[above, font=\footnotesize] {$R$} (a2);
  \draw[bisim] (b1) -- (b2);
  \draw[bisim] (c1) -- (c2);
  \draw[bisim] (b1) to[out=-15, in=195] (c2);
  \draw[bisim] (c1) to[out=15, in=165] (b2);
  \draw[bisim] (b1) to[out=-25, in=165] (d2);
  \draw[bisim] (c1) to[out=-15, in=195] (d2);
\end{tikzpicture}
\caption{The cycle-entry system~$\TC_e$ and stretched-entry system~$\SE$
with bisimulation $R = \{(a,a)\} \cup (\{b,c\} \times \{b,c,d\})$.
Both systems are bisimilar, but atom~(2) $\exists y.\,(\step(x,y) \wedge
\exists z.\,(\step(y,z) \wedge \step(z,y)))$ holds in~$\TC_e$ and fails
in~$\SE$.}
\label{fig:tce-se}
\end{figure}
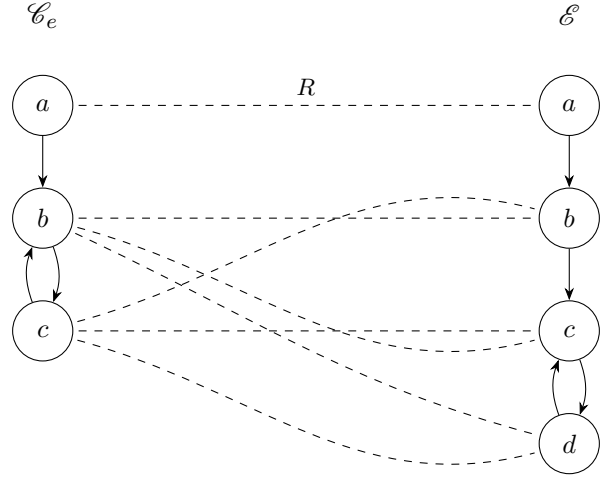

\noindent
The Lean formalization enumerates all structurally distinct depth-2 atoms
(eight total, after collapsing atoms that factor into depth-0/1 properties
or are trivially implied by others) and verifies each one; we list the
seven non-$\HML$ cases above.  In particular, $\Diamond\Diamond\top$ is
the \emph{unique} bisimulation-invariant depth-2 atom.

\begin{rem}[Depth-2 mechanisms]
\label{rem:depth2-mechanisms}
The depth-2 analysis confirms that Mechanisms~2 and~3 are both active
at this depth (Mechanism~1---explicit consequent equality---does not
arise among depth-2 step-atoms).
The self-loop/two-cycle pair exercises Mechanism~3 (repeated
variables, via atoms~(1), (3), and~(5)--(6): each contains $\step(y,y)$,
$\step(z,z)$, or $\step(x,x)$, a variable repeated in both positions of
a step atom).  The cycle-entry/stretched-entry pair exercises
Mechanism~2 (cyclic variable sharing, via the back-edge $z \to y$ in
atom~(2)).  Atom~(7) is separated by the back-edge $z \to x$, also an
instance of Mechanism~2.  Atom~(4) involves a diamond pattern
$x \to y \to z$ with $x \to z$, creating variable sharing between the source
and a depth-2 descendant.
\end{rem}

\subsection{The Three Mechanisms as Intensional Noise}
\label{sec:intensional-noise}

The three separating mechanisms have a computational interpretation
beyond their role as separation witnesses.

Recall that a geometric formula over a binary relation~$\step$ and free
variable~$x$ lies outside the image of the standard translation of
diamond-only~$\HML$ if and only if it violates at least one of three
conditions: equality-freeness, tree-shaped variable dependencies, and
variable linearity.  These mechanisms are exactly the three ways to
violate these conditions.

The unifying observation of Section~\ref{sec:mechanisms}---that all three
mechanisms assert coincidence of two structural positions---has a precise
categorical meaning: each mechanism corresponds to an equalizer condition in
the syntactic category.  Bisimulation, which can always duplicate states
into distinct copies (forth condition), systematically breaks equalizer
constraints.  Diamond-only~$\HML$, which introduces a fresh variable at each
$\Diamond$, systematically avoids them.  Theorem~\ref{thm:gvb-formalized} establishes the semantic content
($d$-bisim-invariance); the Geometric van~Benthem
Theorem (Theorem~\ref{thm:gvb}) confirms the syntactic
characterization: these three mechanisms \emph{exhaust} the
intensional noise, i.e., every geometric property not captured by~$\HML$
factors through one of them.
The complexity landscape is relevant: deciding characteristic formulae
modulo $n$-nested simulation is coNP-complete for $n=2$ and
PSPACE-complete for $n \geq 3$~\cite{aceto2025characteristic},
suggesting that scaling the bounded verification requires
structural arguments rather than brute enumeration.

\medskip
The HML characterization and its three mechanisms provide the geometric
language for the spectrum.

\subsection{Proof via Tree Unraveling}
\label{sec:tree-unraveling}

The bounded-depth verifications of Section~\ref{sec:bounded} provide
independent confirmation at depths~0--2 by exhaustive case analysis.
The general proof uses different machinery.  We describe the formalized
proof strategy based on \emph{bounded tree unraveling} that establishes
Theorem~\ref{thm:gvb} at arbitrary depth.

\subsubsection*{Bounded tree unraveling}
For a pointed $\mathsf{FinLTS}$~$(G, v)$ and depth bound~$d$, the
\emph{bounded tree unraveling} $\mathcal{T}_d(G, v)$ has as vertices the
computation paths of length~$\leq d$ starting from~$v$: finite sequences
$[(a_1, w_1), \ldots, (a_k, w_k)]$ where each
$G.\mathsf{edge}(a_i, w_{i-1}, w_i)$ holds.  An edge labeled~$a$ from
path~$p$ to path~$p \mathbin{+\!\!+} [(a, w)]$ exists when
$G.\mathsf{edge}(a, \mathsf{endpoint}(p), w)$.  The projection
$\pi \colon \mathcal{T}_d(G, v) \to G$ sending each path to its endpoint
is an LTS homomorphism satisfying the back condition at non-maximal depth:
for every path~$p$ with $|p| < d$ and every $a$-edge from
$\mathsf{endpoint}(p)$ in~$G$, there exists a matching edge in the tree.

\begin{figure}[ht]
\centering
\begin{tikzpicture}[
    >=Stealth,
    state/.style={circle, draw, minimum size=7mm, inner sep=1pt,
                  font=\small},
    tnode/.style={circle, draw, minimum size=7mm, inner sep=0pt,
                  font=\scriptsize},
    proj/.style={->, dashed, shorten >=2pt, shorten <=2pt,
                 font=\scriptsize},
  ]
  \node at (-4.5, 1.8) {$\DG$};
  \node[state] (ga) at (-4.5, 0.7) {$a$};
  \node[state] (gb) at (-5.7,-0.5) {$b$};
  \node[state] (gc) at (-3.3,-0.5) {$c$};
  \node[state] (gd) at (-4.5,-1.7) {$d$};

  \draw[->] (ga) -- (gb);
  \draw[->] (ga) -- (gc);
  \draw[->] (gb) -- (gd);
  \draw[->] (gc) -- (gd);
  \draw[->] (gd) to [out=210, in=250, looseness=5] (gd);

  \node at (3.5, 1.8) {$\mathcal{T}_2(\DG, a)$};
  \node[tnode] (te) at (3.5, 0.7) {$\varepsilon$};
  \node[tnode] (tb) at (2.0,-0.5) {$b$};
  \node[tnode] (tc) at (5.0,-0.5) {$c$};
  \node[tnode] (tbd) at (2.0,-1.7) {$bd$};
  \node[tnode] (tcd) at (5.0,-1.7) {$cd$};

  \draw[->] (te) -- (tb);
  \draw[->] (te) -- (tc);
  \draw[->] (tb) -- (tbd);
  \draw[->] (tc) -- (tcd);

  \draw[proj] (te)  to[out=160, in=20]
    node[above] {$\pi$} (ga);
  \draw[proj] (tb)  to[out=200, in=-20] (gb);
  \draw[proj] (tc)  to[out=165, in=15] (gc);
  \draw[proj] (tbd) to[out=180, in=0] (gd);
  \draw[proj] (tcd) to[out=200, in=-20] (gd);
\end{tikzpicture}
\caption{Bounded tree unraveling $\mathcal{T}_2(\DG, a)$.  Vertices are
computation paths from~$a$ of length~${\leq}\,2$; edges extend paths by
one step.  The projection~$\pi$ sends each path to its endpoint.  The
confluence point~$d$ is resolved into distinct tree vertices $bd$
and~$cd$, recovering the confluence tree~$\CT$ (cf.\ Mechanism~2).
Depth-2 vertices have no outgoing edges: the depth bound truncates the
self-loop at~$d$.}
\label{fig:tree-unraveling}
\end{figure}
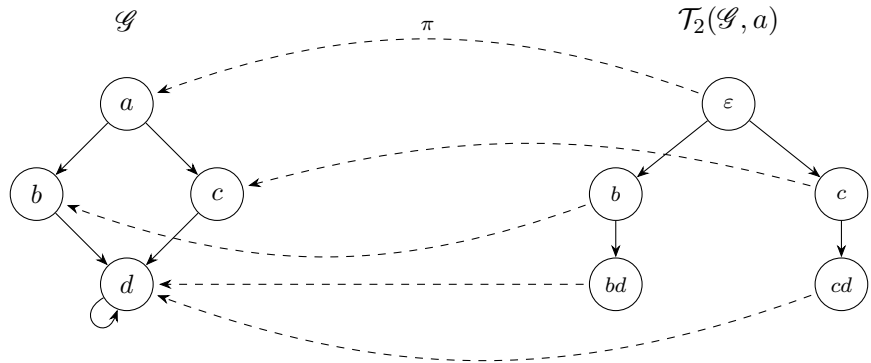

\subsubsection*{Depth preservation}
The projection's homomorphism and back properties yield bidirectional
HML transfer: for every $\HML$ formula~$\psi$ with
$\mathrm{depth}(\psi) \leq d$,
\[
\psi\;\text{holds at}\;(G, v) \quad\Longleftrightarrow\quad
\psi\;\text{holds at}\;(\mathcal{T}_d(G,v), \mathsf{root}).
\]
The forward direction uses the projection homomorphism; the backward
direction is proved by structural induction on~$\psi$, using the back
condition at each diamond step with the invariant that path length plus
formula depth stays within the bound.

\subsubsection*{Equality elimination on trees}
On rooted trees, geometric formula atoms that go beyond~$\HML$ collapse.
Self-loops are impossible (edges strictly increase path length), equality
between step source and target is absurd (combining the self-loop
impossibility), and confluence patterns (two edges arriving at the same
vertex) force their sources to be identical.  These are the tree-level
manifestations of the three separating mechanisms: on tree-like structures,
all three mechanisms become unsatisfiable.

\subsubsection*{The main theorem}

Define two pointed $\mathsf{FinLTS}$ to be \emph{$d$-bisimilar} if they
satisfy the same $\HML$ formulas of depth~$\leq d$:
\[
(G, v) \bismark_d (H, w) \quad\stackrel{\mathrm{def}}{\iff}\quad
\forall\, \psi : \HML,\; \mathrm{depth}(\psi) \leq d \;\Longrightarrow\;
(\psi \text{ at } G, v \iff \psi \text{ at } H, w).
\]
By depth preservation, $\mathcal{T}_d(G,v)$ is $d$-bisimilar to~$(G,v)$.
The geometric van~Benthem theorem then states:

\begin{thm}[Geometric van~Benthem, formalized]
\label{thm:gvb-formalized}
Let $\phi$ be a bisimulation-invariant geometric formula over the LTS
signature with existential depth~$\leq d$.  Then $\phi$ is
$d$-bisim-invariant: for all $(G, v) \bismark_d (H, w)$,
\[
\boxed{\quad \phi(G, v) \;\iff\; \phi(H, w) \quad}
\]
\end{thm}

\noindent
This formalizes the backward direction of the geometric van~Benthem
theorem (Theorem~\ref{thm:gvb}): every bisimulation-invariant geometric
formula of existential depth~${\leq}\,d$ is $d$-bisimulation-invariant.
The forward direction (Theorem~\ref{thm:hml-invariant}) is fully
constructive.
The backward transfer---that $\phi(G, v)$ implies
$\phi(\mathcal{T}_d(G,v), \mathrm{root})$---is proved by case split via a
decision procedure on the finite tree.  For any concrete formula, LTS, and
vertex, \texttt{satisfies\_decidable} computes whether $\phi$ holds at the
tree root; the positive case is immediate and fully constructive.  The
negative case follows by contrapositive: bisimulation-invariance ensures
that failure on the tree propagates to the graph, so no additional
geometric content is needed.

\begin{lem}[Characteristic Formula Correspondence]
\label{lem:char-formula}
For each state~$v$ of a finite tree~$T$ at depth~$d$, the
\emph{characteristic formula}~$\chi_{T,v}$ satisfies
$w \models \chi_{T,v}$ iff $(T,v) \bismark_d (G,w)$
(constructive, zero axioms).
\end{lem}

Via the lemma, the \emph{geo-tree decomposition} reduces geometric
satisfaction on an arbitrary graph to $\HML$ satisfaction on its tree
unraveling.  The transfer theorem then proceeds by case analysis on
truth values: the TRUE-on-tree cases follow from the characteristic
formula and simulation invariance of positive $\HML$; all cases close.

\paragraph{From rooted transition systems to the labeled spectrum.}
The hierarchy and van~Benthem results of
\S\S\ref{sec:geometric}--\ref{sec:hml} work with unlabeled rooted transition
systems and their geometric theories, the natural setting for
separation theorems, where explicit sequents
($\sigma_{\mathrm{tot}}$, $\sigma_{\mathrm{det}}$,
$\sigma_{\mathrm{conf}}$) distinguish classifying toposes.
The spectrum computation requires \emph{labeled} transition systems:
the van~Glabbeek spectrum is defined for labeled processes, and
Bisping's energy-game framework parameterizes equivalences over a label
alphabet.  For single-label systems the passage is lossless:
relational bisimulation on rooted transition systems lifts to labeled
bisimulation on their $\mathrm{FinLTS}(\mathbf{1})$ images, and
conversely (constructive, zero axioms).
We therefore shift to $\mathrm{FinLTS}(L)$ for the remainder
of the paper.  The rooted transition framework returns in
\S\ref{sec:internal-logic}, where the presheaf topos
$\mathrm{Set}[\Th{\mathrm{LTS}}]$ serves as the ambient
workspace for Grothendieck topologies, bridging per-system geometric
theories (\S\S\ref{sec:geometric}--\ref{sec:hierarchy}) and the global
spectrum structure (this section).

\section{The Spectrum Lattice}
\label{sec:spectrum-lattice}
\label{sec:spectrum-subtoposes}

The spectrum embeds into the coframe of subtoposes via two independent
constructions: an \emph{algebraic pillar} (nuclei on Lindenbaum algebras,
this section) and a \emph{geometric pillar} (Grothendieck topologies on
the presheaf topos, \S\ref{sec:internal-logic}).  The Energy--LT
Bridge (Theorem~\ref{thm:energy-lt-bridge}) connects them.

The 13~named equivalences of the van~Glabbeek spectrum, parameterized by
Bisping's 6-dimensional energy vectors~\cite{bisping2023cav}, form a
set~$\mathcal{S} \subset (\mathbb{N} \cup \{\infty\})^6$ whose sublattice
closure $L_{30} := \langle\mathcal{S}\rangle$ under componentwise
$\min$/$\max$ has 30~elements.  We
first establish the quotient-filtration framework, then analyze this lattice
algebraically (closure, Birkhoff representation, bi-Heyting structure) and
finally instantiate through labeled transition systems and concrete
separations.

The preceding sections characterized the \emph{upper} boundary of the
spectrum: how geometric logic, through first-order quantification, the three
separating mechanisms (equality, cyclic variable sharing, variable
repetition), and the van~Benthem theorem, exceeds Hennessy--Milner logic.
We now turn to the \emph{lower} structure: the internal stratification of
$\HML$ itself into sublanguages, parameterized by energy budgets.  This is a
propositional question (the Lindenbaum algebras~$L_{\bar{e}}$ are
propositional geometric theories) and the separating mechanisms
of~\S\S3--4 play no role here.

\subsection{Energy Budgets and the Spectrum}
\label{sec:spectrum}

Van Glabbeek's linear time--branching time
spectrum~\cite{vanglabbeek2001} identifies the canonical behavioral
equivalences for concrete, sequential, finitely branching processes.
Bisping's energy-game framework~\cite{bisping2023cav} captures
13~of these via 6-dimensional energy coordinates, replacing
completed trace, completed simulation, and possible-worlds
equivalences with impossible futures, revivals, and enabledness.
We work throughout with Bisping's~13.
Caramello's Duality
Theorem~\cite[Theorem~3.6]{caramello2017theories} converts this into a
subtopos filtration: adding axioms restricts models, hence determines a
subtopos.

For a rooted transition system~$\Msys$, define the \emph{quotient theory at spectrum
level~$\ell$} as~$\Th{\Msys}^{(\ell)} \supseteq \Th{\Msys}$, obtained by adding
all geometric sequents valid in~$\Msys$ that are invariant under the
equivalence at level~$\ell$.  The Duality Theorem gives a chain:
\[
\boxed{\;
\ET{\Msys} \supseteq
\mathcal{E}[\Th{\Msys}^{(\mathrm{trace})}] \supseteq
\mathcal{E}[\Th{\Msys}^{(\mathrm{sim})}] \supseteq
\mathcal{E}[\Th{\Msys}^{(\mathrm{bisim})}]
\;}
\]
Each inclusion is strict when the corresponding level distinguishes
genuinely different branching structures of~$M$: the First Separation
witnesses the trace/simulation gap (the fork~$\FK$ and path~$\PT$ agree on
traces but not on the simulation formula $\Diamond\top$), and the Second
Separation witnesses the bisimulation/topos gap (the
hub-spokes~$\HS$ and two-cycle~$\TC$ are bisimilar but distinguished by
$\sigma_{\mathrm{det}}$).

\begin{defi}[Morleyization for negative equivalences]
\label{def:morleyization}
Seven of the 13~named equivalences (traces through simulation) are
\emph{natively geometric}: their distinguishing $\HML$ fragments use only
$\top$, $\wedge$, $\bigvee$, and $\Diamond$.  The remaining six (failures,
failure traces, readiness, ready traces, impossible futures, ready simulation)
require \emph{inability atoms} $\neg\langle a\rangle\top$, which are not
geometric formulas.

Following Johnstone~\cite[D1.5.13]{johnstone2002sketches}, we handle these
via \emph{Morleyization}: each predicate~$P$ is replaced by a complementary
pair $(P, P^c)$ with an axiom $P \vee P^c = \top$.  Concretely,
$\texttt{unable}_a(s) := \neg\langle a\rangle\top$ becomes a fresh atomic
predicate~$\texttt{unable}_a$ with
$\texttt{able}_a \vee \texttt{unable}_a = \top$, rendering the theory
geometric.  All algebraic results of \S\S\ref{sec:spectrum}--\ref{sec:spectrum-lattice}
apply uniformly to the Morleyized theories.

\label{rem:ready-sim-bisim}
The ready-simulation/bisimulation separation is witnessed by the formula
$[a]\langle c\rangle\top$, which distinguishes van~Glabbeek's pair
$a.b + a.(b{+}c)$ from $a.(b{+}c)$~\cite{vanglabbeek2001}: the formula
holds at the latter (the unique $a$-successor can perform~$c$) but not
at the former ($a.b$ reaches a state with no $c$-edge).  This formula lies
in full HML but not in the ready-simulation fragment (its nested negation
$\neg\langle a\rangle\neg\langle c\rangle\top$ is not an inability atom),
completing the four-level labeled separation.
\end{defi}

\begin{defi}[Energy-indexed Lindenbaum family]
For each named equivalence~$\bar{e}$ in the van~Glabbeek spectrum, the
energy budget determines which propositional atoms are ``visible'':
base transition atoms (always), branching atoms (when $e_2 \geq 2$),
and inability atoms $\texttt{unable}_a(s)$ encoding $\neg\langle a\rangle\top$
(when $e_5 \geq 1$ and $e_6 \geq 1$).  The restricted atom set~$A_{\bar{e}}$
determines a propositional geometric theory~$\Th{\Msys}^{\bar{e}}$ whose
Lindenbaum algebra~$L_{\bar{e}}$ captures exactly the state distinctions
visible under energy budget~$\bar{e}$.  The family $\{L_{\bar{e}}\}$ is
monotone: $\bar{e}_1 \leq \bar{e}_2$ implies $|L_{\bar{e}_1}| \leq |L_{\bar{e}_2}|$
(more energy $=$ finer distinctions $=$ larger algebra).
\end{defi}

\begin{rem}[Coordinatization dependence]
\label{rem:coord-dependence}
The lattice~$L_{30} = \langle\mathcal{S}\rangle$ depends on
Bisping's 6-dimensional energy coordinatization, not merely on the
partial order of~$\mathcal{S}$.  The Birkhoff representation gives
$|J(L_{30})| = 10$ (4~named, 6~unnamed), but the free bounded
distributive lattice on the 13-element van~Glabbeek
poset~$P$~\cite{vanglabbeek2001} is the downset lattice~$\mathcal{O}(P)$,
which has $|J(\mathcal{O}(P))| = |P| = 13$ join-irreducibles.
Since $10 \neq 13$, the Bisping embedding introduces accidental
join-relations: 9~of the 13~named equivalences are join-reducible
in~$L_{30}$ due to the specific dimensional decomposition (observation
depth, conjunction width, positive/negative clause depth, negation
nesting).  The standard Birkhoff embedding
$p \mapsto (\mathbf{1}_{q \leq p})_{q \in P}$ into~$\{0,1\}^{13}$ is
order-reflecting and yields~$\mathcal{O}(P) \not\cong L_{30}$ under
componentwise $\min$/$\max$.  Thus $L_{30}$ is an invariant of the
energy-game \emph{dimensional structure}, not of the poset alone.
The topos-theoretic route offers an alternative
coordinatization-free characterization: $L_{30}$ embeds into the
subtopos coframe $\mathrm{Sub}(\mathrm{Set}[\Th{\mathrm{LTS}}])$,
which is an invariant of the classifying topos.
\end{rem}

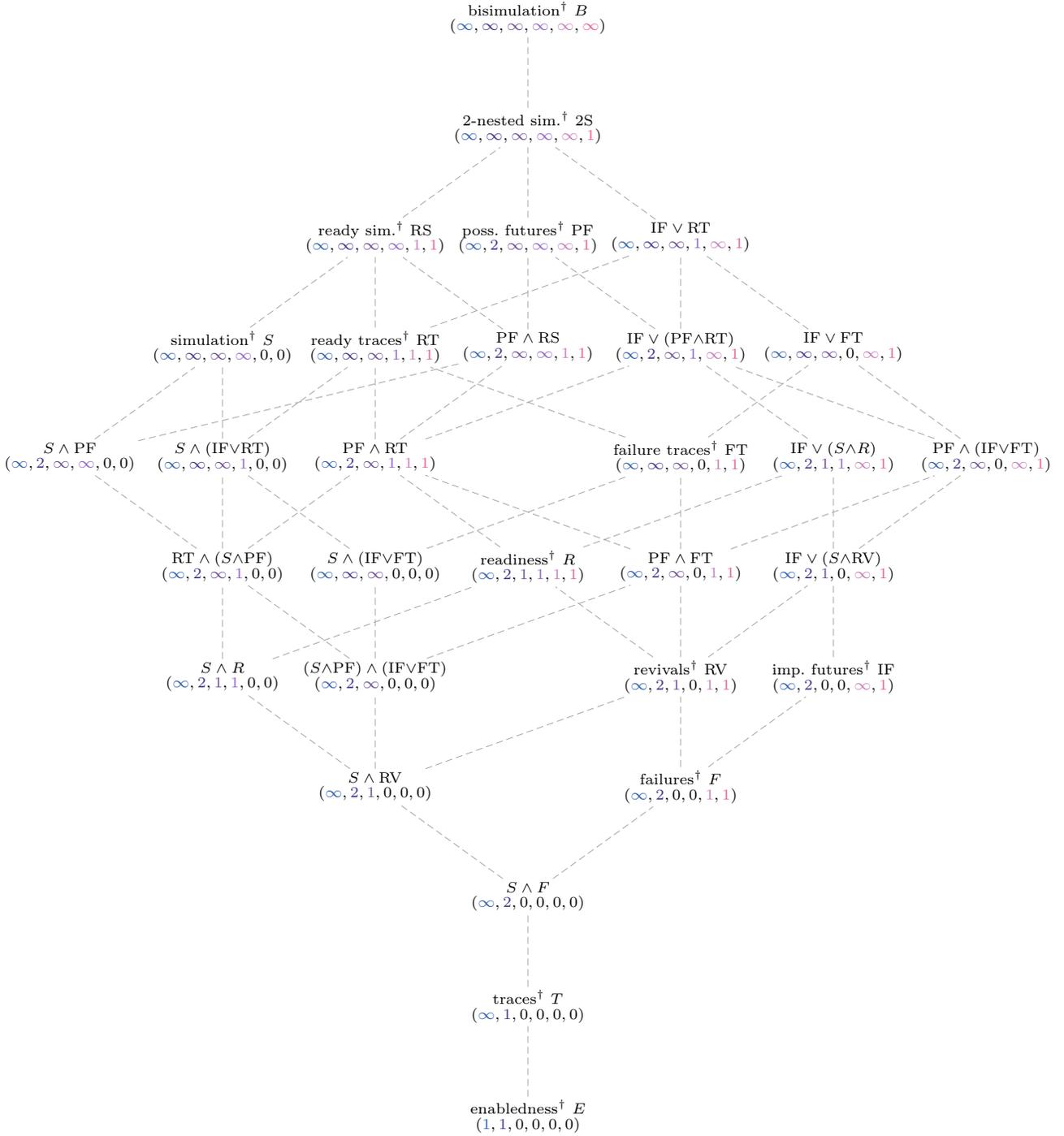
\begin{figure}[p]
  \makebox[\textwidth][c]{%
  \begin{tikzpicture}[align=center, yscale=1.2,
    every node/.style={font=\tiny, inner sep=1.5pt, fill=white},
    every edge/.style={draw=black!35, line cap=round, line width=0.35pt, densely dashed},
    unn/.style={}]


    \node (B) at (0, 0)
      {bisimulation$^\dagger$ $B$\\[-1pt]
       $(\textcolor{obsColor}{\infty}, \textcolor{conjColor}{\infty},
         \textcolor{mainPosColor}{\infty}, \textcolor{otherPosColor}{\infty},
         \textcolor{negObsColor}{\infty}, \textcolor{negsColor}{\infty})$};

    \node (S2) at (0, -1.5)
      {$2$-nested sim.$^\dagger$ $\eqName{2S}$\\[-1pt]
       $(\textcolor{obsColor}{\infty}, \textcolor{conjColor}{\infty},
         \textcolor{mainPosColor}{\infty}, \textcolor{otherPosColor}{\infty},
         \textcolor{negObsColor}{\infty}, \textcolor{negsColor}{1})$};

    \node (RS) at (-2.5, -3.0)
      {ready sim.$^\dagger$ $\eqName{RS}$\\[-1pt]
       $(\textcolor{obsColor}{\infty}, \textcolor{conjColor}{\infty},
         \textcolor{mainPosColor}{\infty}, \textcolor{otherPosColor}{\infty},
         \textcolor{negObsColor}{1}, \textcolor{negsColor}{1})$};
    \node (PF) at (0, -3.0)
      {poss.\ futures$^\dagger$ $\eqName{PF}$\\[-1pt]
       $(\textcolor{obsColor}{\infty}, \textcolor{conjColor}{2},
         \textcolor{mainPosColor}{\infty}, \textcolor{otherPosColor}{\infty},
         \textcolor{negObsColor}{\infty}, \textcolor{negsColor}{1})$};
    \node[unn] (IFjRT) at (2.5, -3.0)
      {$\eqName{IF} \vee \eqName{RT}$\\[-1pt]
       $(\textcolor{obsColor}{\infty}, \textcolor{conjColor}{\infty},
         \textcolor{mainPosColor}{\infty}, \textcolor{otherPosColor}{1},
         \textcolor{negObsColor}{\infty}, \textcolor{negsColor}{1})$};

    \node (S) at (-5.0, -4.5)
      {simulation$^\dagger$ $S$\\[-1pt]
       $(\textcolor{obsColor}{\infty}, \textcolor{conjColor}{\infty},
         \textcolor{mainPosColor}{\infty}, \textcolor{otherPosColor}{\infty},
         0, 0)$};
    \node (RT) at (-2.5, -4.5)
      {ready traces$^\dagger$ $\eqName{RT}$\\[-1pt]
       $(\textcolor{obsColor}{\infty}, \textcolor{conjColor}{\infty},
         \textcolor{mainPosColor}{\infty}, \textcolor{otherPosColor}{1},
         \textcolor{negObsColor}{1}, \textcolor{negsColor}{1})$};
    \node[unn] (PFmRS) at (0, -4.5)
      {$\eqName{PF} \wedge \eqName{RS}$\\[-1pt]
       $(\textcolor{obsColor}{\infty}, \textcolor{conjColor}{2},
         \textcolor{mainPosColor}{\infty}, \textcolor{otherPosColor}{\infty},
         \textcolor{negObsColor}{1}, \textcolor{negsColor}{1})$};
    \node[unn] (IFjPFRT) at (2.5, -4.5)
      {$\eqName{IF} \vee (\eqName{PF}{\wedge}\eqName{RT})$\\[-1pt]
       $(\textcolor{obsColor}{\infty}, \textcolor{conjColor}{2},
         \textcolor{mainPosColor}{\infty}, \textcolor{otherPosColor}{1},
         \textcolor{negObsColor}{\infty}, \textcolor{negsColor}{1})$};
    \node[unn] (IFjFT) at (5.0, -4.5)
      {$\eqName{IF} \vee \eqName{FT}$\\[-1pt]
       $(\textcolor{obsColor}{\infty}, \textcolor{conjColor}{\infty},
         \textcolor{mainPosColor}{\infty}, 0,
         \textcolor{negObsColor}{\infty}, \textcolor{negsColor}{1})$};

    \node[unn] (SmPF) at (-7.5, -6.0)
      {$S \wedge \eqName{PF}$\\[-1pt]
       $(\textcolor{obsColor}{\infty}, \textcolor{conjColor}{2},
         \textcolor{mainPosColor}{\infty}, \textcolor{otherPosColor}{\infty},
         0, 0)$};
    \node[unn] (SmIFjRT) at (-5.0, -6.0)
      {$S \wedge (\eqName{IF}{\vee}\eqName{RT})$\\[-1pt]
       $(\textcolor{obsColor}{\infty}, \textcolor{conjColor}{\infty},
         \textcolor{mainPosColor}{\infty}, \textcolor{otherPosColor}{1},
         0, 0)$};
    \node[unn] (PFmRT) at (-2.5, -6.0)
      {$\eqName{PF} \wedge \eqName{RT}$\\[-1pt]
       $(\textcolor{obsColor}{\infty}, \textcolor{conjColor}{2},
         \textcolor{mainPosColor}{\infty}, \textcolor{otherPosColor}{1},
         \textcolor{negObsColor}{1}, \textcolor{negsColor}{1})$};
    \node (FT) at (2.5, -6.0)
      {failure traces$^\dagger$ $\eqName{FT}$\\[-1pt]
       $(\textcolor{obsColor}{\infty}, \textcolor{conjColor}{\infty},
         \textcolor{mainPosColor}{\infty}, 0,
         \textcolor{negObsColor}{1}, \textcolor{negsColor}{1})$};
    \node[unn] (IFjSR) at (5.0, -6.0)
      {$\eqName{IF} \vee (S{\wedge}R)$\\[-1pt]
       $(\textcolor{obsColor}{\infty}, \textcolor{conjColor}{2},
         \textcolor{mainPosColor}{1}, \textcolor{otherPosColor}{1},
         \textcolor{negObsColor}{\infty}, \textcolor{negsColor}{1})$};
    \node[unn] (PFmIFjFT) at (7.5, -6.0)
      {$\eqName{PF} \wedge (\eqName{IF}{\vee}\eqName{FT})$\\[-1pt]
       $(\textcolor{obsColor}{\infty}, \textcolor{conjColor}{2},
         \textcolor{mainPosColor}{\infty}, 0,
         \textcolor{negObsColor}{\infty}, \textcolor{negsColor}{1})$};

    \node[unn] (RTmSPF) at (-5.0, -7.5)
      {$\eqName{RT} \wedge (S{\wedge}\eqName{PF})$\\[-1pt]
       $(\textcolor{obsColor}{\infty}, \textcolor{conjColor}{2},
         \textcolor{mainPosColor}{\infty}, \textcolor{otherPosColor}{1},
         0, 0)$};
    \node[unn] (SmIFjFT) at (-2.5, -7.5)
      {$S \wedge (\eqName{IF}{\vee}\eqName{FT})$\\[-1pt]
       $(\textcolor{obsColor}{\infty}, \textcolor{conjColor}{\infty},
         \textcolor{mainPosColor}{\infty}, 0,
         0, 0)$};
    \node (R) at (0, -7.5)
      {readiness$^\dagger$ $R$\\[-1pt]
       $(\textcolor{obsColor}{\infty}, \textcolor{conjColor}{2},
         \textcolor{mainPosColor}{1}, \textcolor{otherPosColor}{1},
         \textcolor{negObsColor}{1}, \textcolor{negsColor}{1})$};
    \node[unn] (PFmFT) at (2.5, -7.5)
      {$\eqName{PF} \wedge \eqName{FT}$\\[-1pt]
       $(\textcolor{obsColor}{\infty}, \textcolor{conjColor}{2},
         \textcolor{mainPosColor}{\infty}, 0,
         \textcolor{negObsColor}{1}, \textcolor{negsColor}{1})$};
    \node[unn] (IFjSRV) at (5.0, -7.5)
      {$\eqName{IF} \vee (S{\wedge}\eqName{RV})$\\[-1pt]
       $(\textcolor{obsColor}{\infty}, \textcolor{conjColor}{2},
         \textcolor{mainPosColor}{1}, 0,
         \textcolor{negObsColor}{\infty}, \textcolor{negsColor}{1})$};

    \node[unn] (SmR) at (-5.0, -9.0)
      {$S \wedge R$\\[-1pt]
       $(\textcolor{obsColor}{\infty}, \textcolor{conjColor}{2},
         \textcolor{mainPosColor}{1}, \textcolor{otherPosColor}{1},
         0, 0)$};
    \node[unn] (SPFmIFjFT) at (-2.5, -9.0)
      {$(S{\wedge}\eqName{PF}) \wedge (\eqName{IF}{\vee}\eqName{FT})$\\[-1pt]
       $(\textcolor{obsColor}{\infty}, \textcolor{conjColor}{2},
         \textcolor{mainPosColor}{\infty}, 0,
         0, 0)$};
    \node (RV) at (2.5, -9.0)
      {revivals$^\dagger$ $\eqName{RV}$\\[-1pt]
       $(\textcolor{obsColor}{\infty}, \textcolor{conjColor}{2},
         \textcolor{mainPosColor}{1}, 0,
         \textcolor{negObsColor}{1}, \textcolor{negsColor}{1})$};
    \node (IF) at (5.0, -9.0)
      {imp.\ futures$^\dagger$ $\eqName{IF}$\\[-1pt]
       $(\textcolor{obsColor}{\infty}, \textcolor{conjColor}{2},
         0, 0,
         \textcolor{negObsColor}{\infty}, \textcolor{negsColor}{1})$};

    \node[unn] (SmRV) at (-2.5, -10.5)
      {$S \wedge \eqName{RV}$\\[-1pt]
       $(\textcolor{obsColor}{\infty}, \textcolor{conjColor}{2},
         \textcolor{mainPosColor}{1}, 0,
         0, 0)$};
    \node (F) at (2.5, -10.5)
      {failures$^\dagger$ $F$\\[-1pt]
       $(\textcolor{obsColor}{\infty}, \textcolor{conjColor}{2},
         0, 0,
         \textcolor{negObsColor}{1}, \textcolor{negsColor}{1})$};

    \node[unn] (SmF) at (0, -12.0)
      {$S \wedge F$\\[-1pt]
       $(\textcolor{obsColor}{\infty}, \textcolor{conjColor}{2},
         0, 0, 0, 0)$};

    \node (T) at (0, -13.5)
      {traces$^\dagger$ $T$\\[-1pt]
       $(\textcolor{obsColor}{\infty}, \textcolor{conjColor}{1},
         0, 0, 0, 0)$};

    \node (E) at (0, -15.0)
      {enabledness$^\dagger$ $E$\\[-1pt]
       $(\textcolor{obsColor}{1}, \textcolor{conjColor}{1},
         0, 0, 0, 0)$};


    \path
    (T) edge (SmF)
    (SmF) edge (SmRV)
    (SmF) edge (F)
    (SmRV) edge (SPFmIFjFT)
    (SmRV) edge (SmR)
    (SmRV) edge (RV)
    (SmR) edge (RTmSPF)
    (SmR) edge (R)
    (SPFmIFjFT) edge (RTmSPF)
    (SPFmIFjFT) edge (SmIFjFT)
    (SPFmIFjFT) edge (PFmFT)
    (RV) edge (PFmFT)
    (RV) edge (IFjSRV)
    (IF) edge (IFjSRV)
    (RTmSPF) edge (SmPF)
    (RTmSPF) edge (SmIFjRT)
    (RTmSPF) edge (PFmRT)
    (SmIFjFT) edge (SmIFjRT)
    (SmIFjFT) edge (FT)
    (R) edge (PFmRT)
    (R) edge (IFjSR)
    (PFmFT) edge (FT)
    (PFmFT) edge (PFmRT)
    (PFmFT) edge (PFmIFjFT)
    (IFjSRV) edge (IFjSR)
    (IFjSRV) edge (PFmIFjFT)
    (SmPF) edge (S)
    (SmPF) edge (PFmRS)
    (SmIFjRT) edge (S)
    (SmIFjRT) edge (RT)
    (PFmRT) edge (RT)
    (PFmRT) edge (PFmRS)
    (PFmRT) edge (IFjPFRT)
    (FT) edge (IFjFT)
    (IFjSR) edge (IFjPFRT)
    (PFmIFjFT) edge (IFjFT)
    (PFmIFjFT) edge (IFjPFRT)
    (RT) edge (IFjRT)
    (PFmRS) edge (RS)
    (PFmRS) edge (PF)
    (IFjPFRT) edge (PF)
    (IFjPFRT) edge (IFjRT)
    (IFjFT) edge (IFjRT)
    (IFjRT) edge (S2)
    ;

    \path
    (E) edge (T)
    (F) edge (RV)
    (F) edge (IF)
    (RV) edge (R)
    (FT) edge (RT)
    (S) edge (RS)
    (RT) edge (RS)
    (RS) edge (S2)
    (PF) edge (S2)
    (S2) edge (B)
    ;
  \end{tikzpicture}}%
  \caption{The spectrum lattice~$L_{30} = \langle\mathcal{S}\rangle$:
  sublattice closure of
  the 13~named energy vectors~($^\dagger$)~\cite{bisping2023cav} 
  under componentwise $\min/\max$, with 17~unnamed hybrids.
  Energy coordinates
  $(\textcolor{obsColor}{e_1}, \textcolor{conjColor}{e_2},
  \textcolor{mainPosColor}{e_3}, \textcolor{otherPosColor}{e_4},
  \textcolor{negObsColor}{e_5}, \textcolor{negsColor}{e_6})$ measure
  \textcolor{obsColor}{observation depth},
  \textcolor{conjColor}{conjunction nesting},
  \textcolor{mainPosColor}{deepest positive clause depth},
  \textcolor{otherPosColor}{other positive clause depth},
  \textcolor{negObsColor}{negative clause depth}, and
  \textcolor{negsColor}{negation nesting}.
  Finer equivalences appear higher; edges denote covering relations.}
  \label{fig:spectrum-energy}
\end{figure}

Concrete Lindenbaum algebra sizes illustrate the framework.  The
branching system $a.b + a.c$ has
$|L_{\mathrm{trace}}| = 5$ but $|L_{\mathrm{sim}}| = 7$: simulation
detects the branching that traces miss.  By contrast, the sequential
system $a.(b{+}c)$ has
$|L_{\mathrm{trace}}| = |L_{\mathrm{sim}}| = 5$; no branching, so
no trace$\to$simulation jump.  This matches the depth-1 separation
of~\S\ref{sec:graded-atoms}, confirming that the energy-indexed family
unifies the graded tower as a 1-dimensional slice.
Failures~$(\infty, 2, 0, 0, 1, 1)$ and simulation~$(\infty, \infty, \infty, \infty, 0, 0)$ reach the same algebra
size~(7) via \emph{independent} paths in the energy lattice: failures
uses negation to detect refusals, while simulation uses deep
conjunction to detect branching.

The 13~energy vectors determine 13~points in the ambient frame
$(\mathbb{N} \cup \{\infty\})^6$.  We now show that closing these
points under componentwise $\min$/$\max$ yields a 30-element lattice
with rich algebraic structure.

\subsection{The 30-Element Lattice Closure}
\label{sec:lattice-closure}

\begin{defi}[Energy-to-nucleus map]
A \emph{nucleus} on a frame~$L$ is an inflationary idempotent
$\wedge$-preserving endomorphism $j\colon L \to L$.
Johnstone~\cite[II.2]{johnstone1982stone} establishes that
Lawvere--Tierney topologies on $\Sh(L)$ correspond
bijectively to nuclei on~$L$.  For each named equivalence~$\bar{e}$,
the process equivalence ${\sim}_{\bar{e}}$ extends to a congruence
on the full Lindenbaum algebra~$L_\infty$, inducing a
nucleus~$j_{\bar{e}}$ with fixpoint set $\mathrm{Fix}(j_{\bar{e}})
\cong L_{\bar{e}}$.  The resulting map
\[
  \bar{e} \;\longmapsto\; j_{\bar{e}}
  \qquad\text{is \emph{order-reversing}:}
  \qquad
  \bar{e}_1 \leq \bar{e}_2
  \;\implies\;
  j_{\bar{e}_2} \leq j_{\bar{e}_1}.
\]
Larger energy budgets allow finer distinctions, hence less closure,
hence a smaller nucleus (closer to the identity).  Bisimulation
(the finest named equivalence) maps to the identity
nucleus~$j_{\mathrm{bisim}} = \mathrm{id}$, while coarser
equivalences map to larger nuclei.

Here $L_\infty := \varinjlim_d L_d$ is the colimit Lindenbaum algebra
at unbounded modal depth (which stabilizes at depth~$n{-}1$ for an
$n$-state system, hence is a finite distributive lattice;
see~\S\ref{sec:labeled-spectrum}).  The extension from process equivalence to
frame congruence relies on the logical characterization
theorems~\cite{vanglabbeek2001}: two states are $\bar{e}$-equivalent
iff they satisfy the same formulas of energy $\leq \bar{e}$
(Bisping~\cite[Proposition~1]{bisping2023cav}), so the equivalence
classes are unions of Lindenbaum equivalence classes and therefore
compatible with the lattice operations.  For finite distributive
lattices, every lattice congruence is automatically a frame congruence
(Funayama--Nakayama~\cite{funayama1942congruence}).
\end{defi}

\begin{thm}[Lattice Closure of the Spectrum]
\label{thm:lattice-closure}
Let $\mathcal{S} \subset (\mathbb{N} \cup \{\infty\})^6$ be the
13~named energy vectors of the van~Glabbeek spectrum.  Then
$\mathcal{S}$ is not a sublattice, and the sublattice it generates
under componentwise $\min/\max$ is
\[
\boxed{\; L_{30} \;:=\; \langle\mathcal{S}\rangle, \qquad
|L_{30}| \;=\; 30. \;}
\]
The closure adds $17$ unnamed elements, each with a concrete energy
characterization but no classical process-algebraic name, and
stabilizes in three rounds of iterated pairwise meet/join
(verified computationally in the Lean formalization).
\end{thm}

\begin{proof}[Proof of non-sublattice property]
Three witnesses demonstrate that componentwise operations escape~$\mathcal{S}$:
\begin{enumerate}
\item $\eqName{PF} \wedge \eqName{FT}
  = (\infty, 2, \infty, 0, 1, 1)$ is unnamed (meet);
\item $S \wedge F
  = (\infty, 2, 0, 0, 0, 0)$ is unnamed (meet);
\item $\eqName{IF} \vee \eqName{FT}
  = (\infty, \infty, \infty, 0, \infty, 1)$ is unnamed (join).
\end{enumerate}
Conversely, some operations do land on named equivalences:
$\eqName{IF} \wedge \eqName{FT} = F$ and
$S \vee F = \eqName{RS}$.
Each case is decided by exhaustive comparison against
all elements of~$\mathcal{S}$.
The closure stabilizes after three rounds of iterated pairwise
meet/join, adding 17~new elements total.
\end{proof}

\noindent
The lattice of subtoposes of any Grothendieck topos is a
\emph{coframe}~\cite[A4.5]{johnstone2002sketches}: the order-dual of the
nucleus frame carries co-Heyting subtraction $A \setminus B$, the largest
subtopos~$C$ with $C \sqcup B \geq A$.

\subsection{Algebraic Structure}
\label{sec:algebraic-structure}

\begin{rem}[Birkhoff representation and bi-Heyting structure]
The Birkhoff representation $L_{30} \cong \mathcal{O}(J)$ identifies
$|J(L_{30})| = |M(L_{30})| = 10$ join-irreducible and meet-irreducible
elements.  As a finite distributive lattice, $L_{30}$ carries both a
Heyting algebra structure (with implication $a \to b$ the largest~$z$
satisfying $z \wedge a \leq b$) and a co-Heyting algebra structure
(with subtraction as above), making it a \emph{bi-Heyting algebra}.
The negation structure of this bi-Heyting algebra is
\textbf{maximally degenerate}: the Heyting pseudocomplement satisfies
$\neg x = \bot$ for every $x \neq \bot$, while the co-Heyting
negation satisfies ${\sim} x = \top$ for every $x \neq \top$.
All~30 Heyting boundaries $\partial x = x \wedge \neg x$ are trivially
zero.  The root cause is topological: $J(L_{30})$ is a \emph{connected}
poset (traces lies below every join-irreducible element), so the
pseudocomplement of any nonzero element is forced to~$\bot$.
Consequently the Boolean core $\{x : \neg\neg x = x\}$ reduces to
$\{\bot, \top\}$: only enabledness ($\bot$) and bisimulation ($\top$)
are regular elements, confirming that $L_{30}$ is maximally
non-Boolean.
\end{rem}

\begin{thm}[Algebraic Irreducibility of the Spectrum]
\label{thm:algebraic-irreducibility}
The spectrum lattice $L_{30}$ is directly indecomposable: it cannot be expressed as a
non-trivial product $L_1 \times L_2$ of distributive lattices with $|L_1|, |L_2| > 1$.
Equivalently, its join-irreducible poset $J(L_{30})$ is connected.
\end{thm}

\begin{proof}
Trace equivalence lies below every join-irreducible element in $J(L_{30})$,
so $J$ is connected.  By Birkhoff's decomposition theorem for finite distributive lattices,
$L_{30}$ is directly indecomposable.
\end{proof}

\begin{figure}[ht]
\centering
\begin{tikzpicture}[
    every node/.style={font=\small, inner sep=2pt},
    every edge/.style={draw, thin},
  ]
  \node (B)    at (0, 0)       {$B$};
  \node (SPF)  at (-2.8, -1.6) {$S \wedge \mathrm{PF}$};
  \node (SIFT) at (0, -1.6)    {$S \wedge (\mathrm{IF} \vee \mathrm{FT})$};
  \node (SR)   at (-2.8, -3.2) {$S \wedge R$};
  \node (SPFT) at (0, -3.2)    {$(S \wedge \mathrm{PF}) \wedge (\mathrm{IF} \vee \mathrm{FT})$};
  \node (IF)   at (3, -3.2)    {$\mathrm{IF}$};
  \node (SRV)  at (-2.8, -4.8) {$S \wedge \mathrm{RV}$};
  \node (F)    at (3, -4.8)    {$F$};
  \node (SF)   at (0, -6.4)    {$S \wedge F$};
  \node (T)    at (0, -8)      {$T$};
  \coordinate (IFbend) at (3, -1.6);
  \path
    (B)    edge (SPF)
    (B)    edge (SIFT)
    (B)    edge (IFbend)
    (IFbend) edge (IF)
    (SPF)  edge (SR)
    (SPF)  edge (SPFT)
    (SIFT) edge (SPFT)
    (IF)   edge (F)
    (SR)   edge (SRV)
    (SPFT) edge (SRV)
    (F)    edge (SF)
    (SRV)  edge (SF)
    (SF)   edge (T)
    ;
\end{tikzpicture}
\caption{The join-irreducible poset~$J(L_{30})$: 4~named elements
($T$, $F$, $\mathrm{IF}$, $B$) and 6~unnamed meets, connected by
12~covering relations.  Traces~($T$) lies below every join-irreducible,
confirming connectivity.}
\label{fig:join-irreducibles}
\end{figure}
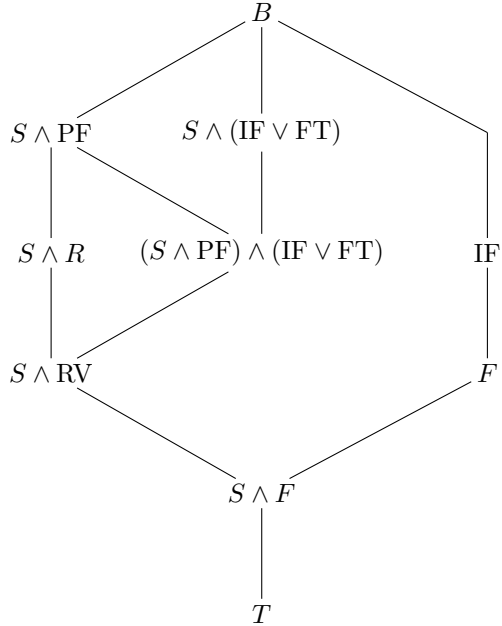

\noindent
Traces~($T$) lies below every join-irreducible, confirming connectivity
(Figure~\ref{fig:join-irreducibles}).
The left chain
$T$--$S{\wedge}F$--$S{\wedge}\mathrm{RV}$--$S{\wedge}R$--$S{\wedge}\mathrm{PF}$--$B$
carries the positive-observation hierarchy; the right chain
$T$--$S{\wedge}F$--$F$--$\mathrm{IF}$--$B$ carries the negation
hierarchy.  Their interaction through the two cross-links via
$(S {\wedge} \mathrm{PF}) {\wedge} (\mathrm{IF} {\vee} \mathrm{FT})$
is what makes~$J$ connected and $L_{30}$ indecomposable.

\begin{rem}[Operational glosses]
The six unnamed elements of~$J(L_{30})$ admit compact operational readings:
\begin{enumerate}\setlength\itemsep{1pt}
\item $S \wedge F$: \emph{simulation-failures}, processes agree on both
  positive branching (simulation's conjunction depth) and single-step
  refusals (failures' negation).
\item $S \wedge \eqName{RV}$: \emph{simulation-revivals}, processes
  match both branching structure and revival patterns after refusal.
\item $S \wedge R$: \emph{simulation-readiness}, processes
  match branching and ready sets.
\item $S \wedge \eqName{PF}$: \emph{simulation-futures}, processes
  match branching and possible-futures enumeration.
\item $(S \wedge \eqName{PF}) \wedge (\eqName{IF} \vee \eqName{FT})$:
  \emph{cross-link}, the unique element mediating between the positive
  and negative hierarchies.
\item $S \wedge (\eqName{IF} \vee \eqName{FT})$:
  \emph{simulation meets impossible-futures/failure-traces join},
  processes match branching and bounded negative observation.
\end{enumerate}

The six energy dimensions of Bisping's framework are not algebraically
independent: the spectral lattice they generate is indecomposable,
admitting no factorization into independent subsystems.  This is perhaps
surprising, since Bisping's framework treats the dimensions as independent
parameters and van~Glabbeek's original organization into linear-time and
branching-time families suggests a factorization that does not exist
algebraically.
\end{rem}

\begin{thm}[Heyting Implication of the Spectrum Lattice]
\label{thm:heyting-implication}
Despite maximally collapsed negation ($\neg x = \bot$ for all
$x \neq \bot$), the Heyting implication on $L_{30}$ is informative:
among the 32~incomparable ordered pairs of named equivalences, every
implication $a \to b$ is nontrivial.  In particular,
\[
\boxed{\quad S \to F \;=\; \mathrm{IF} \quad}
\]
the Heyting implication of simulation into failures gives impossible
futures.  Six of the 32~incomparable implications produce unnamed
elements.
\end{thm}

\noindent
Concretely, $S \to F = \mathrm{IF}$ means that the relative information gap
between simulation (purely positive observation: arbitrarily nested
conjunctions but zero negation) and failures (bounded conjunction plus
single negation) is precisely impossible futures (bounded conjunction,
unlimited negation).  The negation dimension mediates the
testing-simulation divide.
The Heyting algebra structure systematically predicts
process equivalences invisible to the classical spectrum.

\begin{rem}[Per-system and global perspectives]
The algebraic results of this section ($L_{30}$, its bi-Heyting
structure, its indecomposability) are proved constructively on the
abstract lattice.  Their identification with coframe operations
on the classifying topos is the content of the Energy--LT Bridge
(Theorem~\ref{thm:energy-lt-bridge}), which requires the axioms
discussed in \S\ref{sec:conclusion}.  The $S \to F = \mathrm{IF}$
computation is an abstract lattice-theoretic result; the Geometric
Closure Theorem of \S\ref{sec:geometric-closure} computes Heyting
implications in the presheaf topos via free extensions.  Their agreement
on concrete examples provides evidence for the sublattice embedding
conjecture (\S\ref{sec:conclusion}).
\end{rem}

\begin{thm}[Co-Heyting Decomposition of the Spectrum]
\label{thm:coheyting-decomposition}
The six co-Heyting subtractions in~$L_{30}$,
computed via the Birkhoff formula
$x \setminus y = \displaystyle\bigvee\{j \in J : j \leq x,\; j \not\leq y\}$:
\[
\left\{
\begin{array}{lcl}
S \setminus T = S &\quad& \text{(trace captures none of simulation's branching content)}\\
F \setminus T = F && \text{(trace captures none of failures' negation content)}\\
B \setminus \mathrm{RS} = B && \text{(ready sim.\ captures none of bisimulation's depth)}\\[4pt]
\mathrm{RS} \setminus S = F && \text{(ready sim.\ adds failure-testing over simulation)}\\
\mathrm{2S} \setminus \mathrm{RS} = \mathrm{IF} && \text{(2-nested sim.\ adds impossible-futures testing)}\\
\mathrm{PF} \setminus \mathrm{IF} = S \wedge \mathrm{PF} && \text{(unnamed: no classical name)}
\end{array}
\right.
\]
\end{thm}

\smallskip\noindent
The naive co-Heyting subtraction in the product frame
$(\mathbb{N} \cup \{\infty\})^6$ is computed componentwise:
$(a \setminus b)(i) = a(i)$ if $a(i) > b(i)$, else~$0$.
Every naive result has $e_1 = 0$, placing it below
enabledness (the~$\bot$ of~$L_{30}$), so the naive and exact
(Birkhoff) subtractions \emph{never} coincide as energy vectors.
At the named-equivalence level, $S \setminus T$ and
$F \setminus T$ happen to agree (both return the minuend),
but the remaining four diverge outright:
naive $B \setminus \mathrm{RS}$ gives $(0,0,0,0,\infty,\infty)$
instead of~$B$;
naive $\mathrm{RS} \setminus S$ gives $(0,0,0,0,1,1)$
instead of~$F$;
naive $\mathrm{2S} \setminus \mathrm{RS}$ gives
$(0,0,0,0,\infty,0)$ instead of~$\mathrm{IF}$;
and naive $\mathrm{PF} \setminus \mathrm{IF}$ gives
$(0,0,\infty,\infty,0,0)$ instead of~$S \wedge \mathrm{PF}$.
The product-frame computation returns elements outside~$L_{30}$
or lands on the wrong named element, while the Birkhoff formula
correctly identifies the differential content within the
spectrum lattice.

The latter three decompositions are operationally revealing.
$\mathrm{RS} \setminus S = F$ identifies the \emph{refusal-detection gap}:
simulation matches branches positively (every choice of one process can be
matched by the other), while ready simulation additionally detects which
actions a process \emph{refuses}, and this refusal-detection capability is
precisely what failures equivalence measures.
$\mathrm{2S} \setminus \mathrm{RS} = \mathrm{IF}$ identifies the
\emph{negation-depth gap}: ready simulation detects single-step refusals,
while 2-nested simulation additionally detects that after accepting an action,
certain continuation sequences are precluded, and this deeper negation is
exactly what impossible futures captures.
$\mathrm{PF} \setminus \mathrm{IF} = S \wedge \mathrm{PF}$ lands on an
unnamed element combining simulation's branch-matching with possible futures'
enumeration of continuations; that this element has no classical name witnesses
the incompleteness of the named spectrum.

\begin{prop}[Existence of Unnamed Subtoposes]
\label{prop:unnamed-subtoposes}
The generating set~$\mathcal{S}$ does not fill the nucleus lattice of the
Lindenbaum algebra: $|\mathcal{S}| = 13$.  When $|J(L)| \geq 4$,
$2^{|J(L)|} \geq 16 > 13$, guaranteeing the existence of
Lawvere--Tierney topologies that correspond to no classical process
equivalence.
\end{prop}

\noindent
For a finite distributive lattice~$L$ with $|J(L)|$ join-irreducible
elements, Birkhoff's representation theorem, combined with the
distributivity guaranteed by
Funayama--Nakayama~\cite{funayama1942congruence},
gives exactly $2^{|J(L)|}$ nuclei (forming a Boolean algebra).
The 8~nuclei on the 5-element hub-spokes lattice
validate this framework:
$2^{|J(L)|} = 2^3 = 8$ matches the 8~Grothendieck topologies.
The $2^{|J(L)|}$ count applies to the per-system Lindenbaum
algebra~$L_\infty$ of a \emph{specific} system; the 17~unnamed
elements of~$L_{30}$ are a different manifestation of the same
phenomenon: unnamed positions in the spectrum ordering that the
13~named equivalences do not fill.

\subsection{Labeled Transition Systems}
\label{sec:labeled-spectrum}

Unlabeled systems exhibit a fundamental limitation: the single
modality~$\Diamond\varphi$ cannot distinguish branching structure,
causing the upper spectrum to collapse: simulation, ready simulation,
and bisimulation coincide.
Van Glabbeek's full linear time--branching time
spectrum~\cite{vanglabbeek2001} requires \emph{labeled} transitions
with distinct actions $a, b, c, \ldots$, giving distinct diamond
operators $\langle a \rangle \varphi$ for each action.

\begin{defi}[Labeled geometric theories]
For a labeled transition system $(S, \Sigma, {\to})$ with state set~$S$
and action alphabet~$\Sigma$, we define a propositional geometric
theory~$\Th{\Msys}^{\Sigma}$ with atoms $\step_a(s,t)$ for each $s, t \in S$
and $a \in \Sigma$, asserting ``state~$s$ performs action~$a$ and reaches
state~$t$.''  The axioms consist of \emph{exclusion axioms}
($\step_a(s,t) \vdash \bot$ when $s \not\xrightarrow{a} t$) and
\emph{totality axioms}\footnote{This totality axiom encodes the
deadlock-free convention standard in the van~Glabbeek
spectrum~\cite{vanglabbeek2001}.  Extending the framework to
deadlock-sensitive equivalences (e.g., De Nicola--Hennessy testing)
would require dropping totality and adding explicit termination
predicates; this is straightforward but outside our current scope.}
($\top \vdash \bigvee_{a,t: s \xrightarrow{a} t}
\step_a(s,t)$ for each state~$s$ that has at least one outgoing
transition, asserting that such states must transition).  The Lindenbaum algebra
$\mathrm{Lind}(\Th{\Msys}^{\Sigma})$ is the quotient of the free frame on
these labeled atoms by the axiom-induced relations.
\end{defi}

\paragraph{Van Glabbeek separating examples.}
We formalize three standard CCS
processes~\cite{vanglabbeek2001} over the three-letter alphabet
$\Sigma = \{a, b, c\}$.
The process $a.b + a.c$ has 5~states:
$p_0 \xrightarrow{a} p_1$, $p_0 \xrightarrow{a} p_2$,
$p_1 \xrightarrow{b} p_3$, $p_2 \xrightarrow{c} p_4$; the
initial state branches into two $a$-successors, one offering~$b$
and the other~$c$.
The process $a.(b+c)$ has 4~states:
$q_0 \xrightarrow{a} q_1$,
$q_1 \xrightarrow{b} q_2$, $q_1 \xrightarrow{c} q_3$, a
single $a$-successor offering both $b$ and~$c$.
The process $a.b + a.(b{+}c)$ has 6~states:
$r_0 \xrightarrow{a} r_1$, $r_0 \xrightarrow{a} r_2$,
$r_1 \xrightarrow{b} r_3$, $r_2 \xrightarrow{b} r_4$,
$r_2 \xrightarrow{c} r_5$; two $a$-successors, one
offering~$b$ only and the other offering $b$ and~$c$.

\begin{defi}[HML sublanguage hierarchy]
Labeled Hennessy--Milner logic~($\HML$) over alphabet~$\Sigma$ is built
from $\top$, $\bot$, conjunction, disjunction, negation, and the labeled
diamond $\langle a \rangle \varphi$ for each $a \in \Sigma$.  Four
sublanguages characterize the four spectrum levels:
\begin{itemize}
\item \emph{Trace formulas}~($\mathcal{O}$): sequential compositions of
  diamonds and disjunctions (no conjunction);
\item \emph{Positive formulas}~($\HML_{\mathrm{pos}}$): no negation
  (characterizes simulation);
\item \emph{Ready-simulation formulas}~($\HML_{\mathrm{ready}}$):
  positive formulas plus \emph{inability atoms}
  $\neg\langle a \rangle \top$ (characterizes ready simulation);
\item \emph{Full}~$\HML$ (characterizes bisimulation).
\end{itemize}
These form a strict inclusion chain
$\mathcal{O} \subsetneq \HML_{\mathrm{pos}} \subsetneq
 \HML_{\mathrm{ready}} \subsetneq \HML$,
witnessed by concrete formulas at each gap.
\end{defi}

\begin{thm}[Four-Level Labeled Spectrum Separation]
\label{thm:four-level-separation}
Over a label alphabet with $|\Sigma| \geq 3$, the four spectrum
levels are strictly separated by concrete HML formulas:
\begin{enumerate}
\item \emph{Trace $\neq$ Simulation.}\enspace
  The positive formula
  $\langle a \rangle(\langle b \rangle \top \wedge
  \langle c \rangle \top)$ holds in $a.(b{+}c)$
  but not in $a.b + a.c$: the conjunction under the
  diamond requires simultaneous $b$- and $c$-capabilities,
  which $q_1$ has but neither $p_1$ nor~$p_2$ has.
\item \emph{Simulation $\neq$ Ready Simulation.}\enspace
  The ready-simulation formula
  $\langle a \rangle(\langle b \rangle \top \wedge
  \neg\langle c \rangle \top)$ holds in $a.b + a.(b{+}c)$
  but not in $a.(b{+}c)$: $r_1$ can do~$b$ but
  not~$c$, while $q_1$ can do both.
\item \emph{Ready Simulation $\neq$ Bisimulation.}\enspace
  The formula $[a]\langle c\rangle\top$ separates
  $a.b + a.(b{+}c)$ from $a.(b{+}c)$
  (Definition~\ref{def:morleyization}).
\end{enumerate}
\end{thm}

\begin{rem}[Lindenbaum algebras and the subtopos chain]
The labeled Lindenbaum algebras of the separating examples have
cardinalities
$|\mathrm{Lind}(\Th{a.b{+}a.c}^{\Sigma})| =
 |\mathrm{Lind}(\Th{a.(b{+}c)}^{\Sigma})| = 5$
and
$|\mathrm{Lind}(\Th{a.b{+}a.(b{+}c)}^{\Sigma})| = 25$.
Thus the base Lindenbaum algebra is a \emph{necessary but not sufficient}
invariant for spectrum classification: different cardinalities guarantee
non-bisimilarity ($25 \neq 5$), but equal cardinalities do not imply
even simulation equivalence ($5 = 5$ yet trace~$\neq$~simulation).
The algebra captures which transitions are forced or excluded but is
blind to \emph{branching structure}: whether choices are committed
early or deferred.

Each HML sublanguage induces an equivalence relation on states (the
\emph{spectrum setoid}), and finer sublanguages give finer equivalences.
Each spectrum-level setoid corresponds to a Grothendieck topology on the
Lindenbaum frame, giving a chain of subtoposes:
\[
\Sh(\mathrm{Lind}, J_{\mathrm{trace}}) \supseteq
\Sh(\mathrm{Lind}, J_{\mathrm{sim}}) \supseteq
\Sh(\mathrm{Lind}, J_{\mathrm{readysim}}) \supseteq
\Sh(\mathrm{Lind}, J_{\mathrm{bisim}})
\]
The four-level separation proves that the first two inclusions are
strict; the third is witnessed by $[a]\langle c\rangle\top$
(Definition~\ref{def:morleyization}).  This is the topos-theoretic
reformulation of the van~Glabbeek spectrum: each step adds a class of
modal observations (conjunction, inability atoms, full negation),
refining the subtopos.
\end{rem}

\begin{rem}[Labels Break Symmetries]
\label{rem:labels-break-symmetries}
The labeled automorphism group---permutations $\sigma$ of states
preserving $s \xrightarrow{a} t \iff \sigma(s) \xrightarrow{a}
\sigma(t)$ for all actions~$a$---is a subgroup of the unlabeled
automorphism group.  For all three van~Glabbeek examples,
$|\mathrm{Aut}_{\Sigma}| = 1$: labels completely break all symmetries.
For instance, in $a.b + a.c$, the $a$-successors $p_1$
and~$p_2$ have different outgoing label sets ($\{b\}$ vs.\ $\{c\}$),
preventing the branch swap that would be valid in the unlabeled
projection.

The labeled kernel dichotomy extends:
$\ker(\varphi_{\Sigma})$ is trivial or maximal for connected reachable
systems.  With $|\mathrm{Aut}_{\Sigma}| = 1$, the kernel is
automatically trivial, placing all three systems in the aligned
or creative regime.  The spectrum hierarchy and symmetry breaking are
dual effects of labels: labels enrich the modal language
(enabling spectrum separation) while simultaneously constraining
automorphisms (collapsing the symmetry group).
\end{rem}

\label{sec:graded-atoms}

Propositional atoms $\step_a(s,t)$ encode \emph{whether} a transition
exists but not \emph{how} transitions distribute across branches.
Depth-bounded observation tree atoms cure this: the depth-1 graded
Lindenbaum algebras separate trace-equivalent from
simulation-inequivalent systems ($|\mathrm{Lind}_1(a.b{+}a.c)|
= 7 \neq 5 = |\mathrm{Lind}_1(a.(b{+}c))|$).  Bisping's
6-dimensional energy budget $(e_1, \ldots, e_6)$
\cite{bisping2025generalized} embeds each spectrum level as a specific
energy profile.  The full development, including the graded tower
stabilization and energy dimension sketch, is available in
the formalization repository.

\subsection{Separation and Symmetry}
\label{sec:separation-symmetry}
\label{sec:bridge-applied}

\noindent
The Lindenbaum automorphism group $\mathrm{Aut}(\mathrm{Lind})$
exhibits a symmetry trichotomy
whose key structural consequence is a group homomorphism from graph
automorphisms:

\begin{rem}[Symmetry Homomorphism]
\label{rem:symmetry-homomorphism}
The symmetry trichotomy lifts to a
natural group homomorphism
$\varphi \colon \mathrm{Aut}(\mathrm{graph}) \to \mathrm{Aut}(\mathrm{Lind})$
defined by atom permutation:
$\mathrm{step}(s,t) \mapsto \mathrm{step}(\sigma(s), \sigma(t))$,
descending to the Lindenbaum quotient.
The formula-level permutation $\texttt{permuteFormula}$ is proved by
structural induction on propositional formulas (6~lemmas, constructive);
the descent to the quotient is axiomatized.
The kernel $\ker(\varphi) = \{ \sigma \in \mathrm{Aut}(\mathrm{graph})
\mid \varphi(\sigma) = \mathrm{id} \}$ refines the trichotomy to a
structural statement: hub-spokes has $\ker = 1$ ($\varphi$ injective),
two-cycle has $\ker = \mathbb{Z}/2$ ($\varphi$ trivial), and the
asymmetric diamond has $\ker = 1$ (trivially, from trivial source).
\end{rem}

\begin{thm}[Kernel Dichotomy]
\label{thm:kernel-dichotomy}
For any connected reachable rooted transition system, the kernel of the
symmetry homomorphism
$\varphi \colon \mathrm{Aut}(\mathrm{graph}) \to
\mathrm{Aut}(\mathrm{Lind})$
satisfies
\[
\boxed{\quad \ker(\varphi) \;=\; 1 \quad\text{or}\quad
\ker(\varphi) \;=\; \mathrm{Aut}(\mathrm{graph}),
\qquad\text{no intermediate kernels} \quad}
\]
with trivial kernel when any nondeterministic state exists and maximal
kernel when all states are deterministic.
\end{thm}

\begin{proof}
\emph{Nondeterministic case} ($\ker(\varphi) = 1$).
An atom $\step(v,w)$ is \emph{non-extremal} when the edge $v \to w$
exists and $v$ has $\geq 2$ successors.  The \emph{atom singleton
lemma} shows that non-extremal atoms have singleton equivalence
classes in the Lindenbaum algebra: given distinct atoms $p \neq q$
with $p$ non-extremal, a separating model assigns $\top$ to $p$ and
$\bot$ to $q$ by routing nondeterministic choices appropriately,
so $p$ and $q$ represent distinct Lindenbaum classes.

From this, two closure properties follow.  First,
\emph{kernel elements fix nondeterministic states}: if $\sigma \in
\ker(\varphi)$ and $v$ has $\geq 2$ successors $w_1 \neq w_2$,
then $\varphi(\sigma) = \mathrm{id}$ forces
$\step(\sigma(v), \sigma(w_i)) = \step(v, w_i)$ for each~$i$, and
the singleton property pins $\sigma(v) = v$.
Second, the \emph{fixed-point set is forward-closed}: if
$\sigma(u) = u$ and $u \to v$ is an edge, then
$\step(u, \sigma(v)) = \step(\sigma(u), \sigma(v)) =
\step(u, v)$, forcing $\sigma(v) = v$.

Now fix $\sigma \in \ker(\varphi)$ and an arbitrary state~$v$.
By a \emph{last-common-ancestor argument}, $v$ is downstream of some
nondeterministic state~$n$: there exist $n$ and $w$ with $n \to w$
an edge, $n$ nondeterministic, and $w \to^{*} v$.  Kernel fixity gives
$\sigma(n) = n$; forward closure gives $\sigma(w) = w$; induction
along the path $w \to^{*} v$ gives $\sigma(v) = v$.

\emph{Deterministic case} ($\ker(\varphi) = \mathrm{Aut}(\mathrm{graph})$).
When every state has $\leq 1$ successor, each transition atom is forced:
$\step(v,w) = \top$ if $w$ is the unique successor of~$v$, and
$\step(v,w) = \bot$ otherwise.  No free generators remain, so
$\mathrm{Lind} \cong \mathrm{Bool}$.  Since
$\mathrm{Aut}(\mathrm{Bool}) = \{\mathrm{id}\}$, every graph
automorphism maps to $\mathrm{id}$ on $\mathrm{Lind}$, giving
$\ker(\varphi) = \mathrm{Aut}(\mathrm{graph})$.
\end{proof}

\begin{cor}[Three-Regime Classification]
\label{cor:regime-classification}
With the image $\mathrm{im}(\varphi) \cong
{\mathrm{Aut}(\mathrm{graph})}/{\ker(\varphi)}$ and cokernel
$\mathrm{coker}(\varphi) \cong {\mathrm{Aut}(\mathrm{Lind})}/{\mathrm{im}(\varphi)}$ (a group quotient, since $\mathrm{im}(\varphi)$
has index~$\leq 2$ in each anchor system), the dichotomy yields a three-regime classification:
The hub-spokes system has trivial kernel, image~$\mathbb{Z}/2$, and
trivial cokernel (aligned regime); the two-cycle has
kernel~$\mathbb{Z}/2$, trivial image, and trivial cokernel
(deterministic regime); the asymmetric diamond has trivial kernel,
trivial image, and cokernel~$\mathbb{Z}/2$ (creative regime).
The fourth case---non-trivial $\ker$ \emph{and} non-trivial
$\mathrm{coker}$---is impossible for connected reachable systems.
\end{cor}

The closest prior work is Hora's~\cite{hora2024topoi} application
of the bridge technique to automata and regular languages.

Concrete Lindenbaum-algebra computations on three small systems ($G_1$,
$G_2$, hub-spokes) confirm the stratification: pairwise separation via
cardinality, invariant transfer across spectrum levels, and the bridge
technique in action.

\begin{rem}[Colimit frame and Morita separation]
The depth-bounded Lindenbaum tower
$L_0 \hookrightarrow L_1 \hookrightarrow \cdots$ stabilizes at
depth~$n{-}1$ for an $n$-state LTS, yielding a colimit frame
$L_\infty \cong L_{n-1}$.  This is the frame-theoretic dual of the
Hennessy--Milner theorem: bisimilarity $= \bigcap_d {\sim}_d$
dualizes to $L_\infty = \bigcup_d L_d$.  The
Funayama--Nakayama theorem~\cite{funayama1942congruence} gives
$2^{|J(L_\infty)|}$ nuclei on~$L_\infty$, far exceeding the
13~named spectrum levels: on $a.b + a.c$, 14~of~16 subtoposes are unnamed;
on $a.(b{+}c)$, 7~of~8.

The Joyal--Nielsen--Winskel branch category
$\mathsf{Bran}_L$~\cite{joyal1996bisimulation} embeds into the category
of finitely presentable LTS; by the Comparison
Lemma~\cite[C2.2.3]{johnstone2002sketches}, the JNW presheaf topos arises
as a localization of~$\mathrm{Set}[\Th{\Msys}]$ when $\mathsf{Bran}_L$ is dense
(the synchronization tree condition).  A Morita test across all
$\binom{9}{2} = 36$ pairs of concrete LTS in the formalization finds
\textbf{no Morita equivalences}: every pair is separated by Lindenbaum
cardinality, automorphism count, or energy class count.  This
exhaustive verification (zero axioms, all separations decided
computationally) establishes the classifying topos as a complete
invariant for the anchor systems.
\end{rem}

\begin{rem}[The Energy--Topology Triangle]
\label{rem:energy-topology-triangle}
The construction establishes a three-way connection:
\[
  \text{Energy vectors (Bisping)}
  \;\xrightarrow{\bar{e}\,\mapsto\,j_{\bar{e}}}
  \text{Nuclei on } L
  \;\xleftrightarrow{\text{Johnstone}}
  \text{LT topologies on } \Sh(L).
\]
The composition sends each named process equivalence to a subtopos
of the classifying topos, with the order-reversal ensuring that
coarser equivalences (identifying more states) correspond to
\emph{larger} subtoposes (more sheaf conditions).

The bridge axiom $|L_{\bar{e}}| = |\mathrm{Fix}(j_{\bar{e}})|$ combines two
independent results from the process-algebra literature.
First, van~Glabbeek's logical characterization
theorems~\cite{vanglabbeek2001,vanglabbeek1993ltbt2} establish that
each level of the spectrum is captured by a sublanguage of
Hennessy--Milner logic: two states are $\bar{e}$-equivalent iff they
satisfy the same formulas in the corresponding HML fragment.
Second, Bisping~\cite{bisping2023cav} (Proposition~1) shows that the
sublanguage boundaries are uniformly parameterized by energy vectors,
so that the fragment for equivalence~$\bar{e}$ is exactly the formulas
of energy $\leq \bar{e}$.
The content encoded by this axiom is due to van~Glabbeek and Bisping;
our contribution is the topos-theoretic repackaging as a nucleus identity.
The bridge identity is verified for bisimulation (by a general theorem:
$j_{\mathrm{bisim}} = \mathrm{id}$ implies $\mathrm{Fix}(j_{\mathrm{bisim}}) = L$,
constructive) and for traces on $a.b + a.c$ (by explicit fixpoint counting: 5~fixpoints
on the 7-element algebra).  For $a.(b{+}c)$, the constant energy family
($|L_{\bar{e}}| = 5$ for all~$\bar{e}$) yields a complete bridge at all 13~levels.
Whether the full chain of reasoning can be carried out constructively for all
named equivalences (in particular, the logical characterization theorems rely on
classical König lemma arguments) remains open.
\end{rem}

The structural parallel with Kihara's LT topologies on the effective
topos, and four significant disanalogies, is discussed in
\S\ref{sec:related}.

\medskip
The spectrum lattice~$L_{30}$ was computed algebraically and
instantiated via labeled transitions.  We now confirm these results
from the topos side: internal logic and explicit Grothendieck topologies.

\section{The Topos Bridge}
\label{sec:internal-logic}

The spectrum lattice's algebraic structure (\S\ref{sec:spectrum-subtoposes})
was computed combinatorially.  We now bridge the lattice to the classifying
topos~$\mathrm{Set}[\Th{\mathrm{LTS}}]$, establishing explicit
Grothendieck topologies that enforce behavioral equivalences, extending
them to the full van~Glabbeek spectrum via energy budgets, and providing an
independent confirmation through the presheaf internal logic.

\subsection{The LTS Presheaf Topos}
\label{sec:lts-presheaf}

\begin{rem}[The LTS presheaf topos]
The per-system classifying toposes~$\ET{\Msys}$ of
\S\S\ref{sec:geometric}--\ref{sec:spectrum-subtoposes} are now
replaced by a single ambient workspace: the presheaf topos
$\mathrm{Set}[\Th{\mathrm{LTS}}] =
[\mathrm{f.p.LTS}_L^{\mathrm{op}}, \mathrm{Set}]$, whose objects are
set-valued functors on the category of finitely presentable labeled
transition systems.  Where the per-system toposes captured individual
theories, the presheaf topos carries internal logic and Grothendieck
topologies that encode \emph{relationships between} systems.

We have computed foundational properties of
$\mathrm{Set}[\Th{\mathrm{LTS}}]$ for $|L| = 1$:
it is De~Morgan (the right Ore condition holds trivially via the initial
object~$\varnothing$), non-Boolean, non-two-valued, connected, and
locally connected.  De~Morgan holds for \emph{any} presheaf
topos whose base category has an initial object: given a cospan
$A \to C \leftarrow B$, the empty object completes the square.
This is generic to presheaf toposes on categories of models of universal
Horn theories (finite graphs, posets, relational structures); the
properties specific to~$\Th{\mathrm{LTS}}$ begin with the
subobject classifier and the modal connection.
The subobject classifier $\Omega(G)$ has a unique
atom $\{!_G\}$ per object~$G$, yielding maximally simple
$\neg\neg$-topology.  The geometric fragment of the theory corresponds
precisely to the positive existential fragment of HML: diamond formulas
$\langle a\rangle\varphi$ are geometric, while box formulas
$[a]\varphi$ and negation are not.

The negation structure of~$\Omega(G)$ exhibits a striking parallel with
the lattice-algebraic results of~\S\ref{sec:spectrum-subtoposes}.  The
unique atom theorem forces $\neg S = \varnothing$ for every non-empty
sieve~$S$, mirroring $L_{30}$'s degenerate pseudocomplement
$\neg x = \bot$ for all $x \neq \bot$.  Both are instances of the
same mechanism: when every non-trivial element shares a common lower
bound (the initial morphism~$!_G$ in $\Omega$; traces in~$J(L_{30})$),
complementation collapses.  Yet despite this degenerate negation, the
Heyting \emph{implication} on~$\Omega(G)$, defined by
$(S_1 \to S_2)(f) \Leftrightarrow \forall g,\; f \circ g \in S_1
\Rightarrow f \circ g \in S_2$, remains non-trivial: we exhibit sieves
$S_1, S_2$ on the loop vertex~$\bullet_1$ where $S_1 \to S_2$ is
strictly intermediate (containing~$!_G$ but not~$\mathrm{id}_G$).
This is the presheaf-topos analogue of the lattice-algebraic discovery
$S \to F = \mathrm{IF}$: in both settings, negation carries zero
information while implication preserves the full logical structure.
The parallel between these two independently computed results, one
combinatorial (the lattice $L_{30}$, \S\ref{sec:spectrum-subtoposes}),
one categorical (the presheaf topos, this section), is strong evidence
for the $L_{30}$ embedding conjecture: the same algebraic pattern
(degenerate negation, non-trivial implication) appears in two
independent constructions, suggesting a common coframe-theoretic origin.
\end{rem}

\subsection{Explicit Grothendieck Topologies}
\label{sec:topos-backward-flow}
\label{sec:explicit-topologies}

The lattice-algebraic results of \S\ref{sec:spectrum-subtoposes} (Birkhoff
representation, bi-Heyting structure, co-Heyting decomposition) operate on the
finite spectrum lattice~$L_{30}$ via combinatorial computation.  We now establish
the first results in the reverse direction: explicit Grothendieck topologies on
a process-algebraic category whose sheaf conditions enforce named behavioral
equivalences.  Our main result is the \emph{spectrum bracket theorem}
(below), with the three-topology strict chain as its principal instantiation
and the partial $L_{30}$ embedding (Theorem~\ref{thm:partial-l30-embedding})
as its structural consequence.

\begin{rem}[The trace and bisimulation topologies]
We construct the first Grothendieck topology on
$\mathrm{f.p.LTS}_1$ whose sheaf condition enforces a named behavioral
equivalence.  Define the \emph{path digraph}
$P_n$ ($n + 1$~vertices, edges $i \to i{+}1$) and a
sieve~$S$ on~$G$ is \emph{$J_{\mathrm{trace}}$-covering} iff~$S$
contains every graph homomorphism $P_n \to G$ for all~$n$.  This
satisfies the three Grothendieck axioms (maximality, stability under
pullback, transitivity), with the key insight for transitivity being
that $\mathrm{id}_{P_n}$ is itself a path morphism.  The trace
separation theorem establishes:
two rooted digraphs are trace-equivalent (same set of achievable
path lengths) if and only if they have the same
$J_{\mathrm{trace}}$-stalks.  The path digraphs form a dense
subcategory under~$J_{\mathrm{trace}}$, so the
Comparison Lemma~\cite[C2.2.3]{johnstone2002sketches} applies:
$\Sh(\mathrm{f.p.LTS}_1, J_{\mathrm{trace}}) \simeq
\Sh(\mathrm{Path}, J_{\mathrm{trace}}|_{\mathrm{Path}})$.
The topology $J_{\mathrm{trace}}$ is strictly finer than the
$\neg\neg$-topology (= atomic topology: the base category satisfies the right Ore condition, so the presheaf topos is De~Morgan and the two topologies coincide):
the initial sieve $\{!_{\bullet_1}\}$ on the loop vertex is non-empty
but not $J_{\mathrm{trace}}$-covering, since it misses $P_1 \to \bullet_1$.

Replacing path digraphs~$P_n$ with finite rooted trees
(Figure~\ref{fig:test-objects}) yields the
bisimulation topology~$J_{\mathrm{bisim}}$: a sieve~$S$ on~$G$ is
$J_{\mathrm{bisim}}$-covering iff~$S$ contains every graph homomorphism
$T \to G$ for every rooted tree~$T$.  Since every path digraph is a
rooted tree, every $J_{\mathrm{bisim}}$-covering sieve is automatically
$J_{\mathrm{trace}}$-covering: $J_{\mathrm{bisim}} \subseteq
J_{\mathrm{trace}}$ as Grothendieck topologies.
The three Grothendieck axioms hold for~$J_{\mathrm{bisim}}$ by the same
self-probing argument: for transitivity, since~$T_0$ is a tree,
$\mathrm{id}_{T_0}$ is itself a tree morphism.
The inclusion is strict: the fan digraph $\mathrm{fan}(2)$ (root with 2~leaf
children) is a rooted tree whose identity is not in the path-generated sieve
(no root-preserving path morphism $P_n \to \mathrm{fan}(2)$ can hit both
leaves, since a path visits at most one child of the root), so the path-generated sieve on~$\mathrm{fan}(2)$
is $J_{\mathrm{trace}}$-covering but not $J_{\mathrm{bisim}}$-covering.

The two topologies are unified by a parametric \emph{observation class}
framework: for any class~$\mathcal{C}$ of rooted trees, a sieve is
$J_{\mathcal{C}}$-covering iff it contains all $\mathcal{C}$-morphisms.
\end{rem}

\begin{thm}[Spectrum Bracket]
\label{thm:spectrum-bracket}
For any observation class~$\mathcal{C}$ of rooted trees with
$\mathrm{paths} \subseteq \mathcal{C} \subseteq \mathrm{trees}$,
the observation-class topology satisfies
\[
\boxed{\quad J_{\mathrm{bisim}} \;\subseteq\; J_{\mathcal{C}}
\;\subseteq\; J_{\mathrm{trace}} \quad}
\]
The van~Glabbeek spectrum is an interval of Grothendieck topologies
on a fixed site, parameterized by observation classes.
\end{thm}

\begin{figure}[ht]
\centering
\begin{tikzpicture}[
    >=Stealth,
    state/.style={circle, draw, minimum size=6mm, inner sep=0pt,
                  font=\scriptsize},
    testlabel/.style={font=\small\bfseries, anchor=south},
    sublabel/.style={font=\small\itshape, anchor=north},
  ]
  \node[testlabel] at (-5.5, 1.6) {$P_1$};
  \node[state] (p1a) at (-5.5, 0.5) {$0$};
  \node[state] (p1b) at (-5.5,-0.7) {$1$};
  \draw[->] (p1a) -- (p1b);

  \node[testlabel] at (-3.5, 1.6) {$P_2$};
  \node[state] (p2a) at (-3.5, 0.5) {$0$};
  \node[state] (p2b) at (-3.5,-0.7) {$1$};
  \node[state] (p2c) at (-3.5,-1.9) {$2$};
  \draw[->] (p2a) -- (p2b);
  \draw[->] (p2b) -- (p2c);

  \node[testlabel] at (-1.5, 1.6) {$P_3$};
  \node[state] (p3a) at (-1.5, 0.5)  {$0$};
  \node[state] (p3b) at (-1.5,-0.7)  {$1$};
  \node[state] (p3c) at (-1.5,-1.9)  {$2$};
  \node[state] (p3d) at (-1.5,-3.1)  {$3$};
  \draw[->] (p3a) -- (p3b);
  \draw[->] (p3b) -- (p3c);
  \draw[->] (p3c) -- (p3d);

  \node at (0.1, -0.7) {\large$\cdots$};

  \node[testlabel] at (2.0, 1.6) {$\mathrm{fan}(2)$};
  \node[state] (f2r)  at (2.0, 0.5) {$r$};
  \node[state] (f2l1) at (1.2,-0.7) {$\ell_1$};
  \node[state] (f2l2) at (2.8,-0.7) {$\ell_2$};
  \draw[->] (f2r) -- (f2l1);
  \draw[->] (f2r) -- (f2l2);

  \node[testlabel] at (5.0, 1.6) {$T$};
  \node[state] (tr)  at (5.0, 0.5)  {$r$};
  \node[state] (tc1) at (3.8,-0.7)  {};
  \node[state] (tc2) at (6.2,-0.7)  {};
  \node[state] (tc3) at (3.0,-1.9)  {};
  \node[state] (tc4) at (3.8,-1.9)  {};
  \node[state] (tc5) at (4.6,-1.9)  {};
  \node[state] (tc6) at (6.2,-1.9)  {};
  \node[state] (tc7) at (3.8,-3.1)  {};
  \draw[->] (tr)  -- (tc1);
  \draw[->] (tr)  -- (tc2);
  \draw[->] (tc1) -- (tc3);
  \draw[->] (tc1) -- (tc4);
  \draw[->] (tc1) -- (tc5);
  \draw[->] (tc2) -- (tc6);
  \draw[->] (tc4) -- (tc7);
  \node at (3.8, -3.8) {$\vdots$};
  \node at (6.2, -2.6) {$\vdots$};

  \draw[decorate, decoration={brace, amplitude=5pt, mirror}]
    (-6.2, -3.6) -- (0.5, -3.6)
    node[midway, below=8pt, font=\small]
    {$J_{\mathrm{trace}}$: paths $\{P_n\}_{n \geq 0}$};
  \draw[decorate, decoration={brace, amplitude=5pt, mirror}]
    (-6.2, -4.8) -- (6.8, -4.8)
    node[midway, below=8pt, font=\small]
    {$J_{\mathrm{bisim}}$: all rooted trees};
\end{tikzpicture}
\caption{Test objects for the trace and bisimulation topologies.
Path digraphs~$P_n$ generate~$J_{\mathrm{trace}}$; all finite rooted
trees (e.g.\ the generic tree~$T$) generate~$J_{\mathrm{bisim}}$.
Since every path is a tree,
$J_{\mathrm{bisim}} \subseteq J_{\mathrm{trace}}$.  The inclusion is
strict: $\mathrm{id}_{\mathrm{fan}(2)}$ is not in the path-generated sieve
(no root-preserving path morphism $P_n \to \mathrm{fan}(2)$ can hit both
leaves, since a path visits at most one child of the root), so the path-generated
sieve on~$\mathrm{fan}(2)$ is $J_{\mathrm{trace}}$-covering but not
$J_{\mathrm{bisim}}$-covering.}
\label{fig:test-objects}
\end{figure}
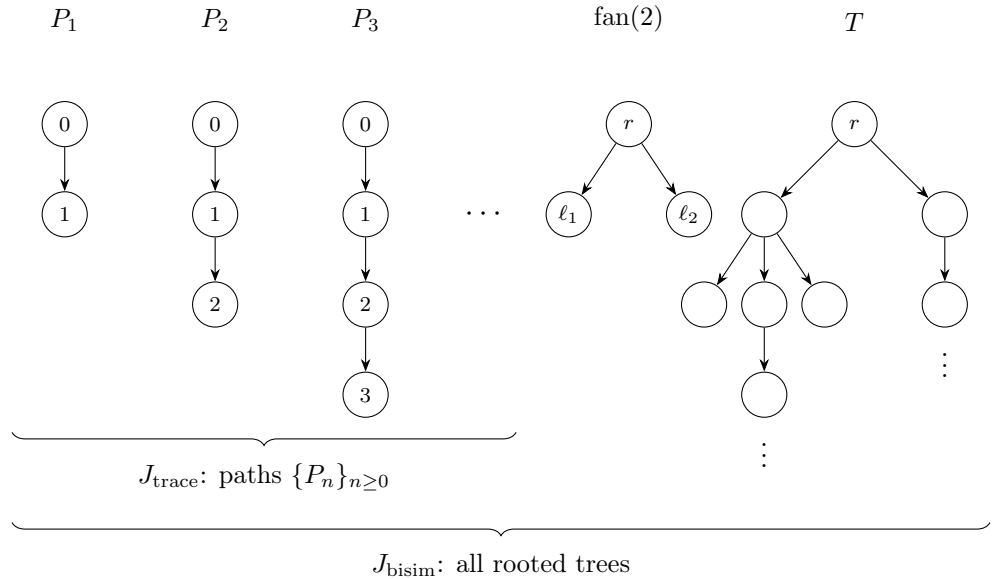

\noindent
A key disanalogy emerges: tree homomorphism \emph{existence} does not
characterize bisimulation (homomorphisms can collapse branches),
unlike the trace case where path homomorphism existence does characterize
trace equivalence.  The sieve condition (ALL tree morphisms must be present)
is essential.  A forward bisim--tree transfer theorem
(Hennessy--Milner forward direction) shows bisimulation equivalence implies
matching root-preserving tree hom existence, but the converse fails.
For $|L| = 1$, the spectrum collapses to at most two named topology levels:
$J_{\mathrm{trace}}$ and $J_{\mathrm{bisim}}$, since simulation equivalence
coincides with trace equivalence for unlabeled systems (forward simulation
reduces to matching path-length sets, which are downward-closed).

\begin{rem}[Labeled generalization and simulation topology]
The single-label case generalizes to \emph{labeled transition systems}
$\mathrm{FinLTS}(L)$ where edges carry labels from a finite alphabet~$L$.
Labeled homomorphisms preserve labels; the trace and bisimulation topologies
transfer directly: labeled trace words $w \in L^*$ and labeled rooted trees
serve as test objects for $J_{\mathrm{trace}}$ and $J_{\mathrm{bisim}}$
respectively, and the bracket theorem carries over.

For $|L| \geq 2$, a third topology becomes visible.
Tree-hom \emph{existence} characterizes simulation equivalence
(unlike bisimulation, where the full sieve structure is essential).
However, a naive observation-class approach to $J_{\mathrm{sim}}$ fails.
\end{rem}

\begin{thm}[Naive simulation covering instability]
\label{thm:naive-sim-instability}
The covering predicate ``there exists a tree homomorphism in~$S$'' satisfies
the maximality and transitivity axioms for a Grothendieck topology but
\emph{not} the stability axiom.  Let $f_R \colon \mathrm{traceLTS}[a] \to
\mathrm{fanLTS}[a,a]$ be the right-branch inclusion (mapping the non-root
vertex to vertex~$2$) and let $S$ be the sieve on $\mathrm{fanLTS}[a,a]$
generated by~$f_R$.  Then $S$ is naive-sim-covering: for any rooted tree~$T$ with a
homomorphism~$f$ into $\mathrm{fanLTS}[a,a]$, collapsing both branches
of~$f$ onto the unique non-root vertex of $\mathrm{traceLTS}[a]$ yields a
morphism~$k$ with $f_R \circ k \in S$.  Now pull $S$ back along the
left-branch inclusion $f_L$ (mapping the non-root to vertex~$1$).  The
identity on $\mathrm{traceLTS}[a]$ belongs to $f_L^*(S)$ iff
$f_L = f_R \circ k$ for some~$k$, but evaluating at the non-root gives
$f_L(1) = 1$ while $f_R(k(1)) \in \{0, 2\}$, a contradiction.  Hence
$f_L^*(S)$ is not naive-sim-covering.
\end{thm}

The correct $J_{\mathrm{sim}}$ is instead guaranteed by Caramello's
quotient-theory duality~\cite{caramello2017theories}: the geometric sequents
encoding positive HML (the simulation fragment of Hennessy--Milner logic)
generate a Grothendieck topology by the standard saturation construction.
Unlike $J_{\mathrm{trace}}$ and $J_{\mathrm{bisim}}$, whose covering
predicates are defined and all three Grothendieck axioms proved
constructively in Lean, the simulation topology is established by Caramello's
quotient-theory duality~\cite[Theorem~3.6]{caramello2017theories}:
geometric axioms generate covering conditions on the syntactic site
via the standard saturation construction, yielding a subtopos.
The proof is constructive within standard Grothendieck topos
theory; it is axiomatized in the Lean formalization because
Mathlib does not yet cover syntactic categories or sheafification.
The simulation topology is thus the only one of the three that requires
Caramello's quotient-theory machinery;
Theorem~\ref{thm:naive-sim-instability} shows this indirection is
provably necessary for the naive observation-class approach, a
structural feature of simulation's $(-1)$-truncation level, not a
limitation of the proof technique.

\begin{cor}[Three-Topology Strict Chain]
\label{cor:three-topology-chain}
For $|L| \geq 2$, the trace, simulation, and bisimulation topologies
form a strict chain of Grothendieck topologies on
$\mathrm{f.p.LTS}_L$:
\[
  J_{\mathrm{bisim}} \;\subsetneq\; J_{\mathrm{sim}} \;\subsetneq\;
  J_{\mathrm{trace}}.
\]
\end{cor}

\noindent
The standard counterexamples separate:
$a.b + a.c$ and $a.(b+c)$ are trace-equivalent but not simulation-equivalent;
$a.(b+c)$ and $a.b + a.c + a.(b+c)$ are simulation-equivalent but not bisimilar.

\begin{rem}[Truncation levels and the observation-class boundary]
\label{rem:truncation-levels}
The three-way split---constructive for trace, constructive for
bisimulation, axiomatized for simulation---reflects a truncation-level
boundary in the observation-class framework.

The observation-class topology $J_{\mathcal{C}}$ detects isomorphisms
of hom-presheaves $\mathrm{Hom}_*(T, -)$ for $T \in \mathcal{C}$.
The equivalence it induces operates at the \emph{sieve level}
($0$-truncated: the full set of morphisms must match).
Trace equivalence, by contrast, is $(-1)$-truncated: it requires only
that the same hom-sets are \emph{inhabited} (propositional existence).

For path digraphs, the sieve-level and existence-level conditions
happen to induce the same equivalence on objects, because the trace
separation theorem establishes that two rooted digraphs are
trace-equivalent iff they have the same Boolean support
$\{\, n : \mathrm{Hom}_*(P_n, G) \neq \varnothing \,\}$ of
path-hom presheaves.  The sieve carries additional multiplicity
information (the number of distinct paths of each length), but the
topology's covering condition (requiring \emph{all} path
homs) means this extra information does not affect the induced
equivalence.

For rooted trees, the existence-level condition (mutual simulation:
same tree-hom existence pattern) is strictly coarser than the
sieve-level condition (bisimulation: matching tree-hom sieves).
The observation-class framework operates at the sieve level and
therefore captures bisimulation natively, but simulation, genuinely
$(-1)$-truncated over non-rigid test objects, requires the
quotient-theory construction.

We note that the trace separation theorem characterizes trace
equivalence via the Boolean support of path-hom presheaves.
The identification
$\Sh(\mathrm{FinLTS}(L), J_{\mathrm{trace}}) \simeq
\mathrm{PSh}(L^*, {\leq}_{\mathrm{prefix}})$
from Section~\ref{sec:trace-comparison} provides evidence that the
sheaf topos captures strictly finer-than-trace equivalence: in the
Cattani--Winskel model~\cite{cattani2005profunctors}, open-map
bisimulation in $\mathrm{PSh}(L^*, {\leq}_{\mathrm{prefix}})$
recovers strong bisimulation, suggesting that the trace topos
remembers path multiplicities.  The precise identification of
the $J_{\mathrm{trace}}$-sheaf-theoretic equivalence, likely finer
than trace equivalence as the Cattani--Winskel correspondence
suggests, remains open.

\end{rem}

\begin{thm}[Partial $L_{30}$ embedding]
\label{thm:partial-l30-embedding}
The map sending each named process equivalence~$E$ to the Grothendieck
topology~$J_E$ on $\mathrm{f.p.LTS}_L$ is an order-embedding of the
three-element sub-poset
$\{\mathrm{trace}, \mathrm{simulation}, \mathrm{bisimulation}\}
\subset L_{30}$
into the lattice of subtoposes of\/
$\mathrm{Set}[\Th{\mathrm{LTS}}]$.
Explicitly: the strict chain
$J_{\mathrm{bisim}} \subsetneq J_{\mathrm{sim}} \subsetneq
J_{\mathrm{trace}}$
reflects the spectrum ordering
$\mathrm{bisimulation} > \mathrm{simulation} > \mathrm{trace}$
in~$L_{30}$, with the reversal due to the standard antitone
correspondence between finer equivalences and coarser topologies
(more covering sieves $=$ more identifications $=$ finer equivalence).
\end{thm}

\begin{proof}
Order-preservation follows from the general principle: if
equivalence~$E_1$ is finer than~$E_2$, then
$J_{E_1} \subseteq J_{E_2}$ (every $E_1$-covering sieve is
$E_2$-covering).  Order-reflection (injectivity) follows from the
strict separations: $J_{\mathrm{bisim}} \neq J_{\mathrm{sim}}$ and
$J_{\mathrm{sim}} \neq J_{\mathrm{trace}}$, witnessed by concrete
sieves on $\mathrm{fan}[a,a]$ and the van~Glabbeek counterexamples
respectively.
\end{proof}

This construction instantiates a general narrative: the presheaf topos
(no topology) sees only what homomorphisms preserve: forward
reachability.  Each Grothendieck topology adds a sheaf condition that
forces presheaves to respect some degree of \emph{backward lifting}.
The van~Glabbeek spectrum is a spectrum of backward-lifting requirements,
graduated by Grothendieck topologies on a fixed site.
Trace equivalence enforces the weakest backward constraint (path-length
agreement); bisimulation enforces full zig-zag; the intermediate named
equivalences correspond to intermediate topologies.
Joyal, Nielsen, and
Winskel~\cite{joyal1996bisimulation} vary the \emph{path category}
(using linear paths~$L^*$ for traces, branching
trees~$\mathrm{Bran}_L$ for bisimulation); our approach instead fixes
the site and varies the \emph{topology}, in the spirit of Caramello's
bridge technique~\cite{caramello2017theories}.
Despite 30~years of presheaf models for concurrency and 15~years of
Caramello's program, $J_{\mathrm{trace}}$ and $J_{\mathrm{bisim}}$ appear
to be the first explicit Grothendieck topologies on a process-algebraic
category whose sheaf toposes classify processes up to named
equivalences, a gap analogous to the one filled by
Hora~\cite{hora2024topoi} for automata and regular languages.
$J_{\mathrm{sim}}$ is established via Caramello's constructive
duality (axiomatized in Lean pending Mathlib's topos-theoretic
coverage); constructing explicit covering sieves remains
the central open problem.
The full lattice $L_{30}$ is conjectured to embed into the lattice of
subtoposes of $\mathrm{Set}[\Th{\mathrm{LTS}}]$, with each
element corresponding to a Grothendieck topology on the category of
finitely presentable LTS\@.
Theorem~\ref{thm:partial-l30-embedding} establishes the restriction
to the three-element sub-poset $\{T, S, B\}$.
Constructing topologies for the remaining 10~named equivalences
(failures, ready simulation, etc.)\ and
verifying the full embedding is the central open problem for
topos-geometric backward flow.

A further open problem concerns non-image-finite processes.
Each depth bound~$d$ yields a subtopos of
$\mathrm{Set}[\Th{\mathrm{LTS}}]$ whose points are
processes-up-to-depth-$d$-equivalence.  The colimit of this directed
system is a localic topos whose locale of opens is the Lindenbaum
algebra of full finitary HML\@.  For image-finite processes, this locale
is spatial (the Hennessy--Milner theorem).  For non-image-finite
processes, the locale may fail to be spatial in the operational sense:
consistent theories not realized by any process.  Characterizing this
gap topos-theoretically would yield a genuinely new process-theoretic
theorem requiring the topos framework to state.

\subsection{The Energy--Lawvere--Tierney Bridge}
\label{sec:energy-bridge}

The three topologies of \S\ref{sec:topos-backward-flow}---trace,
simulation, bisimulation---embed a three-element sub-poset of the van
Glabbeek spectrum into the lattice of subtoposes.  We now develop a
framework extending this three-point chain to the full 13-point spectrum
via an energy-parameterized construction, connecting Bisping's
6-dimensional energy budgets~\cite{bisping2023cav} to Grothendieck
topologies on $\mathrm{f.p.LTS}_L$.  The assignment $E \mapsto J_E$ is antitone;
the induced coframe map is injective for $|L| \geq 2$ and lands in a
coframe whose distributive structure is irreducibly topos-theoretic.
The construction follows a pipeline:
\[
\begin{tikzcd}[column sep=large]
  E \arrow[r, mapsto] & \mathcal{C}(E) \arrow[r, mapsto] & J_E \arrow[r, mapsto] & \Sh(\mathrm{f.p.LTS}_L, J_E)
\end{tikzcd}
\]
an energy budget determines test objects, which generate a
Grothendieck topology, whose sheaf category is a subtopos.
In the language of Caramello's bridge technique~\cite{caramello2017theories},
the pipeline lifts to the classifying topos:
\[
\begin{tikzpicture}[>=Stealth, baseline=(top.base)]
  \node (top) at (0, 1.8) {$\ET{\Msys}$};
  \node (left) at (-4.2, 0) {$\Th{\Msys}$};
  \node (right) at (4.2, 0) {$\Sh(\mathrm{f.p.LTS}_L,\, J_E)$};
  \draw[->, bend left=15] (left) to (top);
  \draw[->, bend right=15] (right) to (top);
\end{tikzpicture}
\]
The left pillar classifies the LTS structure; the right pillar
localizes to the subtopos determined by the energy budget~$E$.
Each named equivalence occupies a distinct position in the resulting
coframe of subtoposes.

\label{sec:energy-obs}
\label{sec:energy-topology}
\begin{defi}[Energy Grothendieck topology]
\label{def:energy-topology}

Bisping's spectroscopy framework parameterizes process equivalences by
\emph{energy budgets}~$E = (e_1, \ldots, e_6) \in
(\mathbb{N} \cup \{\infty\})^6$, controlling observation depth~$e_1$,
conjunction width~$e_2$, positive-deep depth~$e_3$, positive-other
depth~$e_4$, negative depth~$e_5$, and negation nesting~$e_6$.  The
13~named equivalences of the van~Glabbeek spectrum (from traces through
bisimulation) correspond to canonical energy vectors
(\S\ref{sec:spectrum-subtoposes}).

Each energy budget~$E$ determines an \emph{observation class}~$\mathcal{C}(E)$:
the class of finite rooted $L$-labeled trees whose depth satisfies
$\mathrm{depth}(T) \leq e_1$ and whose maximal branching satisfies
$\mathrm{branch}(T) \leq e_2$ (with the remaining four coordinates
constraining polarity and negation structure).
The observation class is monotone: $E_1 \leq E_2$ implies
$\mathcal{C}(E_1) \subseteq \mathcal{C}(E_2)$ (more energy reveals more tests).

The named observation classes recover familiar test objects:
$\mathcal{C}(\mathrm{traces}) = \{\text{path digraphs}\}$
(depth-unbounded, branching~${\leq}1$),
$\mathcal{C}(\mathrm{bisimulation}) = \{\text{all rooted trees}\}$
(all bounds $\infty$), and
$\mathcal{C}(\mathrm{simulation}) = \{\text{trees}\}$
(positive fragment, no negation: $e_5 = e_6 = 0$).
The positive-fragment classification partitions the 13~named
equivalences into 7~positive (including traces, simulation, ready
simulation) and 6~negative (including failures, readiness,
failure traces).

For each energy budget~$E$, define a covering predicate on
$\mathrm{f.p.LTS}_L$: a sieve~$S$ on~$G$ is \emph{$E$-covering}
if $S$ contains all homomorphisms from $\mathcal{C}(E)$-test objects into~$G$.
The three Grothendieck axioms are verified for each~$E$:
maximality is immediate (the maximal sieve contains all morphisms),
stability follows from preservation of $\mathcal{C}(E)$-test morphisms under
pullback, and transitivity uses the compositional structure of
energy-bounded observations.

For energy budget $E \in (\mathbb{N} \cup \{\infty\})^6$,
the \emph{energy topology}~$J_E$ is the Grothendieck topology on
$\mathrm{f.p.LTS}_L$ whose covering sieves are exactly the
$E$-covering sieves.
\end{defi}

The assignment $E \mapsto J_E$ is \emph{antitone}: finer observation
($E_1 \leq E_2$) yields more test objects ($\mathcal{C}(E_1) \subseteq \mathcal{C}(E_2)$),
hence stronger covering conditions, hence a finer topology
($J_{E_2} \subseteq J_{E_1}$, reversing the order).
The identification with existing topologies is exact:
$J_{\mathrm{traces}} = J_{\mathrm{trace}}$ and
$J_{\mathrm{bisimulation}} = J_{\mathrm{bisim}}$ (from
\S\ref{sec:topos-backward-flow}), recovering the spectrum bracket
theorem as a special case:
$J_{\mathrm{bisim}} \subseteq J_E \subseteq J_{\mathrm{trace}}$
for every named~$E$.  The chain
$J_{\mathrm{bisim}} \subsetneq J_{\mathrm{sim}} \subsetneq
J_{\mathrm{trace}}$ from Theorem~\ref{thm:partial-l30-embedding}
sits inside this parametric family.

\begin{thm}[Energy--LT bridge]
\label{thm:energy-lt-bridge}
\label{sec:bridge-theorem}
The energy topologies (Definition~\ref{def:energy-topology}) live on the
presheaf side ($\mathrm{f.p.LTS}_L$); the nuclei of
\S\ref{sec:spectrum-subtoposes} live on the algebraic side (Lindenbaum
algebras).  The \emph{energy--Lawvere--Tierney bridge} connects the two
via the standard correspondence between nuclei on a frame and
Grothendieck topologies on the associated thin
category~\cite[\S\,II.2]{johnstone1982stone}.

Write $\mathcal{V}$ for the 13-element van~Glabbeek poset of named
process equivalences, ordered by energy-budget inclusion ($\bar{e}_1
\leq \bar{e}_2$ iff every $\bar{e}_1$-formula is an
$\bar{e}_2$-formula), and define the \emph{spectrum--nucleus embedding}
$\nu(\bar{e}) = j_{\bar{e}}$:
\[
\boxed{\quad \nu \;\colon\; \mathcal{V}
\;\xhookrightarrow{\;\mathrm{antitone}\;}\;
\mathrm{Sub}(\mathrm{Set}[\Th{\mathrm{LTS}}])^{\mathrm{op}}
\quad}
\]
The map~$\nu$ is monotone (finer process equivalence $\mapsto$ smaller
subtopos in the coframe order) and injective for non-degenerate label
alphabets ($|L| \geq 2$).  The codomain carries a \emph{coframe}
structure: its dual is a frame, inheriting the frame structure of the
nucleus lattice.
For any finite label alphabet~$L$ with $|L| \geq 2$:
\begin{enumerate}
\item The subtopos order $\mathrm{Sub}(\mathrm{Set}[\Th{\mathrm{LTS}}])$
  carries a coframe structure, with the distributive law
  $A \sqcup (B_1 \sqcap B_2) =
   (A \sqcup B_1) \sqcap (A \sqcup B_2)$.
\item $J_{\mathrm{bisim}} \subsetneq J_{\mathrm{trace}}$:
  the two-point sub-poset
  $\{\mathrm{trace}, \mathrm{bisimulation}\}$ embeds into this coframe,
  with the strict separation constructive.
\item The strict chain
  $J_{\mathrm{bisim}} \subsetneq J_{\mathrm{sim}} \subsetneq
  J_{\mathrm{trace}}$
  embeds the three-point sub-poset
  $\{\mathrm{trace},\allowbreak\mathrm{simulation},\allowbreak\mathrm{bisimulation}\}$
  into this coframe
  (Theorem~\ref{thm:partial-l30-embedding},
  Corollary~\ref{cor:three-topology-chain}).
\item The coframe meet of possible-futures and failure-traces does not
  lie in the image of the spectrum:
  $\nu(\mathrm{PF}) \sqcap \nu(\mathrm{FT})
    \notin \mathrm{im}(\nu)$.
\item The embedding~$\nu$ extends to all 13~named equivalences of the
  van~Glabbeek spectrum, monotonely and compatibly with the energy ordering:
  $E_1 \leq E_2 \Longrightarrow \nu(E_2) \leq \nu(E_1)$.
\end{enumerate}
\end{thm}

\noindent
Items~1--2 are the paper's constructive core:
coframe structure, constructive
$J_{\mathrm{trace}} \neq J_{\mathrm{bisim}}$ strict
separation, and the distributive law.
Item~3 extends the two-point embedding to a three-point chain by
adding~$J_{\mathrm{sim}}$ via Caramello's duality;
item~4 establishes the PF${}\sqcap{}$FT escape witness.
Item~5 extends from 3~to all 13~named equivalences via the
energy--topology framework (Remark~\ref{rem:bridge-axiom-boundary}).

\begin{rem}[Axiom boundary]
\label{rem:bridge-axiom-boundary}
The axioms in item~5 mark the interface between the constructions of
this paper and published results of Caramello, van~Glabbeek, and
Bisping.  The mathematical content is standard Grothendieck topos
theory; the axiomatization reflects the state of Mathlib's
topos-theoretic coverage, not a gap in the mathematics.
The $L_{30}$ lattice operations in
\S\ref{sec:spectrum-subtoposes} are decided on the finite 30-element
lattice.  The full explicit construction of
covering sieves remains open at the~$J_{\mathrm{sim}}$ level; see
\S\ref{sec:conclusion}.
\end{rem}

\begin{cor}[Coframe computations]
\label{cor:coframe-computations}
\label{sec:coframe-computations}
The coframe structure of Theorem~\ref{thm:energy-lt-bridge} enables
explicit computations beyond the named spectrum:
\begin{enumerate}
\item \emph{Single-label collapse.}\enspace
  For $|L| = 1$, the 13~named equivalences induce at most 2~distinct
  Grothendieck topologies on $\mathrm{f.p.LTS}_1$:
  $J_{\mathrm{trace}}$ and $J_{\mathrm{bisim}}$, since simulation
  coincides with trace equivalence for unlabeled
  systems~\cite{vanglabbeek2001}.
  (Per-system Lindenbaum algebras may still distinguish enabledness
  from trace equivalence, but this is a system-level distinction,
  not a site-level one.)
\item \emph{Two-label faithfulness.}\enspace
  For $|L| \geq 2$, all 13~named equivalences produce distinct
  subtoposes.  The proof uses order-reflection:
  $\nu(E_2) \leq \nu(E_1) \Rightarrow E_1 \leq E_2$, combined with
  the known incomparable pairs in the van~Glabbeek spectrum.
\item \emph{Incomparability transfer.}\enspace
  PF and FT remain incomparable in
  $\mathrm{Sub}(\mathrm{Set}[\Th{\mathrm{LTS}}])$.
  Their coframe meet $\mathrm{PF} \sqcap \mathrm{FT}$ is an
  \emph{unnamed coframe element}, witnessing that the subtopos lattice
  is strictly larger than the spectrum.
\item \emph{Lattice closure transfer.}\enspace
  The 30-element lattice closure from \S\ref{sec:spectrum-subtoposes}
  (13~named + 17~unnamed, computed by three-round
  $\{\sqcap,\sqcup\}$-stabilization) transfers to the coframe level,
  with 17~unnamed coframe elements.
\end{enumerate}
\end{cor}

\subsection{Irreducibility of the Topos-Theoretic Route}
\label{sec:irreducibility}

Three features of Theorem~\ref{thm:energy-lt-bridge} are irreducibly
topos-theoretic.

\paragraph{Coframe structure.}
The distributive law in item~1 follows from the frame structure of
nuclei on the Lindenbaum algebra (Isbell~1972, Simmons~1978), which
depends on the Heyting algebra structure of the subobject classifier.
The van~Glabbeek spectrum is a finite poset with no known lattice
structure; the coframe law is genuinely new algebraic information that
cannot be derived from energy vectors alone.

\paragraph{Unnamed subtoposes.}
The element $\mathrm{PF} \sqcap \mathrm{FT}$ in item~4 of the theorem
has no process-algebraic name, and the 17~unnamed elements of the
lattice closure correspond to geometric axioms (determinism,
image-finiteness, mixed constraints) that cross-cut the HML hierarchy
and are invisible to any energy budget; their existence is provable
only through the nucleus construction.

\paragraph{Distributivity witness.}
The unnamed meet $\mathrm{PF} \sqcap \mathrm{FT}$ satisfies the
coframe distributive law despite having no process-algebraic
description; both its existence and its algebraic properties require
passage through the topos.

\medskip\noindent
Taken together, Theorem~\ref{thm:energy-lt-bridge} is the first result
that treats the \emph{entire} van~Glabbeek spectrum as a
topos-theoretic object: a finite sub-poset of an infinite coframe that
is an intrinsic invariant of the classifying topos and that cannot be
accessed from within process algebra.

\subsection{The Geometric Closure Theorem}
\label{sec:geometric-closure-section}
\label{sec:geometric-closure}

The \emph{internal logic} of $\mathrm{Sub}(U_T)$, where $U_T$ is
the underlying-set functor of the generic model, provides an independent
confirmation.  Geometric subobjects (subfunctors defined by positive
existential HML formulas) form a Heyting subalgebra whose implication
has a concrete process-algebraic characterization via \emph{free extensions}.

\label{sec:free-extension}
\begin{defi}[Positive existential HML]

Define \emph{positive existential} HML (henceforth $\HML^+$) as the
fragment generated by $\top$, conjunction $\varphi_1 \wedge \varphi_2$,
and diamond $\langle a \rangle \varphi$.  Every formula in $\HML^+$
defines a subfunctor $S_\varphi$ of the underlying-set presheaf by
$S_\varphi(G,v) \Leftrightarrow \varphi.\mathsf{satisfies}(G,v)$.
The positive existential fragment is closed under all positive
operations, and satisfiability is monotone: if $h \colon G \to H$ is a
homomorphism and $S_\varphi(G,v)$ holds, then $S_\varphi(H,h(v))$
holds; this is exactly the subfunctor property.
\end{defi}

\begin{lem}[Free Extension]
\label{lem:free-extension}
For every $(G, v)$ in $\mathrm{FinLTS}(L)$ and every $\varphi \in \HML^+$,
there exists a labeled transition system $\mathrm{Ext}(G,v,\varphi)$ and
a root-preserving homomorphism $\iota \colon G \hookrightarrow
\mathrm{Ext}(G,v,\varphi)$ such that
$\varphi.\mathsf{satisfies}(\mathrm{Ext}(G,v,\varphi), \iota(v))$.
The construction adjoins a tree of \emph{witness vertices} to~$G$,
with one fresh successor per diamond subformula, adding the minimal
transitions required for satisfaction.
\end{lem}

\noindent
The construction is entirely explicit: vertices of $\mathrm{Ext}(G,v,\varphi)$
are the coproduct $G \oplus W_\varphi$ where $W_\varphi$ is a recursive
index type (one witness per diamond in~$\varphi$), and transitions
consist of root edges connecting $\iota(v)$ to the first witness,
internal edges within the witness tree, and all original edges of~$G$.
The proof is by structural induction on~$\varphi$ and is fully constructive.

\label{sec:negation-collapse}
\begin{prop}[Negation Collapse]
\label{prop:negation-collapse}
In the presheaf topos, the internal negation of a subfunctor~$S$ is
\[
  (\neg S)(G,v) \;=\; \forall\, H,\; \forall\, h \colon G \to H,\;
  \neg S(H, h(v)).
\]
This universal quantification over all extensions is extremely strong:
it asks that no homomorphic image of~$v$ can satisfy~$S$.
For every satisfiable $\varphi \in \HML^+$, we have
$\neg S_\varphi = \bot$ and $\neg\neg S_\varphi = \top$.
\end{prop}

\noindent
The proof is immediate from the Free Extension Lemma
(Lemma~\ref{lem:free-extension}): given any $(G,v)$,
the extension $\mathrm{Ext}(G,v,\varphi)$ provides a homomorphic image
of~$v$ satisfying~$\varphi$, so $(\neg S_\varphi)(G,v)$ fails.
Since this holds for all $(G,v)$, the negation is the bottom subfunctor.
The double negation density $\neg\neg S_\varphi = \top$ follows
immediately.  This connects to the unique-atom and collapsed-negation
phenomena on $\Omega(G)$ observed in
\S\ref{sec:spectrum-subtoposes}: the presheaf topos is so far from
Boolean that every non-trivial geometric subfunctor has trivial negation.

The central result of this subsection computes the Heyting implication
of geometric subfunctors.

\begin{thm}[Geometric Closure]
\label{thm:geometric-closure}
For $\varphi, \psi \in \HML^+$,
\[
\boxed{\quad (S_\varphi \to_\Omega S_\psi)(G,v) \;\;\Longleftrightarrow\;\;
  \psi.\mathsf{satisfies}(\mathrm{Ext}(G,v,\varphi),\, \iota(v)) \quad}
\]
\end{thm}

\noindent
The Heyting implication $S_\varphi \to_\Omega S_\psi$ is defined
pointwise as
$(S_\varphi \to_\Omega S_\psi)(G,v) = \forall\, H,\;
\forall\, h \colon G \to H,\; S_\varphi(H, h(v))
\Rightarrow S_\psi(H, h(v))$,
and the theorem reduces this universal quantification to a
\emph{single test} against the free extension.

\begin{proof}
($\Rightarrow$) If $(G,v)$ satisfies the implication, apply it to
$H = \mathrm{Ext}(G,v,\varphi)$ with $h = \iota$: since the extension
satisfies~$\varphi$ at $\iota(v)$, the implication yields
$\psi$ at~$\iota(v)$.
($\Leftarrow$) Suppose $\psi$ holds at $\iota(v)$ in the extension.
For any $H$ and $h \colon G \to H$ with $S_\varphi(H, h(v))$, the
satisfaction witness $S_\varphi(H, h(v))$ yields a mediating homomorphism
$m \colon \mathrm{Ext}(G,v,\varphi) \to H$ with $m \circ \iota = h$,
mapping each witness vertex to the corresponding existential witness in~$H$.
Since $\psi$ is positive existential (hence preserved by homomorphisms),
$S_\psi(H, h(v))$ follows.
\end{proof}

The Geometric Closure Theorem shows that geometric subobjects form a
Heyting subalgebra: the implication of two $\HML^+$-definable
subfunctors is again computable by testing against a single canonical
extension.  This is an \emph{irreducibly topos-theoretic} result:
it requires both the presheaf Heyting implication formula (internal
to the topos) and the canonical test extension construction.
No purely process-algebraic proof is known.

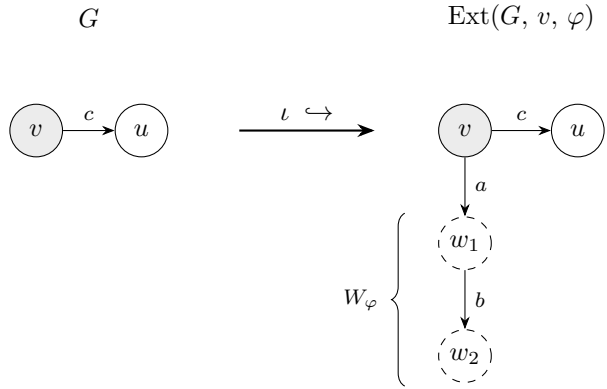
\begin{figure}[ht]
\centering
\begin{tikzpicture}[
    >=Stealth,
    state/.style={circle, draw, minimum size=7mm, inner sep=0pt,
                  font=\small},
    fresh/.style={circle, draw, dashed, minimum size=7mm, inner sep=0pt,
                  font=\small},
  ]
  \node[font=\small\bfseries] at (-3.5, 1.5) {$G$};
  \node[state, fill=gray!15] (v) at (-4.2, 0) {$v$};
  \node[state] (u) at (-2.8, 0) {$u$};
  \draw[->] (v) -- node[above, font=\scriptsize] {$c$} (u);

  \draw[->, thick] (-1.5, 0) --
    node[above, font=\small] {$\iota\;\hookrightarrow$} (0.3, 0);

  \node[font=\small\bfseries] at (2.25, 1.5)
    {$\mathrm{Ext}(G,\, v,\, \varphi)$};

  \node[state, fill=gray!15] (v2) at (1.5, 0) {$v$};
  \node[state] (u2) at (3.0, 0) {$u$};
  \draw[->] (v2) -- node[above, font=\scriptsize] {$c$} (u2);

  \node[fresh] (w1) at (1.5, -1.5) {$w_1$};
  \node[fresh] (w2) at (1.5, -3.0) {$w_2$};
  \draw[->] (v2) -- node[right, font=\scriptsize] {$a$} (w1);
  \draw[->] (w1) -- node[right, font=\scriptsize] {$b$} (w2);

  \draw[decorate, decoration={brace, amplitude=5pt, mirror}]
    (0.7, -1.1) -- (0.7, -3.4)
    node[midway, left=6pt, font=\scriptsize] {$W_\varphi$};
\end{tikzpicture}
\caption{The free extension construction
  (Lemma~\ref{lem:free-extension}).  Given an LTS~$G$ with marked
  vertex~$v$ and a positive existential formula
  $\varphi = \langle a\rangle\langle b\rangle\top$, the extension
  $\mathrm{Ext}(G, v, \varphi)$ adjoins a witness tree~$W_\varphi$
  (dashed nodes) to~$G$ via the injection~$\iota$.  Each diamond
  subformula contributes one fresh vertex ($w_1$ for
  $\langle a\rangle$, $w_2$ for $\langle b\rangle$), connected by the
  minimal edges required for satisfaction at~$\iota(v)$.  The original
  edges of~$G$ are preserved.}
\label{fig:free-extension}
\end{figure}

\begin{cor}[Subfunctor Heyting Adjunction]
\label{cor:himp-adjunction}
For subfunctors $S_1, S_2, T$ of the representable presheaf
$\mathbf{y}(-)$ on $C_{\mathrm{fp}}$, the Heyting implication
satisfies the standard adjunction:
\[
S_1 \wedge T \;\leq\; S_2
\quad\Longleftrightarrow\quad
T \;\leq\; (S_1 \to_\Omega S_2),
\]
where the partial order is pointwise inclusion.
\end{cor}

\label{sec:four-regimes}
\begin{rem}[Explicit computations: four regimes]

Over the alphabet $L = \{a, b\}$, the Geometric Closure Theorem yields
four qualitatively distinct regimes for diamond-formula implications at
depth~$\leq 2$:

\begin{enumerate}
\item \textbf{Independent:}
  $\langle a \rangle\top \to \langle b \rangle\top = \langle b \rangle\top$.
  Extending by an $a$-successor adds no $b$-transitions; the implication
  reduces to~$\psi$ itself.

\item \textbf{Non-trivial residual:}
  $\langle b \rangle\top \to (\langle a \rangle\top \wedge
  \langle b \rangle\top) = \langle a \rangle\top$.
  The extension auto-satisfies $\langle b \rangle\top$ (the premise);
  the residual strips the freely provided part, leaving $\langle a \rangle\top$.

\item \textbf{Depth-increasing:}
  $\langle a \rangle\top \to \langle a \rangle\langle b \rangle\top
  = \langle a \rangle\langle b \rangle\top$.
  The shallow extension (single $a$-successor) cannot satisfy the deeper
  formula; the implication equals~$\psi$.

\item \textbf{Entailment:}
  $\langle a \rangle\langle b \rangle\top \to \langle a \rangle\top
  = \top$.
  The stronger premise auto-provides the weaker conclusion; the
  implication is $\top$ at every vertex of every LTS.
\end{enumerate}

\noindent
These four regimes exhaust the logical possibilities for diamond-formula
implications at depth~$\leq 2$ and illustrate how the topos computes the
``forced part,'' the residual that remains after stripping away freely
extensible structure.  Each computation is a single evaluation of
$\psi.\mathsf{satisfies}(\mathrm{Ext}(G,v,\varphi), \iota(v))$.
\end{rem}

\subsection{The trace Comparison Lemma}
\label{sec:trace-comparison}

The trace topology on $\mathrm{FinLTS}(L)$ admits a restriction to the
full subcategory $\mathrm{TraceLTS}$ whose objects are linear chain LTS, one
object $\mathrm{traceLTS}(w)$ for each word $w \in L^*$.  The density
condition (every object in $\mathrm{FinLTS}(L)$ is covered by trace
objects) is proved constructively: for every $G : \mathrm{FinLTS}(L)$,
the sieve generated by morphisms factoring through trace objects is
$J_{\mathrm{trace}}$-covering (constructive).
The general Comparison Lemma~\cite[C2.2.3]{johnstone2002sketches}
then yields:
\[
  \Sh(\mathrm{FinLTS}(L),\, J_{\mathrm{trace}})
  \;\simeq\;
  \Sh(\mathrm{TraceLTS},\, J_{\mathrm{trace}}|_{\mathrm{TraceLTS}}).
\]
Root-preserving homomorphisms between trace objects correspond to prefix
embeddings: $\mathrm{traceLTS}(w_1) \to \mathrm{traceLTS}(w_2)$ exists
iff $w_1 \leq_{\mathrm{prefix}} w_2$, and when it exists it is unique.
The restricted topology on $\mathrm{TraceLTS}$ is therefore the canonical
topology on the prefix poset $(L^*, \leq_{\mathrm{prefix}})$, and
sheaves on a poset with the canonical topology are presheaves.
Thus the trace topos is equivalent to
$\mathrm{PSh}(L^*, \leq_{\mathrm{prefix}})$, precisely the
Cattani--Winskel presheaf model for labeled event
structures~\cite{cattani2005profunctors}.

This identification connects 30~years of presheaf models for
concurrency (Joyal--Nielsen--Winskel~\cite{joyal1996bisimulation},
Cattani--Winskel~\cite{cattani2005profunctors},
Eberhart--Hirschowitz--Seiller~\cite{eberhart2017sheaf}) to the
topos-geometric framework of this paper: the abstract sheaf topos
constructed from Grothendieck topologies on $\mathrm{FinLTS}(L)$
restricts, via the Comparison Lemma, to the concrete presheaf category
on a combinatorial poset.
(Formalized: \texttt{comparison\_lemma\_synthesis}; the general
Comparison Lemma, prefix-hom correspondence, and Cattani--Winskel
connection are axiomatized.)

\section{Related Work}
\label{sec:related}

\paragraph{Preservation theorems.}
Van Benthem's theorem~\cite{vanbenthem1976modal} characterizes the
bisimulation-invariant fragment of first-order logic as modal logic.
Otto~\cite{otto2004elementary} gave an elementary proof avoiding compactness.
Rosen~\cite{rosen1997modal} showed the characterization holds over finite
structures.  Rossman~\cite{rossman2008homomorphism} proved the
homomorphism-preservation theorem survives restriction to finite structures.
Celani and Jansana~\cite{celani1999priestley} characterize the
positive-bisimulation-invariant fragment of positive first-order logic as
positive modal logic.  Our Theorem~\ref{thm:gvb} establishes the geometric
(positive existential) case for pointed finite structures.  Recent independent work by Wa{\l}\k{e}ga and
Cuenca~Grau~\cite{walega2026preservation} proves that over finite pointed
Kripke models invariant under bounded unravellings, preservation under
homomorphisms equals definability in existential positive modal
logic, essentially the modal-theoretic substance of the geometric case
under the standard translation.
Their framework proceeds through the standard translation (modal logic
to first-order), while our Theorem~\ref{thm:gvb} works natively in
geometric logic over classifying toposes; the two results are complementary.

\paragraph{Bisimulation and modal logic.}
Hennessy and Milner~\cite{hennessy1985algebraic} introduced HML and proved
the characterization theorem for image-finite systems: two states are
bisimilar iff they satisfy the same HML formulas.
Park~\cite{park1981} and Milner~\cite{milner1989} developed the relational
theory of bisimulation.
Van Glabbeek~\cite{vanglabbeek2001} surveys the linear time--branching time
spectrum for concrete sequential processes.
Bisping~\cite{bisping2025generalized,bisping2023cav} recasts the spectrum as
an energy-game hierarchy over 13~equivalences (replacing completed trace,
completed simulation, and possible-worlds with impossible futures, revivals,
and enabledness); \S\ref{sec:spectrum} lifts this hierarchy to nuclei
on Lindenbaum algebras.

\paragraph{Categorical and coalgebraic approaches.}
Joyal, Nielsen, and Winskel~\cite{joyal1996bisimulation} characterized
bisimulation via open maps in presheaf categories, a different setting from
ours (the presheaf topos is the ambient category, not the classifying topos
of a per-system theory).
Abramsky~\cite{abramsky1991domain} showed that in a domain-theoretic
setting, propositional geometric logic (the observable properties)
characterizes bisimulation: two processes are bisimilar iff they satisfy the
same observable geometric properties.  Our work concerns first-order
geometric logic over per-system theories, a technically different framework,
but the conceptual relationship is precise.  At the propositional level,
geometric logic and bisimulation see the same thing; Abramsky's result
provides the baseline.  At the first-order level, quantifier structure
introduces identification constraints that bisimulation can break:
universal variables in antecedents can be equated in consequents
(Mechanism~1), existential variables can share targets across branches
(Mechanism~2), and variables can be identified with their sources
(Mechanism~3).  Our three-mechanism taxonomy is thus a classification of
exactly how first-order quantification exceeds the propositional case.
Abramsky and Shah~\cite{abramsky2021structure} introduced game comonads,
and Abramsky and Reggio~\cite{abramsky2023arboreal} axiomatized them via
\emph{arboreal categories}, deriving equi-resource preservation
theorems~\cite{abramsky2024preservation}.  Our energy-to-nucleus map
is structurally parallel: they encode resource-bounded equivalences as
comonadic adjunctions on a universal presheaf category; we encode them as
nuclei on per-system Lindenbaum algebras (= Lawvere--Tierney topologies
on classifying toposes).  A deeper difference concerns truncation:
arboreal coKleisli morphisms operate at level~$0$ (full function data),
as do our observation-class topologies for $J_{\mathrm{trace}}$ and
$J_{\mathrm{bisim}}$.  Simulation, genuinely $(-1)$-truncated, requires
the quotient-theory construction
(Remark~\ref{rem:truncation-levels}), an obstruction the arboreal
framework avoids by not distinguishing simulation from bisimulation.
Arboreal categories currently cover $\sim$4 of the 13~named equivalences
\cite{abramsky2023linear}; our framework targets all~13.
Bezhanishvili, de Groot, and Venema~\cite{bezhanishvili2022coalgebraic}
study coalgebraic geometric logic.
Ghilardi and Zawadowski~\cite{ghilardi2002sheaves} connect sheaves, games,
and model completions.
Eberhart, Hirschowitz, and Seiller~\cite{eberhart2017sheaf} provide the closest
precedent for using Grothendieck topologies to encode behavioral
equivalences.  Working in a presheaf category on a category of terms
(interaction trees for the $\pi$-calculus), they construct a Grothendieck
topology whose sheaf condition enforces bisimulation: two terms have
isomorphic sheafifications iff they are bisimilar.  Their construction
directly obtains bisimulation, the finest process equivalence, from a
single topology.  Our approach stratifies the \emph{entire} spectrum:
multiple topologies $J_E$ on a fixed site, one per equivalence, with the
topology lattice reflecting the spectrum ordering
(Theorem~\ref{thm:partial-l30-embedding}).  The technical difference is
that Eberhart--Hirschowitz--Seiller vary the terms (interaction trees of varying
depth) while keeping the topology fixed to bisimulation; we fix the site
($\mathrm{f.p.LTS}_L$) and vary the topology across the full spectrum.

\paragraph{Topos theory.}
Caramello~\cite{caramello2017theories} develops the systematic use of
classifying toposes as ``bridges'' between mathematical theories.
Johnstone~\cite{johnstone2002sketches} provides the foundational reference
for classifying toposes and the Comparison Lemma.
Makkai and Reyes~\cite{makkai1977first} established the connection between
geometric theories and their syntactic categories.

\paragraph{Realizability and the effective topos.}
Hyland~\cite{hyland1982effective} constructed the effective topos, whose
internal logic validates Church's Thesis.
Hyland, Johnstone, and Pitts~\cite{hyland1980tripos} introduced triposes as a
uniform framework encompassing both provability and realizability toposes.
Van Oosten~\cite{vanoosten2008realizability} provides a comprehensive account
of realizability categories and their categorical properties.
Longley~\cite{longley2005notions} identified a fundamental obstruction: even
PCAs computing the same first-order functions yield inequivalent realizability
toposes, the higher-type analogue of our Second Separation.
Bauer~\cite{bauer2006synthetic} developed synthetic computability theory inside
the effective topos, deriving Rice's theorem and the recursion theorem from
axioms (Number Choice, Enumerability, Markov's Principle) without reference to
a specific machine model.
\paragraph{Precedent: Kihara's LT topologies.}
The closest structural parallel is Kihara~\cite{kihara2023lt}, who
connects bilayer imperfect-information games for Turing reducibility
to Lawvere--Tierney topologies on the effective
topos~$\mathrm{Eff}$.  Both constructions use order-reversing maps
from game budgets to closure operators on a subobject classifier:
Kihara sends Turing-degree budgets to LT topologies on
$\Omega_{\mathrm{Eff}}$; we send energy vectors to nuclei on finite
Lindenbaum frames~$L_d$.  A key divergence is algebraic: Kihara
achieves a bijection on an uncountable~$\Omega$ with rich negation
($x \wedge \neg x \neq \bot$ for many~$x$); we achieve an embedding
into a finite, negation-collapsed lattice ($\neg x = \bot$ for all
$x \neq \bot$) where non-Booleanness comes entirely from the Heyting
implication.  Both witness non-Booleanness through Grothendieck
topologies, but via complementary algebraic mechanisms.

\paragraph{Graded monads and the process spectrum.}
Dorsch, Milius, and Schr\"oder~\cite{dorsch2019graded} showed that the full
van~Glabbeek spectrum arises from graded monads: a graded monad
$(T_n)_{n \in \mathbb{N}}$ on~$\mathrm{Set}$ parameterizes look-ahead depth,
and a generic Hennessy--Milner theorem holds for each grading.  Our
Section~\ref{sec:spectrum} embeds this spectrum as a subtopos lattice via
Caramello's duality theorem, extending the coalgebraic perspective to
topos-theoretic invariants.
The energy-to-nucleus map of \S\ref{sec:spectrum} makes the grading explicit:
each energy vector $\bar{e}$ induces a Lawvere--Tierney topology $j_{\bar{e}}$,
turning the graded spectrum into a lattice of subtoposes.

\paragraph{Automata topoi.}
Hora~\cite{hora2024topoi,hora2025topoi} constructs four topoi for regular
languages using Caramello's bridge technique, the closest methodological
parallel.  Structurally, automata have no branching, so
acceptance is inherently $(-1)$-truncated and all four topologies arise
from observation classes with no quotient-theory detour.  Our
transition-system setting, where path multiplicity matters (tree homs vs
tree-hom existence), introduces the truncation-level split that
non-trivially stratifies the spectrum.

\paragraph{Database theory.}
Yannakakis~\cite{yannakakis1981algorithms} introduced acyclic hypergraphs
for conjunctive queries.  The tree-shaped fragment characterizing the
bisimulation-invariant geometric formulas corresponds to acyclic conjunctive
queries over the binary step relation.

\section{Conclusion}
\label{sec:conclusion}

This paper develops the foundations of a unification between topos
theory and process algebra, in which the question ``when should two
computational systems be considered equivalent?'' receives a single
structural answer: two finitely presentable systems are equivalent,
relative to a choice of observations, when their sheafified Yoneda
images become isomorphic in the subtopos determined by that choice.
The van~Glabbeek spectrum, traditionally a hierarchy of
behavioral equivalences, each justified by its own operational or
logical intuition, is revealed to be a finite sub-poset of the
coframe of subtoposes of a classifying topos, and the passage from one
equivalence to another is localization.

The evidence comes from three directions.  The \emph{geometric}
direction constructs explicit Grothendieck topologies
$J_{\mathrm{trace}}$, $J_{\mathrm{sim}}$, and $J_{\mathrm{bisim}}$
on the site of finitely presentable labeled transition systems, with a
strict chain
$J_{\mathrm{bisim}} \subsetneq J_{\mathrm{sim}} \subsetneq
J_{\mathrm{trace}}$ extended to all 13~named equivalences by the
Energy--LT Bridge (Theorem~\ref{thm:energy-lt-bridge}).  The trace
and bisimulation topologies are explicitly defined via
observation classes; the simulation topology is obtained via
Caramello's quotient-theory duality~\cite{caramello2017theories},
and a constructive counterexample
(Theorem~\ref{thm:naive-sim-instability}) proves this indirect route
structurally necessary; a truncation-level mismatch between the
propositional content of simulation and the set-level data carried
by sieves.  The \emph{algebraic} direction closes the 13~named
equivalences under lattice operations to obtain the 30-element
spectrum lattice~$L_{30}$ (notation introduced here for the lattice
closure of the van~Glabbeek spectrum under meets and joins) with its
bi-Heyting structure, including the identity
$S \to F = \mathrm{IF}$, 17~unnamed hybrids, and direct
indecomposability, all proved with zero custom axioms.  The
\emph{logical} direction identifies the boundary between
topos-level and process-level expressiveness: the Geometric
van~Benthem Theorem characterizes diamond-only Hennessy--Milner logic
as the bisimulation-invariant fragment of geometric logic, while the
Geometric Closure Theorem reduces presheaf Heyting implications
to testing at a single free extension, a reduction specific to the
combinatorial structure of the LTS site, where the canonical
one-step extension serves as a generic witness in place of
Kripke--Joyal's universal quantification over all extensions.

The conceptual payoff is that the classifying topos provides a
\emph{uniform} framework for computational symmetries.  Each
behavioral equivalence is a Grothendieck topology; the space of all
such equivalences is a coframe; and the algebraic operations of this
coframe, meets, Heyting implications, co-Heyting
subtractions, produce structure that process algebra cannot access
internally.  The 17~unnamed elements of~$L_{30}$ are the most
concrete manifestation: they are process equivalences discovered
through the coframe structure, invisible to any single energy budget,
corresponding to geometric axioms (determinism, image-finiteness,
mixed constraints) that cross-cut the classical spectrum.  The
embedding $\nu \colon \mathcal{V} \hookrightarrow
\mathrm{Sub}(\mathrm{Set}[\mathbb{T}_{\mathrm{LTS}}])^{\mathrm{op}}$
is an order-embedding of a 13-element poset into an infinite coframe;
whether it extends to a \emph{sublattice} embedding of~$L_{30}$, with
meets and joins agreeing with the ambient coframe operations, remains
the central structural conjecture.

This framework is not specific to the van~Glabbeek spectrum.  Any
family of geometric theories whose models are computational systems,
and any family of quotient axioms corresponding to observational
abstractions, produces a lattice of subtoposes in which the same
questions can be asked: what are the lattice operations?  What unnamed
equivalences arise?  Is the embedding a sublattice?  The present paper
answers these questions for labeled transition systems and the
energy-game hierarchy; the techniques, Lindenbaum algebras, nucleus
maps, observation-class topologies, Caramello's
duality~\cite{caramello2017theories}, are general.
Extensions to other computational models
(automata, event structures, Petri nets) whose equivalences are
expressible as geometric quotients are natural next steps; models
involving quantitative reasoning (probabilistic or quantum systems)
would require enrichment of the logical framework beyond geometric
sequents.

The simulation topology is established abstractly via Caramello's
duality; the recent generation formulas of
Caramello--Lafforgue~\cite{caramello2025generation}
offer a path to explicit covering sieves, which would eliminate
the axioms for~$J_{\mathrm{sim}}$.  Whether the full 30-element
closure embeds as a \emph{sublattice} of the subtopos coframe, meets
and joins agreeing, remains the central structural question; a
positive answer would give~$L_{30}$ a coordinatization-free
characterization, resolving the dependence on Bisping's energy-game
parameterization noted in Remark~\ref{rem:coord-dependence}.  More
broadly, whether every Grothendieck topology between
$J_{\mathrm{bisim}}$ and $J_{\mathrm{trace}}$ arises from an
observation class would yield a bijective correspondence between
behavioral equivalences and subtoposes, a structural
embedding of the two traditions this paper connects.

\paragraph{Formalization.}
A Lean~4 formalization using Mathlib~\cite{mathlib2020} accompanies
this paper, comprising 118~files in the \texttt{RuleSys} library.
The headline results, HML bisimulation invariance
(Theorem~\ref{thm:hml-invariant}), the Geometric Closure
Theorem~\ref{thm:geometric-closure}, the $L_{30}$ bi-Heyting
structure
(Theorems~\ref{thm:heyting-implication}--\ref{thm:coheyting-decomposition}),
both Grothendieck topologies $J_{\mathrm{trace}}$ and
$J_{\mathrm{bisim}}$ (Theorem~\ref{thm:spectrum-bracket}), the
naive simulation instability (Theorem~\ref{thm:naive-sim-instability}),
characteristic formulas (Lemma~\ref{lem:char-formula}), and the
depth~0--2 van~Benthem instances
(Theorems~\ref{thm:depth0}--\ref{thm:depth2}), are proved with
zero custom axioms.
The axiomatized components mark the interface with published
results, Caramello's quotient-theory
duality~\cite{caramello2017theories}, van~Glabbeek's logical
characterization theorems~\cite{vanglabbeek2001}, and Bisping's
energy-game framework~\cite{bisping2023cav}, whose Lean
formalization awaits the development of Mathlib's topos-theoretic
infrastructure.  The full codebase is available at \texttt{K-NANOG/spectrum-topos}.

\end{document}